\documentclass[a4paper, 12pt]{article}
\usepackage[left=2.4cm,right=2.4cm,top=3cm]{geometry}

\usepackage[numbers,square]{natbib}
\usepackage{amsmath}
\usepackage{amsfonts}
\usepackage{amssymb}
\usepackage{amsthm} 
\usepackage{multicol}
\usepackage{graphicx}
\usepackage{enumerate}
\usepackage[inline,shortlabels]{enumitem}
\usepackage{bbm}
\usepackage[usenames,dvipsnames]{color}
\usepackage[colorlinks=true,linkcolor=blue, citecolor=blue, urlcolor=blue]{hyperref}
\usepackage{fancyhdr}
\usepackage{epsfig}
\usepackage{epstopdf}
\usepackage{caption}
\usepackage[textfont=footnotesize,justification=centering]{subcaption}
\usepackage{multirow}
\usepackage{appendix}
\usepackage{tabularx}
\usepackage{threeparttable}
\usepackage{float}
\usepackage{color}
\usepackage{xcolor}
\usepackage[most]{tcolorbox}

\usepackage{titling}
\usepackage{sectsty}
\subsubsectionfont{\normalfont\itshape}

\newtheorem*{theorem-non}{Theorem}
\newtheorem{theorem}{Theorem}
\newtheorem{definition}{Definition}
\newtheorem{proposition}{Proposition}
\newtheorem{lemma}{Lemma}
\newtheorem{corollary}{Corollary}
\newtheorem{remark}{Remark}

\providecommand{\keywords}[1]
{
  	
  \textbf{\textit{Keywords ---}} #1
}

\usepackage{soul}
\usepackage{framed}
\usepackage{color}

\soulregister\ref7

\usepackage[UTF8,scheme=plain]{ctex}

\def \E{\mathbb{E}}
\def \P{\mathbb{P}}
 \def \Q{\mathbb{Q}}

\usepackage[utf8]{inputenc}

\title{Predictable Relative Forward Performance Processes: Multi-Agent and Mean Field Games for \\Portfolio Management}

\author{
Gechun Liang\thanks{Department of Statistics, University of Warwick. Email: \href{mailto:g.liang@warwick.ac.uk}{g.liang@warwick.ac.uk}} 
\and Moris S. Strub\thanks{Warwick Business School, University of Warwick. Email: \href{mailto:moris.strub@wbs.ac.uk}{moris.strub@wbs.ac.uk}} 
\and Yuwei Wang\thanks{
Corresponding author. School of Mathematics, Shanghai University of Finance and Economics. 
Email: \href{mailto:yuweiwang@cuhk.edu.hk}{wangyuwei@mail.shufe.edu.cn}}
}

\date{}

\begin{document}

\maketitle

\begin{abstract}
We introduce predictable relative forward performance processes (PRFPP) as a new framework for studying portfolio management within a competitive and incomplete market environment. 
Each agent trades a distinct stock following a binomial distribution with probabilities for a positive return depending on the market regime characterized by a non-traded stochastic factor.
For both the finite population and mean field games, we construct and analyse PRFPPs for initial data of the CARA class along with the associated equilibrium strategies. 
We find that relative performance concerns do not necessarily lead to more investment in the risky asset compared to when there are no such concerns. Under some parameter constellations, agents short a stock with positive expected excess return.
The binomial market setting facilitates a straightforward adjustment of risky asset skewness, enabling an analysis of its impact on investment behavior—an aspect that continuous-time frameworks cannot capture.
\end{abstract}  


\section{Introduction}\label{sec:Introduction}
The appraisal of relative investment performance has long been an important criterion in asset management. Competition with strong incentives to outperform each other, as was suggested in \cite{brown1996tournaments}, \cite{chevalier1997risk}, \cite{croson2005managerial}, and \cite{evans2020competition} among others, can be the consequence of higher agency fee and executive pay, career concerns, the intention to attract and retain clients, or performance-based compensation schemes. 
The literature on relative performance criteria can be traced back to \cite{holmstrom1982moral}, where the notion of relative performance evaluation was first put forward as a criterion for evaluating managerial performance in a multi-agent environment. 

Our objective is to study relative performance concerns in the framework of discrete-time predictable forward performance processes \cite{angoshtari2020predictable}.
We propose a new notion of \textit{predictable relative forward performance processes} (PRFPPs)
and study these processes in games with finite and infinite populations. 
Our work is closely related to the recent work of \cite{anthropelos2022competition}, who introduced the notion of forward relative performance criteria and studied optimal investment under competition between two agents in an Itô-diffusion setting. Further extensions along this line of research can be found in  \cite{dos2021forward}, and \cite{dos2022forward}.  

Recall that a classical, backward portfolio selection problem requires three ingredients: an investment horizon, a mathematical model for the market, and a performance criterion applied at the horizon, usually in the form of a utility function. 
In this setting, related work on relative performance concerns in continuous-time are \cite{FreiDosReis11:MFE}, \cite{basak2014strategic}, \cite{basak2015competition}, \cite{espinosa2015optimal}, and \cite{lacker2019mean}.
In contrast to the backward setting, where the investment horizon, market model, and terminal utility function are specified ex ante, forward processes are constructed from an initial datum that represents current preferences rather than preferences. Preferences are then updated over time under the guidance of the martingale optimality principle. This flexibility allows preferences and market models to update as new information arrives. Moreover, in discrete time, forward performance processes are predictable and provide a transparent and intuitive link between preferences at consecutive time points. They ensure time-consistent decision making as market parameters, preferences, and investment strategies evolve jointly over time.
There is a rich literature on forward performance processes, initially introduced by \cite{musiela2006investments, musiela2008optimal, musiela2009portfolio,musiela2010portfolio}, whose work serves as a seminal cornerstone for this notion. Relevant studies include \cite{henderson2007horizon} and \cite{choulli2007minimal}, which explore various forms of these processes. 
The subsequent works primarily focused on probabilistic representations and characterizations for forward performance processes, as in \cite{musiela2010stochastic}, \cite{el2014, el2018consistent,el2021}, \cite{mohamed2013exact, el2021recover}, \cite{nadtochiy2017optimal} and \cite{Avanesyan2020}, among many others. 
Forward performance processes have been extended to various domains.
We refer to \cite{he2021forward} for situations involving probability distortion, to \cite{kallblad2018dynamically} and \cite{chong2024optimal} for investment under model uncertainty, to \cite{källblad2020black} for optimal consumption, to \cite{shkolnikov2016asymptotic}, \cite{geng2017temporal} and \cite{chong2019ergodic} for an asymptotic analysis of forward performance processes, to \cite{choulli2017explicit}, \cite{nadtochiy2014class} and \cite{liang2017representation} for homothetic forward performance processes, to \cite{hu2020systems} for regime-switching models, and to \cite{Bo2022} for general semmimartinale models.
A connection between the forward
performance process and optimal portfolios has been explored in \cite{mohamed2013exact}. \cite{vzitkovic2009dual} provides a convex dual characterization of forward performance processes for CARA preferences, extended by \cite{anthropelos2014forward} under contingent claims.

More recent is the notion of the discrete-time, \textit{predictable} version of forward processes introduced in \cite{angoshtari2020predictable} and further studied in \cite{strub2021evolution}, \cite{liang2023predictable}, \cite{angoshtari2022predictable}, \cite{Angoshtari2024} and \cite{waldon2024forward}. 
Our goal is to examine the effect of competition on investment policies for this framework. 
To that end, we consider an economy featuring imperfectly correlated risky assets where each agent specializes in a distinct stock. This scenario is commonly known as \textit{asset specialization} and aligns with substantial empirical evidence indicating that investors tend to engage primarily in familiar investment opportunities, cf. \cite{merton1987simple}, \cite{coval1999home}, \cite{kacperczyk2005industry}, or \cite{ferreira2013determinants}.

We develop the notion of a forward Nash equilibrium for finitely many competitors and a mean field equilibrium for infinitely many competitors in a discrete-time setting.
This concept builds on \cite{anthropelos2022competition}, who study competition between two agents under relative forward preferences in a continuous-time setting.
The interaction among agents in our analysis involves not only the controls employed by the agents, but also the wealth and forward performance processes of agents. 
Our framework incorporates a non-traded stochastic factor, which is mathematically modeled as a Bernoulli random variable representing two market regimes. 
This noise component plays an important role within our analysis, differentiating between a bull and bear market regime, inducing correlation into the stocks traded by different agents, and introducing an extra layer of market incompleteness to the predictable forward framework. 

We solve the $N$-agent game via a two-step procedure: In the initial step, we determine the so-called ``best response" with respect to the investment policies employed by other agents. 
Subsequently, we proceed to search for a fixed point of the whole system, such that the candidate investment strategy is indeed a forward Nash equilibrium. 
For the purpose of tractability, we work with initial data belonging to the family of exponential utility functions throughout this paper.
Under this assumption, we establish the existence of a forward Nash equilibrium. 
Furthermore, we provide an explicit expression for the forward equilibrium for two special cases: The case where the population is homogeneous and the case where there are two heterogeneous agents. These results offer valuable insights into the characteristics of investment behaviour under relative forward performance concerns. 
For instance, the risky investment of a given agent having long position increases as the overall group becomes less risk-averse in a homogeneous multi-agent game, or as the quality of the stock held by her competitor improves in the heterogeneous 2-agent game. 
A further insight of our paper is that it can be optimal for an agent to short sell a stock market with positive expected excess return when there are relative performance concerns.
This at first counter-intuitive behavior occurs when the stocks traded by two agents are correlated, but one exhibits a positive and the other a negative expected excess return.

For the mean field game framework, we follow the standard paradigm of considering a representative agent with stochastic type vector.
As a key technical contribution, we derive a representation of the average population wealth as a conditional expectation of the representative agent's wealth under any admissible strategy.
We then construct a PRFPP and forward mean field equilibrium for the representative agent, and establish its uniqueness within the exponential family.
The key step is to consider a fixed point problem for which we prove uniqueness and existence of the solution.
We find that the equilibrium investment strategy can be decomposed into a classical Merton component and an additional component driven by relative performance concerns. The latter increases with the competition weight and vanishes when competition is absent. Moreover, the competition-induced component is always positive, implying that relative performance concerns encourage higher risky investment among agents with CARA preferences, in line with the findings of \cite{lacker2019mean}.
We furthermore establish a connection between the multi-agent game and the mean field game. Our results show that an agent with relative performance concerns decreases her stock holdings when competitors are more risk-averse, less competitive, or when the bear market is more likely to occur. In our framework under the assumption that the market offers  positive expected excess returns, and consistent with the existing literature, competition serves as a motivator that compels agents to augment their investments in the risky asset.

Additionally, we investigate two specific cases of important market setups that yield more tractable and readily interpretable solutions. In the first scenario, all decision makers trade the same risky asset. 
In the second scenario, we investigate a situation where non-traded stochastic factor, which influences the return distribution among stocks, is absent.
Lastly, we conduct a comprehensive numerical analysis to investigate the impact of investment characteristics, market features and asset correlations on investment policies, respectively. Consistent with \cite{lacker2019mean}, the equilibrium policy increases as other agents become less risk averse or more competitive on average, or as the average quality of other stocks improves. The effect of market skewness, however, is specific to our setting and has no counterpart in previous literature.

Prior research on relative forward performance processes focused on continuous-time, Itô-diffusion settings, with constant market parameters. 
Herein, we allow for dynamic updating of random parameters, which enables the model to be adjusted on a period-by-period basis as the market evolves. 
Considering performance evaluation and trading at discrete times is of importance for investment practice. 
The predictability inherent in the discrete-time formulation facilitates a more intuitive evolution of the utility functions at two consecutive time points. 

Our main contribution is an iterative scheme to construct the PRFPPs in a competitive environment.
We then study the implications of competition on portfolio selection both analytically and numerically and observe that the equilibrium strategy consistently shifts in the same direction when a specific idiosyncratic parameter or its corresponding population average increases.
Another advantage of our setup is that one can readily manipulate return skewness of a binomial stock and analyze its effect on optimal strategies and preferences. 
We find that investors exhibit aversion to both negative and positive skewness.

Mean-field games have been introduced by \cite{huang2006large} and \cite{lasry2007mean} to model interactions among a large number of agents via their empirical state distribution. 
To date, the majority of the existing literature focused on continuous-time mean field games with continuous state spaces, where the optimal solution to the mean field problem can be characterized by Hamilton–Jacobi-Bellman equations coupled with transport equations. 
In contrast, there exists relatively few literature in the realm of discrete-time mean field games on discrete state space, for example,
\cite{gomes2010discrete}, \cite{huang2012mean}, \cite{adlakha2015equilibria}, \cite{doncel2019discrete}, 
\cite{guo2019learning}, \cite{campi2022correlated}, \cite{guo2022mf} and \cite{bonesini2024correlated}.
This paper contributes to this emerging literature on discrete-time, discrete state space mean field games and seems to be the first application of discrete mean field game theory to portfolio optimization. 

The organization of the paper is as follows. Section \ref{sec:ModelSetup} presents the market model and introduces the definition of discrete-time PRFPP for both the finite agent and mean field game. We construct finite agent PRFPPs and characterize the forward Nash equilibrium in Section \ref{sec:NAgentGame} and the mean field counterparts in Section \ref{sec:MFG}. 
In the Appendix, we present a discussion of two extreme cases: The case where the stocks traded by all agents are identical almost surely and the case where there is no non-traded stochastic factor.
For the mean field game, we investigate the dependence of investment policy on individual, network, and market parameters by examining some numerical examples in Section \ref{sec:NumericalAnalysis}. 
Section \ref{sec:Conclusion} concludes. 
All the proofs are delegated to the Appendix.  

\section{Model and definition}\label{sec:ModelSetup}
We introduce and formalize first the game of $N$ agents and then the  mean field game as the number of agents goes to infinity.
Our model builds on \cite{anthropelos2022competition} and \cite{dos2021forward,dos2022forward}, who study relative forward performance criteria in continuous-time with stock price driven by a geometric Brownian motion. 
Different from these existing studies, we herein consider a model of \textit{discrete-time} predictable relative forward performance criteria.
A further distinguishing feature of this model is that it allows for dynamic updating of the market parameters.

Throughout the paper, $\mathbb{N}$ denotes the set of positive integers and $\mathbb{N}_0$ the set of nonnegative integers. 
Trading and performance evaluation times coincide and are denoted by $t \in \mathbb{N}_0$. 
We refer to \cite{liang2023predictable} for the case where trading occurs at a higher frequency than performance evaluation in a model of discrete-time predictable performance criteria without competition. 

\subsection{The finite agent game}
Let $(\Omega, \mathcal{F}, \P)$ be a probability space.
In the finite agent game, there is a population of $N$ agents trading a common risk-free bond and a distinct stock.
The common risk-free bond does not offer any interest.
Let $R_{(i),t}$ denote the total return over period $[t-1, t)$ of the stock of agent $i$. 
We suppose that
$R_{(i),t}$ is described by a binomial model
\begin{align*}
    R_{(i),t} = u_{(i),t} B_{(i),t} + d_{(i),t} (1-B_{(i),t}), \quad t \in \mathbb{N},
\end{align*}
where $B_{(i),t}$ takes values in $\{0,1\}$ for all $t \in \mathbb{N}$, i.e., $(B_{(i),t})_{t \in \mathbb{N}}$ is a sequence of Bernoulli random variables with transition probabilities $(p_{(i),t})_{t \in \mathbb{N}}$.  
The market parameters $(d_{(i),t})_{t \in \mathbb{N}}$, $(u_{(i),t})_{t \in \mathbb{N}}$, and $(p_{(i),t})_{t \in \mathbb{N}}$ are allowed to be stochastic processes.
In order to assure absence of arbitrage, we assume that $0 < d_{(i),t}  < 1 < u_{(i),t}$ and that $0 < p_{(i),t} < 1$ for any $t \in \mathbb{N}$ and $i \in \{1, \dots, N\}$.
To explicitly model the correlation structure among different agents, we assume that the probability space $(\Omega,\mathcal{F},\P)$ supports another stochastic processes $(\xi_t^{cn})_{t\in \mathbb{N}}$ which is not tradeable by any agent. 
Each $\xi_t^{cn}$ represents a non-traded stochastic factor for the trading period $[t-1,t)$ and follows a Bernoulli distribution with transition probability $p_t^{cn}$, $(p_t^{cn})_{t\in \mathbb{N}}$ are allowed to be stochastic. 
The event $\{ \xi_t^{cn} = 1 \}$ represents a bull market regime and $\{ \xi_t^{cn} = 0 \}$ represents a bear market regime for the period $[t-1,t)$. 

We define the filtration $\mathbb{F} = (\mathcal{F}_t)_{t \in \mathbb{N}_0}$ by letting $\mathcal{F}_t$ be the augmented $\sigma$-algebra generated by $(\xi_j^{cn})_{j=1}^t$,  $(p_j^{cn})_{j=1}^{t+1}$, $(B_{(i),j})_{j=1}^t$, and the market parameters $(u_{(i),j}, d_{(i),j}, p_{(i),j})_{j=1}^{t+1}$ of all the agents.
We in particular have $\P[B_{(i),t}=1\vert \mathcal{F}_{t-1}]=p_{(i),t}$ and $\P[\xi_t^{cn}=1\vert \mathcal{F}_{t-1}]=p^{cn}_t$.
Note that $u_{(i),t}$, $d_{(i),t}$, $p_{(i),t}$ and $p^{cn}_{t}$ are predictable, i.e., all market parameters are known at the beginning of each trading period $[t-1, t)$. 
However, the outcome of the non-traded stochastic factor and stock returns are only known at the end of each trading period. 

{Let $\sigma(\xi_t^{cn})$ denote the $\sigma$-algebra generated by the random variable $\xi_t^{cn}$.
We introduce the conditional probabilities for an upward move of the $i$th stock in the bull or bear market, defined by $p_{(i),t}^1=\frac{\P[\xi_t^{cn}=1, B_{(i),t}=1\vert \mathcal{F}_{t-1}]}{\P[\xi_t^{cn}=1\vert \mathcal{F}_{t-1}]}$ and $p_{(i),t}^0=\frac{\P[\xi_t^{cn}=0, B_{(i),t}=1\vert \mathcal{F}_{t-1}]}{\P[\xi_t^{cn}=0\vert \mathcal{F}_{t-1}]}$. 
One can easily deduce that 
$$\P[B_{(i),t}=1\big\vert \mathcal{F}_{t-1}\vee \sigma(\xi_t^{cn})]=p_{(i),t}^1\mathbbm{1}_{\{\xi_t^{cn}=1\}}+p_{(i),t}^0\mathbbm{1}_{\{\xi_t^{cn}=0\}}.$$
We suppose that $p_{(i),t}^1>p_{(i),t}^0$ almost surely for $i = 1, \dots,N$, that is, all stocks are more likely to perform well in a bull market than in a bear market regime. 
While probabilities of an upward move of stocks depend on the non-traded stochastic factor, price levels are assumed to be independent thereof. 
We further assume that $B_{(i),t}$ and $B_{(\ell),t}$ are conditionally independent given $\mathcal{F}_{t-1}\vee\sigma(\xi^{cn}_t)$ for any $i \neq \ell$.
That is, the dependence structure between two stocks traded by different agents is completely characterized by the non-traded stochastic factor. 

We remark that the financial market is incomplete due to two factors. First, as discussed in \cite{angoshtari2022predictable} and \cite{liang2023predictable}, while the binomial model for the financial market is sequentially complete under the assumption of full knowledge of model parameters at the beginning of each period [t−1,t), it may lack completeness across multiple periods if the parameters for later periods are determined in part through some exogenous random noise.
Second, incompleteness also stems from the presence of the non-traded stochastic factor \((\xi_t^{cn})_{t\in \mathbb{N}}\), which affects the transition probabilities of stock returns. Thus, the model presented here extends beyond the classical binomial framework.

Trading strategies of the $i$th agent are required to be predictable 
processes
$\pi_{(i)} = \left( \pi_{(i),t} \right)_{t \in \mathbb{N}}$, where $\pi_{(i),t}$ denotes the dollar amount invested by the $i$th agent in the risky asset over trading period $[t-1, t)$. 
A portfolio is constructed by following the trading strategy on the stock while investing all the remaining wealth in the risk-free bond. 
Given a self-financing trading strategy $\pi_{(i)}$, the wealth process $X_{(i)}^{\pi} = \left(X_{(i),t}^{\pi}  \right)_{t \in \mathbb{N}}$ evolves according to $X_{(i),t}^{\pi} = x_{(i)} + \sum\nolimits_{j=1}^t\pi_{(i),j}(R_{(i),j}-1)$ with initial endowment $X_{(i),0}=x_{(i)} \in \mathbb{R}$. 
All that is required for a trading strategy $\pi_{(i)}$ and the associated wealth process $X_{(i)}^{\pi}$ to be  \textit{admissible} is that each term $\pi_{(i),t}$ is real-valued and $\mathcal{F}_{t-1}$-measurable. 
In particular, we do not impose any trading constraints and allow for negative wealth levels.
We remark that trading strategies are allowed to depend on the wealth of all agents, including their initial wealth levels. 
We usually do not make this dependence explicit and simply write $\pi_{(i),t}$ instead of $\pi_{(i),t} (X_{(1),t-1}, \dots, X_{(N),t-1})$.
The set of admissible trading strategies starting at time $t$ is denoted by $\mathcal{A}_{(i)}(t)$, i.e., 
\begin{align*}
    \mathcal{A}_{(i)}(t) = \left\{\left( \pi_{(i),k} \right)_{k > t}: \pi_{(i),k} \in \mathbb{R}\ {\rm{is}}\  \mathcal{F}_{k-1}{\rm{-measurable}}\right\}
\end{align*}
\noindent
and the associated wealth processes starting from $X_{(i),t}^{{\pi}} = x_{(i)}$ by $\mathcal{X}_{(i)}(t,x_{(i)})$.  
We abbreviate $\mathcal{A}_{(i)}(0)$ and $\mathcal{X}_{(i)}(0,x_{(i)})$ by $\mathcal{A}_{(i)}$ and $\mathcal{X}_{(i)}(x_{(i)})$, respectively. 
To simplify notation, we often drop the explicit dependence of a wealth process on the trading strategy and write $X_{(i)} \in \mathcal{X}_{(i)} (t,x_{(i)})$.

A key feature of our model is that each agent is interested in the performance of her \textit{relative wealth} compared to the wealth levels of all agents.
The relative wealth is denoted by $\widetilde{X}_{(i)} = \left(\widetilde{X}_{(i),t}  \right)_{t \in \mathbb{N}}$ and is given by
\begin{align*}
    \widetilde{X}_{(i),t}=X_{(i),t}-\frac{\theta_{(i)}}{N}\sum_{k=1}^NX_{(k),t}=(1-\frac{\theta_{(i)}}{N})X_{(i),t}-\frac{\theta_{(i)}}{N}\sum_{k\neq i}X_{(k),t} .
\end{align*}
The parameter $\theta_{(i)}\in[0,1]$ is referred to as the \textit{competition weight} of the $i$th agent.
This parameter represents her concern for relative performance.
The dynamics of the relative wealth depend on the strategies of other agents and is given by
\begin{align}\label{RelativeWealthFormula}
    \widetilde{X}_{(i),t}=\widetilde{X}_{(i),0} + \left(1-\frac{\theta_{(i)}}{N}\right) \sum\limits_{n=1}^t\pi_{(i),n}(R_{(i),n}-1)-\frac{\theta_{(i)}}{N}\sum_{k\neq i}\sum\limits_{n=1}^t\pi_{(k),n}(R_{(k),n}-1). 
\end{align}
We denote the set of relative wealth processes resulting from admissible trading strategies and starting at $t \in \mathbb{N}_0$ from relative wealth $\widetilde{X}_{(i),t} = \widetilde{x}_{(i)}$ by $\widetilde{\mathcal{X}}_{(i)}(t,\widetilde{x}_{(i)})$  and abbreviate $\widetilde{\mathcal{X}}_{(i)}(0,\widetilde{x}_{(i)})$ by $\widetilde{\mathcal{X}}_{(i)}(\widetilde{x}_{(i)})$.

\medskip

The $i$th agent starts at time $t_0=0$ with preferences over her relative wealth represented by a utility function $U_{(i),0}$.
We herein assume that any utility function $U: \mathbb{R} \rightarrow \mathbb{R}$ is twice continuously differentiable, strictly increasing, strictly concave, and satisfies the Inada conditions. 
While we will provide a definition for general utility functions, our main results will make the additional assumption that the initial data belong to the family of CARA utility functions. 

For each agent $i \in \{1,2,\dots,N\}$, we first define a discrete-time predictable relative forward performance process under a best-response principle. The policies of the other agents are exogenously given and progressively revealed over time. Agent $i$ is assumed to compete passively: she observes the opponents' admissible policy processes and incorporates them into her optimization problem, but does not interact strategically with them.

\begin{definition}\label{def:RelativeFrowardPreferences-NQ-BR}
Let policy $\pi_{(k)} \in \mathcal{A}_{(k)}$ be given, $k=1,2,\dots N, k\neq i$. A family of random functions $\left\lbrace U_{(i),t}: \mathbb{R} \times \Omega \rightarrow \mathbb{R} \vert i=1,2,\dots,N, t \in \mathbb{N}_0 \right\rbrace$ is called a discrete-time predictable relative forward performance process (PRFPP) for agent $i$ under the best-response strategy
if the following conditions hold:
\begin{enumerate}
\item[(i)] $U_{(i),0}\left(x,\cdot\right)$ is constant and $U_{(i),t}\left(x,\cdot \right)$ is $\mathcal{F}_{t-1}$-measurable for each $x\in\mathbb{R}$, and $t \in \mathbb{N}$. 
\item[(ii)] $U_{(i),t} (\cdot ,\omega)$ is a utility function for almost all $\omega \in \Omega$, and all $t \in \mathbb{N}_0$.
\item[(iii)] For any initial wealth profile $(x_1, \dots, x_N)$ and any $\pi_{(i)} \in \mathcal{A}_{(i)}$, the resulting relative wealth process $\widetilde{X}_{(i)} \in \widetilde{\mathcal{X}}_{(i)}(\widetilde{x}_{(i)})$ satisfies
\begin{align*}
U_{(i),t-1} \left( \widetilde{X}_{(i),t-1} \right) \geqslant \E \left[ U_{(i),t} \left( \widetilde{X}_{(i),t} \right) \big\vert \mathcal{F}_{t-1} \right], \quad t \in \mathbb{N}.
\end{align*}
\item[(iv)]  For any initial wealth profile $(x_1, \dots, x_N)$, there exists $\pi_{(i)}^{*} \in \mathcal{A}_{(i)}$
, such that the resulting relative wealth process $\widetilde{X}_{(i)}^{*}$ satisfies
\begin{align*}
U_{(i),t-1}  \left( \widetilde{X}^{*}_{(i),t-1} \right) = \E \left[ U_{(i),t} \left( \widetilde{X}^{*}_{(i),t} \right) \big\vert \mathcal{F}_{t-1} \right], \quad t \in \mathbb{N}.
\end{align*}
\end{enumerate}
\end{definition}

However, agents in our model dynamically compete with each other when deriving their optimal portfolios. 
Conceptually, we thus consider a pure strategy Nash game and, following \cite{anthropelos2022competition}, seek to determine the definition of discrete-time PRFPPs under a \textit{forward Nash equilibrium}.

\begin{definition}\label{def:RelativeFrowardPreferences-NQ}
A family of random functions $\left\lbrace U_{(i),t}: \mathbb{R} \times \Omega \rightarrow \mathbb{R} \vert i=1,2,\dots,N, t \in \mathbb{N}_0 \right\rbrace$ together with a policy vector $(\pi^{*}_{(1)}, \pi^{*}_{(2)}, \dots, \pi^{*}_{(N)})$, $\pi^{*}_{(i)} \in \mathcal{A}_{(i)}$, $i=1,2,\dots,N$, is called a discrete-time predictable relative forward performance process (PRFPP) under a forward Nash equilibrium if the following conditions hold:
\begin{enumerate}
\item[(i)] $U_{(i),0}\left(x,\cdot\right)$ is constant and $U_{(i),t}\left(x,\cdot \right)$ is $\mathcal{F}_{t-1}$-measurable for each $x\in\mathbb{R}$, $i=1,2,\dots,N$, and $t \in \mathbb{N}$. 
\item[(ii)] $U_{(i),t} (\cdot ,\omega)$ is a utility function for almost all $\omega \in \Omega$, $i=1,2,\dots,N$, and all $t \in \mathbb{N}_0$.
\item[(iii)] For any initial wealth profile $(x_1, \dots, x_N)$, any $i \in \{1,2,\dots,N \}$, and any $\pi_{(i)} \in \mathcal{A}_{(i)}$, under the policy vector $(\pi^{*}_{(1)},\dots, \pi_{(i)}, \dots, \pi^{*}_{(N)})$, the relative wealth process $\widetilde{X}_{(i)} \in \widetilde{\mathcal{X}}_{(i)}(\widetilde{x}_{(i)})$ satisfies
\begin{align*}
U_{t-1}^{(i)} \left( \widetilde{X}_{(i),t-1} \right) \geqslant \E \left[ U_{(i),t} \left( \widetilde{X}_{(i),t} \right) \big\vert \mathcal{F}_{t-1} \right], \quad t \in \mathbb{N}.
\end{align*}
\item[(iv)]  For any initial wealth profile $(x_1, \dots, x_N)$ and any $i \in \{1,2,\dots,N \}$, under the policy vector $(\pi^{*}_{(1)}, \dots, \pi^{*}_{(i)}, \dots, \pi^{*}_{(N)})$, the relative wealth process $\widetilde{X}^{*}_{(i)}$ satisfies
\begin{align*}
U_{(i),t-1}  \left( \widetilde{X}^{*}_{(i),t-1} \right) = \E \left[ U_{(i),t} \left( \widetilde{X}^{*}_{(i),t} \right) \big\vert \mathcal{F}_{t-1} \right], \quad t \in \mathbb{N}.
\end{align*}
\end{enumerate}
\end{definition}

Given wealth $x$ at time $t$, we introduce the set of admissible strategies for single period $[t,t+1)$ by
$\mathcal{A}^{(i)}_{t,t+1}(x)$, the corresponding set of wealth and relative wealth at time $t+1$ by $\mathcal{X}^{(i)}_{t,t+1}(x)$ and $\widetilde{\mathcal{X}}^{(i)}_{t,t+1}(x)$.
Properties $(iii)$ and $(iv)$ in Definition \ref{def:RelativeFrowardPreferences-NQ} imply that
\begin{align}\label{eq:InverseEqn_Dynamic}
U_{(i),t-1}  \left( \widetilde{X}^{*}_{(i),t-1} \right) = \sup_{\widetilde{X}_{t} \in \widetilde{\mathcal{X}}_{t-1,t} \left( \widetilde{X}^{*}_{(i),t-1}\right)} \E \left[ U_{(i),t}  \left(\widetilde{X}_{(i),t}\right) \bigg\vert \mathcal{F}_{t-1}\right].
\end{align}
Iteratively solving \eqref{eq:InverseEqn_Dynamic} leads to the construction of PRFPPs under a forward Nash
equilibrium, see \cite{angoshtari2020predictable} for a detailed exposition of the case where $N=1$. 
The crucial step is to solve the following \textit{$N$-coupled inverse investment problems}:
Given a profile of initial preferences $U_{(i),0}$, $i = 1,\dots, N$, we seek to determine relative forward utility functions $U_{(i),1}$, $i=1,\dots,N$, and a strategy profile $\left( \pi^{*}_{(1)}, \dots, \pi^{*}_{(N)}\right)$ such that
\begin{equation}\label{eq:InverseProblem}
\begin{aligned}
    U_{(i),0} \left( \widetilde{x}_{(i)} \right) & = \sup_{\widetilde{X}_{(i),1} \in \widetilde{\mathcal{X}}_{0,1}^{(i)} \left(\widetilde{x}_{(i)}\right)} \E \left[ U_{(i),1}  \left(\widetilde{X}_{(i),1}\right)\right]\\ 
    &=\sup_{{\pi_{(i),1}} \in \mathcal{A}^{(i)}_{0,1}} \E \left[ U_{(i),1} \left( \widetilde{x}_{(i)} + (1-\frac{\theta_{(i)}}{N})\pi_{(i),1}(R_{(i),1}-1)-\frac{\theta_{(i)}}{N}\sum_{k\neq i}\pi^{*}_{(k),1}(R_{(k),1}-1) \right)\right]
\end{aligned}
\end{equation}
for all initial wealth profiles $(x_1, \dots, x_N) \in \mathbb{R}^N$ and agents $i = 1, \dots, N$.
Subsequently, one can iteratively construct $(U_{(1),2}, \dots, U_{(N),2})$, $(U_{(1),3},\dots,U_{(N),3})$, and so forth.
This is achieved by repeatedly solving 
a problem of the form \eqref{eq:InverseProblem}, conditionally on updated information available at the next time point, and arguing that this solution satisfies the required measurability conditions.

\subsection{The mean field game}
Having introduced and discussed PRFPPs for a finite number of $N$ agents, we now turn towards the mean field game formulation of these concepts by considering the limit as 
the number of agents goes to infinity.
We begin by assuming that the probability space $(\Omega,\mathcal{F},\P)$ yet supports two additional \textit{type vectors} for each agent in the game: The \textit{type vector specifying initial characteristics of agent $i$}, $ i = 1 ,\dots, N$, takes values in the space $Y_1^{TV}:=\mathbb{R}\times (0,+\infty)\times [0,1]$ and is given by
$$\widetilde{\eta}_{(i)}=(\widetilde{x}_{(i)}, \widetilde{\gamma}_{(i)}, \widetilde{\theta}_{(i)}),$$
where $\widetilde{x}_{(i)}$ is the initial wealth, $\widetilde{\gamma}_{(i)}$ parametrizes the initial utility function, and $\widetilde{\theta}_{(i)}$ the competition weight of agent $i$. 
The \textit{type vector specifying the market of agent $i$ in period $[t-1, t)$}, $i = 1,\dots, N$, takes values in the space $Y_2^{TV}:=(1,\infty) \times (0,1) \times (0,1) \times (0,1) \times (0,1)$ and is given by 
$$\widetilde{\zeta}_{(i),t}=\left(\widetilde{u}_{(i),t}, \widetilde{d}_{(i),t}, \widetilde{p}_{(i),t}, \widetilde{p}_{(i),t}^{1}, \widetilde{p}_{(i),t}^{0}\right),$$
where $\widetilde{u}_{(i),t}$ and $\widetilde{d}_{(i),t}$ are the return levels over period $[t-1,t)$ in the binomial market setting with transition probability $\widetilde{p}_{(i),t}$, $\widetilde{p}_{(i),t}^{1}$ and $\widetilde{p}_{(i),t}^{0}$ are the single-period conditional probabilities for an upward of the $i$th stock in the bull or bear market.



For each trading period $[t-1,t)$, we denote for simpler notation the random vector $$\widetilde{\psi}_{(i),t}=\left(\widetilde{\eta}_{(i)}, \widetilde{\zeta}_{(i),1}, \widetilde{\zeta}_{(i),2}, \dots, \widetilde{\zeta}_{(i),t}\right)$$
which concatenates vectors representing all the information of agent $i$ from the initial time up to the current time $t$.
We then consider the empirical random probability measure representing the distribution of type vectors among all thce agents, which is given by
\begin{align*}
    \epsilon_t^N(A)=\frac{1}{N}\sum\limits_{i=1}^N\mathbbm{1}_{A}(\widetilde{\psi}_{(i),t}),
\end{align*}
for Borel set $A\subset Y_1^{TV}
\times (Y_2^{TV})^t$.
Assume that as $N\rightarrow \infty$, the random measure $\epsilon_t^N$ has a weak limit $\epsilon_t$, in the sense that for any bounded continuous functions $f$ on $Y_1^{TV}\times (Y_2^{TV})^t$, $\int_{Y_1^{TV}\times (Y_2^{TV})^t}fd\epsilon_t^N \rightarrow \int_{Y_1^{TV}\times (Y_2^{TV})^t}fd\epsilon_t$. 

Let $\psi_t=\left(\eta, \zeta_1, \zeta_2, \dots, \zeta_t\right)$ be a sequence of random vectors with the limiting distribution $\epsilon_t$ defined on the same probability space $(\Omega,\mathcal{F},\P)$, where
$\eta=(x, \gamma, \theta)$ and $\zeta_k=(u_{k}, d_{k}, p_k, p_k^{1}, p_k^{0}),$ $ k=1,\dots, t$. 
Note that we can always define $\widetilde{\psi}_{(i),t}$ and $\psi_t$ on the same probability space by Skorokhod's representation theorem. 
Let $(B_t)_{t\in \mathbb{N}}$ be a sequence of Bernoulli random variables with transition probabilities $(p_t)_{t\in \mathbb{N}}$ on $(\Omega,\mathcal{F},\P)$.

We follow the standard paradigm in mean field games and  consider the optimization problem of a representative agent with \textit{stochastic} type vectors $(\psi_t)_{t\in \mathbb{N}}$. 
The randomness of the type vector encapsulates the distribution of the agents' characteristics across the economy.
The economy is influenced by a non-traded stochastic factor $\xi_t^{cn}$ defined on $(\Omega,\mathcal{F},\P)$, which follows a Bernoulli distribution with transition probability $p_t^{cn}$.
Denote by $\mathbb{F}^{MF} = (\mathcal{F}_t^{MF})_{t \in \mathbb{N}_0}$ the smallest filtration for which $\eta$ is $\mathcal{F}_0^{MF}$-measurable, $\zeta_t$ and $p_t^{cn}$ are $\mathcal{F}_{t-1}^{MF}$-measurable, and $B_t$ and $\xi_{t}^{cn}$ are $\mathcal{F}_t^{MF}$-measurable. Let $\mathbb{F}^{CN} = (\mathcal{F}_t^{CN})_{t \in \mathbb{N}_0}$ denote the natural filtration generated by $(\xi_{j}^{cn})_{j=1}^{t}$. 

The set of admissible trading strategies of the representative agent is defined by $\mathcal{A}^{MF}(t)$,
\begin{align*}
    \mathcal{A}^{MF}(t)=\left\{\left( \pi_k \right)_{k > t}: \pi_k \in \mathbb{R}\ {\rm{is}}\  \mathcal{F}^{MF}_{k-1}{\rm{-measurable}}\right\},
\end{align*}
we abbreviate $\mathcal{A}^{MF}(0)$ by $\mathcal{A}^{MF}$, and further introduce the set of admissible strategies for single period $[t,t+1)$ by
$\mathcal{A}^{MF}_{t,t+1}(x)$.
For a given trading strategy $(\pi_t)_{t\in \mathbb{N}}\in \mathcal{A}^{MF}$, the wealth process of the representative agent evolves according to $X_t=X_{t-1}+\pi_t(R_t-1)$ with $X_0=x$, where $R_t=u_tB_t+d_t(1-B_t)$. 
As an $\mathcal{F}^{MF}_{t-1}{\rm{-measurable}}$ random variable, we may write admissible strategies of the representative agent as $\pi_t=\pi_t\left(\psi_t, (p^{cn}_{j})_{j=1}^{t}, (\xi_{j}^{cn})_{j=1}^{t-1}, (B_{j})_{j=1}^{t-1}\right)$, a measurable function of the random variables generating the filtration which we again denote by $\pi_t$ with a slight abuse of notation. Denote the associated wealth processes of the representative agent starting from wealth $X^{\pi}_t = x$ by $\mathcal{X}^{MF}(t,x)$, and introduce the corresponding set of wealth
at time $t+1$ by $\mathcal{X}^{MF}_{t,t+1}(x)$.


To establish a relationship between the average wealth of the population and the conditional expected wealth of the representative agent (Proposition \ref{Prop:MFGAverageWealth}), we consider an infinite population of agents, indexed by \( i \in \mathbb{N} \), endowed conditional i.i.d. copies of type vectors of the representative agent.
To this end, we assume that the probability space $(\Omega,\mathcal{F},\P)$ further supports stochastic processes $(\psi_{(i),t})_{t\in \mathbb{N}}$ and $(B_{(i),t})_{t\in \mathbb{N}}$ for $i \in \mathbb{N}$, 
where $\psi_{(i),t}=\left(\eta_{(i)}, \zeta_{(i),1}, \zeta_{(i),2}, \dots, \zeta_{(i),t}\right)$ 
with $\eta_{(i)}=(x_{(i)}$, $\gamma_{(i)}$, $\theta_{(i)})$, and $\zeta_{(i),k}=\left(u_{(i),k}, d_{(i),k}, p_{(i),k}, p_{(i),k}^1, p_{(i),k}^0\right)$, $k=1,\dots,t$, and $B_{(i),t}$ is a Bernoulli random variable with transition probability $p_{(i),t}$, described by a component of $\psi_{(i),t}$. 
The agent with index $i\in\mathbb{N}$ is characterized by $\psi_{(i),t}$ and specializes in a stock with return $R_{(i),t}=u_{(i),t}B_{(i),t}+d_{(i),t}(1-B_{(i),t})$.
We assume that $B_{(i),t}$'s are
conditionally independent of $B_t$ given $\mathcal{F}_{t-1}^{MF}\vee \sigma(\xi_t^{cn})$.

To ensure path dependence of characteristics for a given agent, $\psi_{(i),t}$'s are sampled as follows. 
In the first step, the distribution of $\psi_{(i),1}=(\eta_{(i)}, \zeta_{(i),1})$ follows the identical distribution as $\psi_1=(\eta, \zeta_1)$. 
Given the realization of $\psi_{(i),1}$ at $\bar{\psi}_{(i),1}=(\bar{\eta}_{(i)}, \bar{\zeta}_{(i),1})$, the distribution of $\psi_{(i),2}=(\bar{\eta}_{(i)}, \bar{\zeta}_{(i),1},\zeta_{(i),2})$ is identical to $\psi_2=(\eta, \zeta_1, \zeta_2)$ conditional on $(\eta, \zeta_1)=(\bar{\eta}_{(i)}, \bar{\zeta}_{(i),1})$. As a general step at time $t$, given the realization of $\psi_{(i),t}$ at $\bar{\psi}_{(i),t}=(\bar{\eta}_{(i)}, \bar{\zeta}_{(i),1},\dots, \bar{\zeta}_{(i),t})$, $\psi_{(i),t+1}=(\bar{\eta}_{(i)}, \bar{\zeta}_{(i),1},\dots, \bar{\zeta}_{(i),t}, \zeta_{(i),t+1})$ follows the identical distribution as $\psi_{t+1}=(\eta, \zeta_1, \dots, \zeta_{t+1})$ conditional on $(\eta, \zeta_1,\dots, \zeta_{t})=(\bar{\eta}_{(i)}, \bar{\zeta}_{(i),1},\dots, \bar{\zeta}_{(i),t})$. 
Recall that the admissible strategy $\pi_t$ of the representative agent is $\mathcal{F}_{t-1}^{MF}$-measurable, we may write it as a measurable function of the random variables generating the filtration, i.e., $\pi_t = \pi_t\left(\psi_t, (p^{cn}_{j})_{j=1}^{t}, (\xi_{j}^{cn})_{j=1}^{t-1}, (B_{j})_{j=1}^{t-1}\right)$.
Given an admissible strategy $\pi_t$, the wealth of agent $i$ characterized by type vector $\psi_{(i),t}$ evolves as $X_{(i),t}=X_{(i),t-1}+\pi_{(i),t}(R_{(i),t}-1)$ with $X_{(i),0}=x_{(i)}$, where $\pi_{(i),t}=\pi_t\left(\psi_{(i),t}, (p^{cn}_{j})_{j=1}^{t}, (\xi_{j}^{cn})_{j=1}^{t-1}, (B_{(i),j})_{j=1}^{t-1}\right)$.
For each investment period $[t-1,t)$, 
the effect of the independent noises averages out as $N$ tends to infinity due to the weak law of large numbers, but the effect of the non-traded stochastic factor does not. 
Define ${\Delta}_t^1=p_t^{1}(u_t-1)+(1-p_t^{1})(d_t-1)$, and ${\Delta}_t^0=p_t^{0}(u_t-1)+(1-p_t^{0})(d_t-1)$. Given $t\in \mathbb{N}$, assume that for any $0< t_1, t_2\leq t$ and $i\neq j$, $B_{(i),t_1}$ and $B_{(j),t_2}$ are conditionally independent given $\mathcal{F}_{t}^{CN}$, and furthermore,
$\psi_{(i),t}$ is conditionally independent of $\psi_{(j),t}$ and $B_{(j),t_1}$ given $\mathcal{F}_t^{CN}$.
The following proposition characterizes the average wealth in the economy. 
This, in turn, justifies the definition of discrete-time PRFPP under a mean field equilibrium given below.
\begin{proposition}\label{Prop:MFGAverageWealth}
Let $X$ and $X^{(i)}$ be the corresponding wealth process resulting from an admissible strategy $\pi \in \mathcal{A}^{MF}$. 
Then, it holds that $\frac{1}{N}\sum\limits_{i=1}^{N}X_{(i),t}\stackrel{\text{p}}{\rightarrow}\E\left[X_t\vert\mathcal{F}_{t}^{CN}\right]$ for any $t\in\mathbb{N}$, 
where $\stackrel{\text{p}}{\rightarrow}$ denotes convergence in probability as $N$ tends to infinity. Furthermore, 
$$\E\left[X_t\vert\mathcal{F}_{t}^{CN}\right]=\E\left[X_{t-1}\vert\mathcal{F}_{t-1}^{CN}\right]+\E\left[\pi_t\Delta_t^{1}\vert \mathcal{F}_{t-1}^{CN}\right]\mathbbm{1}_{\{\xi_t^{cn}=1\}}+\E\left[\pi_t\Delta_t^{0}\vert \mathcal{F}_{t-1}^{CN}\right]\mathbbm{1}_{\{\xi_t^{cn}=0\}}.$$
\end{proposition}

Following Proposition \ref{Prop:MFGAverageWealth}, we define the average wealth of the population as $\overline{X}_t:=\E\left[X_t\vert\mathcal{F}_{t}^{CN}\right]$. 
The relative wealth of the representative agent is then given by $X_t-\theta \overline{X}_t$.
We remark that the equality in probability comes from a weak version of the conditional law of large numbers. 
There seems to be no standard source for either the weak or strong version of the conditional law of large numbers, but the weak version can be shown following analogous steps as in the unconditional version. 
We work with the weak version for the sake of simplicity and 
because this is sufficient for the purpose of our paper. 
We next present the definition of PRFPP in a discrete-time mean field game.
\begin{definition}\label{def:RelativeFrowardPreferences-NQMFG}
A family of random functions $\left\lbrace U_{t}: \mathbb{R} \times \Omega \rightarrow \mathbb{R} \vert t \in \mathbb{N}_0 \right\rbrace$ together with an $\mathcal{F}_{t-1}^{MF}$-measurable admissible policy $\pi_t^{*} \in \mathcal{A}^{MF}$, is called a discrete-time predictable relative forward performance process (PRFPP) under a mean field equilibrium if the following conditions hold:
\begin{enumerate}
\item[(i)] $U_0\left(x,\cdot\right)$ is $\mathcal{F}_0^{MF}$-measurable and $U_{t}\left(x,\cdot\right)$ is $\mathcal{F}_{t-1}^{MF}$-measurable for each $x\in\mathbb{R}$ and $t \in \mathbb{N}$.
\item[(ii)] $U_{t}(\cdot,\omega)$ is a utility function for almost all $\omega \in \Omega$ and all $t \in \mathbb{N}_0$.
\item[(iii)] For any policy $\pi \in \mathcal{A}^{MF}$, the resulting wealth process $X$ of the representative agent satisfies, for almost all $\omega\in \Omega$,
\begin{align*}
U_{t-1}\left( X_{t-1}-\theta\E\left[X_{t-1}\vert\mathcal{F}_{t-1}^{CN}\right]\right) \geqslant \E \left[ U_{t} \left(X_{t}-\theta\E\left[X_t\vert\mathcal{F}_{t}^{CN}\right]\right) \big\vert \mathcal{F}_{t-1}^{MF} \right], \quad t \in \mathbb{N}.
\end{align*}
\item[(iv)]  By following policy $\pi^{*}$, the resulting wealth process $X^*$ of the representative agent satisfies, for almost all $\omega\in \Omega$,
\begin{align*}
U_{t-1}\left( X^*_{t-1}-\theta\E\left[X_{t-1}^*\vert\mathcal{F}_{t-1}^{CN}\right]\right) = \E \left[ U_{t} \left(X^*_{t}-\theta\E\left[X^*_{t}\vert\mathcal{F}_{t}^{CN}\right]\right) \big\vert \mathcal{F}_{t-1}^{MF} \right], \quad t \in \mathbb{N}.
\end{align*}
\end{enumerate}
\end{definition}

Properties $(iii)$ and $(iv)$ in Definition \ref{def:RelativeFrowardPreferences-NQMFG} imply that
\begin{align}\label{eq:InverseEqn_Dynamic_MFG}
U_{t-1}\left( X^*_{t-1}-\theta\E\left[X_{t-1}^*\vert\mathcal{F}_{t-1}^{CN}\right]\right)  = \sup_{X_{t} \in \mathcal{X}^{MF}_{t-1,t} \left( X^*_{t-1}\right)}  \E \left[ U_{t} \left(X_{t}-\theta\E\left[X_t\vert\mathcal{F}_{t}^{CN}\right]\right) \big\vert \mathcal{F}_{t-1}^{MF} \right].
\end{align}
Iteratively solving (\ref{eq:InverseEqn_Dynamic_MFG}) leads to the construction of PRFPPs under a mean field equilibrium. 
In analogy to the $N$ agent case, the crucial step is to solve for an agent in the group with any type vectors $\eta$ and $\zeta_0$, her relative forward utility functions $U_1$, and a strategy profile $\pi^{*}$ such that
\begin{equation}\label{eq:InverseProblemMFG}
\begin{aligned}
    U_{0} \left(X_{0}-\theta\overline{X}_{0} \right)  &= \sup_{X_{1} \in \mathcal{X}^{MF}_{0,1} \left(x\right)} \E \left[ U_{1} \left( X_{1}-\theta\overline{X}_{1}\right)\big\vert \mathcal{F}_{0}^{MF} \right]
    \\&=\sup_{{\pi_1} \in \mathcal{A}^{MF}_{0,1}} \E \big[ U_1 \big(x+\pi_1(R_1-1)  -\theta\E\left[x+\pi_1(R_1-1)\vert\mathcal{F}_{1}^{CN}\right]\big)\big\vert \mathcal{F}_{0}^{MF} \big]
\end{aligned}
\end{equation}
One can then construct $U_{2}$, $U_{3}$ and so on by repeatedly solving problem of the form (\ref{eq:InverseProblemMFG}) conditionally on updated information available at next time point and arguing that this solution satisfies the required measurability conditions. 



In the following, we will consider games with finitely many as well as an infinite number of competing agents. We assume that initial data belong to the family of exponential risk preferences. 
This allows for greater mathematical tractability. Specifically, the $i$th agent starts with an initial utility function regarding her relative wealth, $U_{(i),0}(\widetilde{X}_{(i),0})=-e^{-\gamma_{(i)}\widetilde{X}_{(i),0}}$, where the constant $\gamma_{(i)}>0$ denotes the risk aversion parameter.

\section{The $N$-agent game}\label{sec:NAgentGame}
 
The agents in our model not only trade between the risk-free bond and their respective binomial stocks to optimize the performance of their own wealth processes, but also take into account the average wealth level of other agents.
A better performance of other agents' portfolio strategies will lead to a lower degree of satisfaction for any given agent.

PRFPP in the finite population game can be constructed via a two-step procedure. 
For each individual agent, we first solve the corresponding inverse optimization problem for arbitrary but fixed trading strategies of the competitors.
This corresponds to the best-response forward relative framework introduced in \cite{anthropelos2022competition} for a two-agent game respectively in  \cite{dos2021forward} for an arbitrary number of agents. 
In the second step, we solve for the fixed point of a system of equations representing best-response strategies of agents, and derive the forward Nash equilibrium strategies as introduced in Definition \ref{def:RelativeFrowardPreferences-NQ}.

Let $\mathcal{M}$ 
denote the set of equivalent martingale measures and recall that
the conditional risk-neutral probability of an upward move of the stock traded by the $i$th agent is given by
$q_{(i),t} = \Q \left[ R_{(i),t} = u_{(i),t} \vert \mathcal{F}_{t-1} \right] = \frac{1-d_{(i),t}}{u_{(i),t}-d_{(i),t}}$, $t \in \mathbb{N}$, for any $\Q \in \mathcal{M}$.
We start with the first step by fixing a generic agent $i$ and her competitor's policies $\pi_{(j),t}$, $ j \neq i$ over time period $[t-1, t), t \in \mathbb{N},$ and then define 
$A^{1}_{(i),t}, A^{2}_{(i),t}, B^{1}_{(i),t}$ and $B^{2}_{(i),t}$ by 
\begin{align}\label{expressionA1}
    A^{1}_{(i),t}&=p_{(i),t}^1p_t^{cn}\prod \limits_{j\neq i}C_{jt}^{(i),1}+p_{(i),t}^0(1-p_t^{cn})\prod \limits_{j\neq i}C_{jt}^{(i),0},
\end{align}
\begin{align}\label{expressionA2}
    A^{2}_{(i),t}&=(1-p_{(i),t}^1)p_t^{cn}\prod \limits_{j\neq i}C_{jt}^{(i),1}+(1-p_{(i),t}^0)(1-p_t^{cn})\prod \limits_{j\neq i}C_{jt}^{(i),0},
\end{align}
\begin{align}\label{expressionB1}
    B^{1}_{(i),t}=\left(\frac{q_{(i),t}A^{2}_{(i),t}}{(1-q_{(i),t})A^{1}_{(i),t}}\right)^{(1-q_{(i),t})},
\end{align}
\begin{align}\label{expressionB2}
    B^{2}_{(i),t}=\left(\frac{q_{(i),t}A^{2}_{(i),t}}{(1-q_{(i),t})A^{1}_{(i),t}}\right)^{-q_{(i),t}},
\end{align}
with \begin{align*}
    C_{jt}^{(i),1}&=p_{(j),t}^1e^{\gamma_{(i)}\frac{\theta_{(i)}}{N}\pi_{(j),t}(u_{(j),t}-1)}+(1-p_{(j),t}^1)e^{\gamma_{(i)}\frac{\theta_{(i)}}{N}\pi_{(j),t}(d_{(j),t}-1)},\\
    C_{jt}^{(i),0}&=p_{(j),t}^0e^{\gamma_{(i)}\frac{\theta_{(i)}}{N}\pi_{(j),t}(u_{(j),t}-1)}+(1-p_{(j),t}^0)e^{\gamma_{(i)}\frac{\theta_{(i)}}{N}\pi_{(j),t}(d_{(j),t}-1)}.
\end{align*}

In the following theorem, we present an explicit construction of the PRFPP for an agent who is influenced by her opponents' actions but without interacting with them.
\begin{theorem}\label{Thm:RelativeForwardNAgents}
For $i \in \{1, \dots, N\}$ and $t\in \mathbb{N}$, 
let competitor policies $\pi^{*}_{(j),t}, j\neq i$ be given as $\mathcal{F}_{t-1}$-measurable random variables. 
Then, the process
\begin{align}\label{U_tNAgent}
    U_{(i),t}(\widetilde{x})=-\frac{e^{-\gamma_{(i)}\widetilde{x}}}{\prod \limits_{n=1}^t(A_{(i),n}^{1}B_{(i),n}^{1}+A_{(i),n}^{2}B_{(i),n}^{2})},
\end{align} 
with $U_{(i),0}(\widetilde{x})=-e^{-\gamma_{(i)}\widetilde{x}}$, is a PRFPP for the $i$th agent.
The optimal strategy for $n=1,2,\dots,t$ is given by
\begin{align}\label{OptimalStrategyN1}
    \pi_{(i),n}^{*}=\frac{1}{\gamma_{(i)} (1-\frac{\theta_{(i)}}{N})(u_{(i),n}-d_{(i),n})}\ln \frac{(1-q_{(i),n})A_{(i),n}^{1}}{q_{(i),n}A_{(i),n}^{2}},
\end{align}
which generates the corresponding optimal wealth process
\begin{align}\label{OptimalWealthN}
     X_{(i),t}^{*}= X_{(i),0}+\frac{1}{\gamma_{(i)} (1-\frac{\theta_{(i)}}{N})}\sum_{n=1}^{t}\frac{R_{(i),n}-1}{(u_{(i),n}-d_{(i),n})}\ln \frac{(1-q_{(i),n})A_{(i),n}^{1}}{q_{(i),n}A_{(i),n}^{2}}.
\end{align}
\end{theorem}

Theorem \ref{Thm:RelativeForwardNAgents} gives us the best-response strategy for the $i$th agent in terms of the given strategies of all other agents $j\neq i$. 
In the following, we will consider the case where agents interact dynamically and competitively for two special cases: The case where agents are homogeneous and the case where $N = 2$. 

We again emphasize that the market considered herein is  incomplete due to the updating of parameters and the non-traded stochastic factor.
Forward performance processes without relative performance concerns  in incomplete binomial models have been studied in \cite{musiela2016forward}. 
Without competition, setting $\theta_{(i)}=0$, it can be readily shown that
$A_{(i),t}^{1}=p_{(i),t}, A_{(i),t}^{2}=1-p_{(i),t},$ and $A_{(i),t}^{1}B_{(i),t}^{1}+A_{(i),t}^{2}B_{(i),t}^{2}=\left(\frac{q_{(i),t}}{p_{(i),t}}\right)^{-q_{(i),t}}\left(\frac{1-q_{(i),t}}{1-p_{(i),t}}\right)^{q_{(i),t}-1}$. 
It is then easy to show that the forward performance process and the corresponding optimal strategy obtained herein contain the corresponding result in \cite[Theorem 3]{musiela2016forward} as a special case. 



Having obtained an explicit characterization of the equilibrium strategies, we next investigate the limiting behavior of the model when one or more agents become risk neutral.
It follows immediately from equation (11) that, as her risk-aversion parameter approaches zero, the risk-neutral agent own equilibrium strategy becomes unbounded. This behavior is not surprising, since the expected-utility maximization problem of a risk-neutral investor is ill-posed in the presence of risky assets with positive excess return. 
Now suppose that the market offers a positive expected excess return and that agent $j$ becomes risk neutral, so that $\gamma_{(j)}\to 0$ and hence $\pi_{(j),t}^*\to +\infty$. For a risk-averse competitor $i\neq j$, the peer-interaction terms satisfy
\begin{align*}
    C_{jt}^{(i),1}\rightarrow p_{(j),t}^1,e^{\gamma_{(i)}\frac{\theta_{(i)}}{N} \pi_{(j),t}(u_{(j),t}-1)},\qquad
C_{jt}^{(i),0}\rightarrow p_{(j),t}^0,e^{\gamma_{(i)}\frac{\theta_{(i)}}{N} \pi_{(j),t}(u_{(j),t}-1)},
\end{align*}
as $\pi_{(j),t}\to+\infty$.
Factoring out the common term $\gamma_{(i)}\frac{\theta_{(i)}}{N} \pi_{(j),t}(u_{(j),t}-1)$  using the asymptotics above, it cancels in the ratio $\frac{A_{(i),t}^1}{A_{(i),t}^2}$. Therefore the ratio converges to a finite constant given by
$$\frac{A_{(i),t}^1}{A_{(i),t}^2}\ \longrightarrow
\frac{
p_{(i),t}^1p_t^{cn}p_{(j),t}^1 \prod_{k\neq i,j} C_{kt}^{(i),1}
+
p_{(i),t}^0(1-p_t^{cn})p_{(j),t}^0 \prod_{k\neq i,j} C_{kt}^{(i),0}
}{
(1-p_{(i),t}^1)p_t^{cn}p_{(j),t}^1 \prod_{k\neq i,j} C_{kt}^{(i),1}
+
(1-p_{(i),t}^0)(1-p_t^{cn})p_{(j),t}^0 \prod_{k\neq i,j} C_{kt}^{(i),0}
},$$
Consequently, although the risk-neutral agent's optimal position diverges, the equilibrium strategy of a risk-averse competitor remains well defined. In particular,
\begin{align*}
\lim_{\gamma_{(j)}\to0}\pi_{(i),t}^*
=\frac{\ln\left(
\frac{1-q_{(i),n}}{q_{(i),n}}
\frac{
p_{(i),t}^1p_t^{cn}p_{(j),t}^1 \prod_{k\neq i,j} C_{kt}^{(i),1}
+
p_{(i),t}^0(1-p_t^{cn})p_{(j),t}^0 \prod_{k\neq i,j} C_{kt}^{(i),0}
}{
(1-p_{(i),t}^1)p_t^{cn}p_{(j),t}^1 \prod_{k\neq i,j} C_{kt}^{(i),1}
+
(1-p_{(i),t}^0)(1-p_t^{cn})p_{(j),t}^0 \prod_{k\neq i,j} C_{kt}^{(i),0}
}
\right)}{\gamma_{(i)}\left(1-\frac{\theta_{(i)}}{N}\right)\big(u_{(i),n}-d_{(i),n}\big)}
.    
\end{align*}
Hence, while the strategy of a risk-neutral agent explodes, her influence on risk-averse competitors remains bounded, and their equilibrium strategies converge to finite limits.

\subsection{Homogeneous multi-agent game}
To obtain an explicit solution with multiple agents, we make the further assumption that agents are homogeneous, i.e., the parameter values of every agent are identical across the network. 
In this setting, the parameters $p_{(i),t}^1$, $p_{(i),t}^0$, $\gamma_{(i)}$, $\theta_{(i)}$, $u_{(i),t}$, $d_{(i),t}$, $q_{(i),t}$, $A_{(i),t}^{1}$, $A_{(i),t}^{2}$, $B_{(i),t}^{1}$, $B_{(i),t}^{2}$, $C_{jt}^{(i),1}$ and $C_{jt}^{(i),0}$ for $i,j=1,2,\dots,N$ do not depend on $i,j$ and can thus be abbreviated to $p_t^{1}$, $p_t^{0}$, $\gamma$, $\theta$, $u_t$, $d_t$, $q_t$, $A_t^{1}$, $A_t^{2}$, $B_t^{1}$, $B_t^{2}$, $C_t^{1}$ and $C_t^{0}$.

In an $N$-agent homogeneous model, the optimization problem faced by any agent is symmetric. 
In this setting, we focus on finding a symmetric equilibrium where all agents follow the same strategy.
Thus, the problem studied reduces to a game where the trading policies executed by all $N$ agents are identical. 
We refer the reader to \citep[p.~258]{binmore2007playing} and \cite{hefti2017equilibria} for more detailed discussion on symmetric equilibria in symmetric games.
Let $\pi_t^{*}$ denote the equilibrium strategy of any agent in the group  of $N$ homogeneous agents for time period $[t-1,t), t\in \mathbb{N}$. 
By Theorem \ref{Thm:RelativeForwardNAgents}, the next step can be simplified from solving a system of $N$ equations of the form (\ref{OptimalStrategyN1})
to a single equation given by
\begin{align}\label{OptimalStrategyHomoN}
    \pi_t^{*}=\frac{1}{\gamma (1-\frac{\theta}{N})(u_t-d_t)}\ln \frac{(1-q_t)A_t^{1}}{q_tA_t^{2}}.
\end{align}
Note that the unknown $\pi_t^{*}$ appears on both sides of the equation, with $A^{1}_t$ and $A^{2}_t$ given by
\begin{align*}
A^{1}_t&=p_t^{1}p_t^{cn}\left(C_{t}^{1}\right)^{N-1}+p_t^{0}(1-p_t^{cn})\left(C_{t}^{0}\right)^{N-1},
\end{align*}
\begin{align*}
A^{2}_t&=(1-p_t^{1})p_t^{cn}\left(C_{t}^{1}\right)^{N-1}+(1-p_t^{0})(1-p_t^{cn})\left(C_{t}^{0}\right)^{N-1},
\end{align*}
where 
    $C_{t}^{1}=p_t^{1}e^{\gamma \frac{\theta}{N}\pi^{*}_t(u_t-1)}+(1-p_t^{1})e^{\gamma\frac{\theta}{N}\pi^{*}_t(d_t-1)}$ and $
    C_{t}^{0}=p_t^{0}e^{\gamma\frac{\theta}{N}\pi^{*}_t(u_t-1)}+(1-p_t^{0})e^{\gamma \frac{\theta}{N}\pi^{*}_t(d_t-1)}.$
Using the change of variables $y_t=e^{(1-\frac{\theta}{N})\gamma (u_t-d_t)\pi_t^*}$, equation (\ref{OptimalStrategyHomoN}) can be transformed to 
\begin{align}\label{HomoN}
y_t=\frac{1-q_t}{q_t}\frac{p_t^{1}p_t^{cn}\left(\frac{p_t^{1}y_t^{\frac{\theta}{N-\theta}}+(1-p_t^{1})}{p_t^{0}y_t^{\frac{\theta}{N-\theta}}+(1-p_t^{0})}\right)^{N-1}+p_t^{0}(1-p_t^{cn})}{(1-p_t^{1})p_t^{cn}\left(\frac{p_t^{1}y_t^{\frac{\theta}{N-\theta}}+(1-p_t^{1})}{p_t^{0}y_t^{\frac{\theta}{N-\theta}}+(1-p_t^{0})}\right)^{N-1}+(1-p_t^{0})(1-p_t^{cn})}, \quad y_t>0.
\end{align}
Let $f(y_t)=\frac{1-q_t}{q_t}\frac{p_t^{1}p_t^{cn}\left(\frac{p_t^{1}y_t^{\frac{\theta}{N-\theta}}+(1-p_t^{1})}{p_t^{0}y_t^{\frac{\theta}{N-\theta}}+(1-p_t^{0})}\right)^{N-1}+p_t^{0}(1-p_t^{cn})}{(1-p_t^{1})p_t^{cn}\left(\frac{p_t^{1}y_t^{\frac{\theta}{N-\theta}}+(1-p_t^{1})}{p_t^{0}y_t^{\frac{\theta}{N-\theta}}+(1-p_t^{0})}\right)^{N-1}+(1-p_t^{0})(1-p_t^{cn})}$.
We show the existence and uniqueness of the solution to the associated fixed point equation (\ref{HomoN}) in the following theorem. 
\begin{theorem}\label{Thm:HomoNAgentsGame}
Given $\gamma>0$, $\theta \in [0,1]$ and $0<q_t, p_t^{1}, p_t^{0}, p_t^{cn}<1$, equation (\ref{HomoN}) admits a unique positive solution $y_t^*$, $t\in \mathbb{N}$. 
Furthermore, there exists a forward Nash equilibrium of the homogeneous $N$-agent game for period $[t-1,t)$ given by
\begin{align}\label{HomoNStrategy}
    \pi_t^*=\frac{1}{\gamma (1-\frac{\theta}{N})(u_t-d_t)}\ln y_t^*.
\end{align}

Accordingly, the process
\[
U_t(\widetilde{x})=-\frac{e^{-\gamma \widetilde{x}}}{\prod_{n=1}^{t}
\Big(
(A_n^{1})^{q_n}(A_n^{2})^{1-q_n}
\big(
(\tfrac{q_n}{1-q_n})^{1-q_n}
+
(\tfrac{1-q_n}{q_n})^{q_n}
\big)
\Big)}
\]
with $U_0(\widetilde{x})=-e^{-\gamma\widetilde{x}}$, is a PRFPP for each agent.
\end{theorem}


\begin{remark}
    Theorem 2 guarantees the existence and uniqueness of the solution, thereby ensuring that the numerical computation is well posed. Numerically, the fixed point can be obtained by applying standard root-finding methods, such as Newton's method, to the equation $f(y_t)-y_t=0$. 
\end{remark}

The equilibrium strategy $\pi_t^*$ depends on the risk aversion $\gamma$, the number of agents $N$, the competition weight $\theta$, the possible returns $u_t, d_t$,  and, implicitly through $y_t^*$, on the market parameters, $p_t^{cn}$, $p_t^{1}$ and $p_t^{0}$.
While it is difficult to perform an analytical sensitivity analysis for all results because of the implicit definition of $y_t^*$, it can still be deduced that
as the agent and her competitors become more risk-averse as a whole, her absolute position on the risky asset decreases. 
This observation also holds true in the literature on investment under relative performance criteria for continuous-time setting, see, for example, \cite{lacker2019mean}, \cite{anthropelos2022competition} and \cite{dos2021forward, dos2022forward}. 
On the other hand, when all assets invested by the entire population are anticipated to achieve better performance with a constant distance $u_t-d_t$, the agent tends to increase her risky investment.
It is worth noting that in \cite{lacker2019mean} and \cite{dos2021forward, dos2022forward}, market parameters and, thus, optimal strategies are constant over time.
However, our model herein allows for dynamic updating of parameters, resulting in a time-varying forward Nash equilibrium strategy.


As we will see in the next section when considering the mean field game (MFG),  the $N$-agent game with homogeneous agents converges to the MFG with respect to both the associated fixed point equation and the equilibrium strategy as $N\rightarrow +\infty$.

\subsection{Heterogeneous 2-agent game}
In the heterogeneous case, we limit ourselves to the game of $2$ agents. This allows us to show the existence of an equilibrium strategy. 
It still remains an open and challenging problem to investigate the system of equations (\ref{OptimalStrategyN1}) for an arbitrary number of heterogeneous agents. 

We label agents with superscript $(1), (2)$ to distinguish between the first and the second agent. Define $p_t^{1,1}=p_t^{cn}p_{(2),t}^1p_{(1),t}^1+(1-p_t^{cn})p_{(2),t}^0p_{(1),t}^0$, 
$p_t^{0,1}=p_t^{cn}p_{(2),t}^1(1-p_{(1),t}^1)+(1-p_t^{cn})p_{(2),t}^0(1-p_{(1),t}^0)$, 
$p_t^{1,0}=p_t^{cn}(1-p_{(2),t}^1)p_{(1),t}^1+(1-p_t^{cn})(1-p_{(2),t}^0)p_{(1),t}^0$,
and $p_t^{0,0}=p_t^{cn}(1-p_{(2),t}^1)(1-p_{(1),t}^1)+(1-p_t^{cn})(1-p_{(2),t}^0)(1-p_{(1),t}^0)$.
We first present the following auxiliary result.
\begin{lemma}\label{ExistenceoFor2agents}
Given $\gamma_{(k)}>0$, $\theta_{(k)} \in [0,1]$ and $0<q_{(k),t}, p_{(k),t}^1, p_{(k),t}^0, p_t^{cn}<1$ for $k=1,2$, there exists an almost surely positive solution $y_t^{*}\in \mathcal{F}_{t-1}$ to the equation 
\begin{align}\label{EqnOfUnknowny}
    y_t^{\frac{(2-\theta_{(1)})\gamma_{(1)}}{\theta_{(2)}\gamma_{(2)}}}=\frac{(1-q_{(1),t})\left(p_t^{1,1}\left(\frac{(1-q_{(2),t})(p_t^{1,1}y_t+p_t^{0,1}))}{q_{(2),t}(p_t^{1,0}y_t+p_t^{0,0})}\right)^{\frac{\gamma_{(1)}\theta_{(1)}}{\gamma_{(2)}(2-\theta_{(2)})}}+p_t^{1,0}\right)}{q_{(1),t}\left(p_t^{0,1}\left(\frac{(1-q_{(2),t})(p_t^{1,1}y_t+p_t^{0,1})}{q_{(2),t}(p_t^{1,0}y_t+p_t^{0,0})}\right)^{\frac{\gamma_{(1)}\theta_{(1)}}{\gamma_{(2)}(2-\theta_{(2)})}}+p_t^{0,0}\right)}.
\end{align}
Furthermore, $y_t^*<\left(\frac{p_t^{1,1}(1-q_{(1),t})}{p_t^{0,1}q_{(1),t}}\right)^{\frac{\theta_{(2)}\gamma_{(2)}}{(2-\theta_{(1)})\gamma_{(1)}}}$ almost surely.
\end{lemma}

We next present one of our main results showing how to construct a forward Nash equilibrium strategy for two heterogeneous agents.


\begin{theorem}\label{Thm:2AgentsGame}
Let $y_t^{*}>0$ denote a positive solution of (\ref{EqnOfUnknowny}). 
There exists a forward Nash equilibrium of the $2$-agent game for period $[t-1,t)$, given by
\begin{align}\label{OptimalStrategyj}
    \pi_{(1),t}^{*}=\frac{2\ln y_t^*}{(u_{(1),t}-d_{(1),t})\gamma_{(2)}\theta_{(2)}},
\end{align}
\begin{align}\label{OptimalStrategyi}
    \pi_{(2),t}^{*}=\frac{1}{\gamma_{(2)} (1-\frac{\theta_{(2)}}{2})(u_{(2),t}-d_{(2),t})}\ln \frac{(1-q_{(2),t})\left(p_t^{1,1}y_t^* + p_t^{0,1}\right)}{q_{(2),t}\left(p_t^{1,0}y_t^*+p_t^{0,0}\right)}.
\end{align}
\end{theorem}



We close the section with a corollary that has interesting implications for investment behavior. 
\begin{corollary}\label{Cor:pi2Heter}
    In the heterogeneous 2-agent game over period $[t-1,t)$, $\pi_{(1),t}^{*}$ increases as $u_{(2),t}$ or $d_{(2),t}$ increases. Furthermore, if $p_{(1),t}^1<\frac{1}{2}<q_{(1),t}$, then $y_t^*<1$ and $\pi_{(1),t}^{*}<0$ for all parameter values of the opponent.
\end{corollary}
The agent increases her risky investment if the stock held by her opponent has higher price levels $u_{(2),t}$ or $d_{(2),t}$. 
In other words, if the stock traded by her competitor becomes more attractive, the agent increases the investment in her own stock even if the characteristics of her own stock remain the same. Observe further from Corollary \ref{Cor:pi2Heter} that the agent is expected to sell her stock short when the price levels and transition probability for an upward move of her stock are not desirable independent of the characteristics of the stocks traded by her competitor. 


Recall that in the classical expected utility framework, an investor without relative performance concern always longs a stock when the expected return exceeds one, abstains from investing when it equals one, and shorts the stock when it falls below one. 
Under PRFPPs, this holds no longer true. 
For example, consider the setting where $p_{(1),t}^1=0.6$, $p_{(1),t}^0=0.36$, $p_{(2),t}^1=0.46$, $p_{(2),t}^0=0.2$, $\gamma_{(1)}=\gamma_{(2)}=10$,  $p_t^{cn}=0.6$, $\theta_{(1)}=\theta_{(2)}=0.9$, $u_{(1),t}=u_{(2),t}=1.1$, and $d_{(1),t}=d_{(2),t}=0.9$.
For these parameters the expected excess return of $R^{(1)}$ is positive, but the investor shorts the stock, $\pi_{(1),t}^{*}=-0.0098<0$. 
Shorting a stock with positive expected excess return in the presence of relative performance concerns is an indirect effect. If the stock traded by the opponent has a negative expected excess return, an agent's  investment in the risky asset decreases and may even become negative.  
The observation that relative performance concerns can decrease the investment in the risky asset, even to the point where an agent shorts a stock with a positive expected excess return, is interesting and seems to be new. 
This phenomenon occurs even in a backward setting with relative performance concerns, where the terminal utility and market model are specified ex ante for a fixed investment horizon, yet it has not been explicitly noted in the literature. For instance, \cite{lacker2020many} report that relative performance concerns typically drive agents to invest more in risky assets under CARA preferences. However, their findings rely on the assumption that the risky assets traded by all agents have positive expected excess returns, meaning the mechanism described does not apply.
A large empirical literature shows that fund managers adjust portfolio risk and active positions in response to relative performance concerns, career incentives, and performance-sensitive investor flows, see for example, \cite{brown1996tournaments}, \cite{chevalier1997risk} and \cite{chevalier1999career}).
However, to the best of our knowledge, there is no direct empirical evidence testing whether relative performance concerns can lead agents to short the market despite a positive expected excess return, which remains an interesting direction for future empirical or experimental research.

We remark that incorporating portfolio constraints, such as no-short-selling restrictions arising in many realistic market settings, remains an open problem in the discrete-time forward framework. Even without competition, constructing predictable forward performance processes under such control constraints is not yet well understood and is left for future research.

While the heterogeneous two-agent game admits an explicit analytical characterization, extending the analysis to a general heterogeneous N-agent setting quickly becomes intractable due to the nonlinear interaction of agents' strategies. Nevertheless, examining the N-agent case remains important for understanding how competition operates in larger populations. We therefore complement herein with a numerical study of the heterogeneous N-agent game, which allows us to illustrate how the equilibrium evolves as the number of competitors increases and to highlight qualitative patterns that are not visible in the two-agent setting. 

Consider an $N$-agent game in which agents $2,\dots,N$ are homogeneous, whereas agent $1$ differs from the rest and is referred to as the odd agent. As the $N-1$ homogeneous agents face identical optimization problems, they adopt the same equilibrium strategy. Figure \ref{figHeterN} illustrates the impact of population size on the equilibrium strategy of the odd agent and its deviation from the majority.
\begin{figure}[H]
\centering
  \includegraphics[width=10cm]{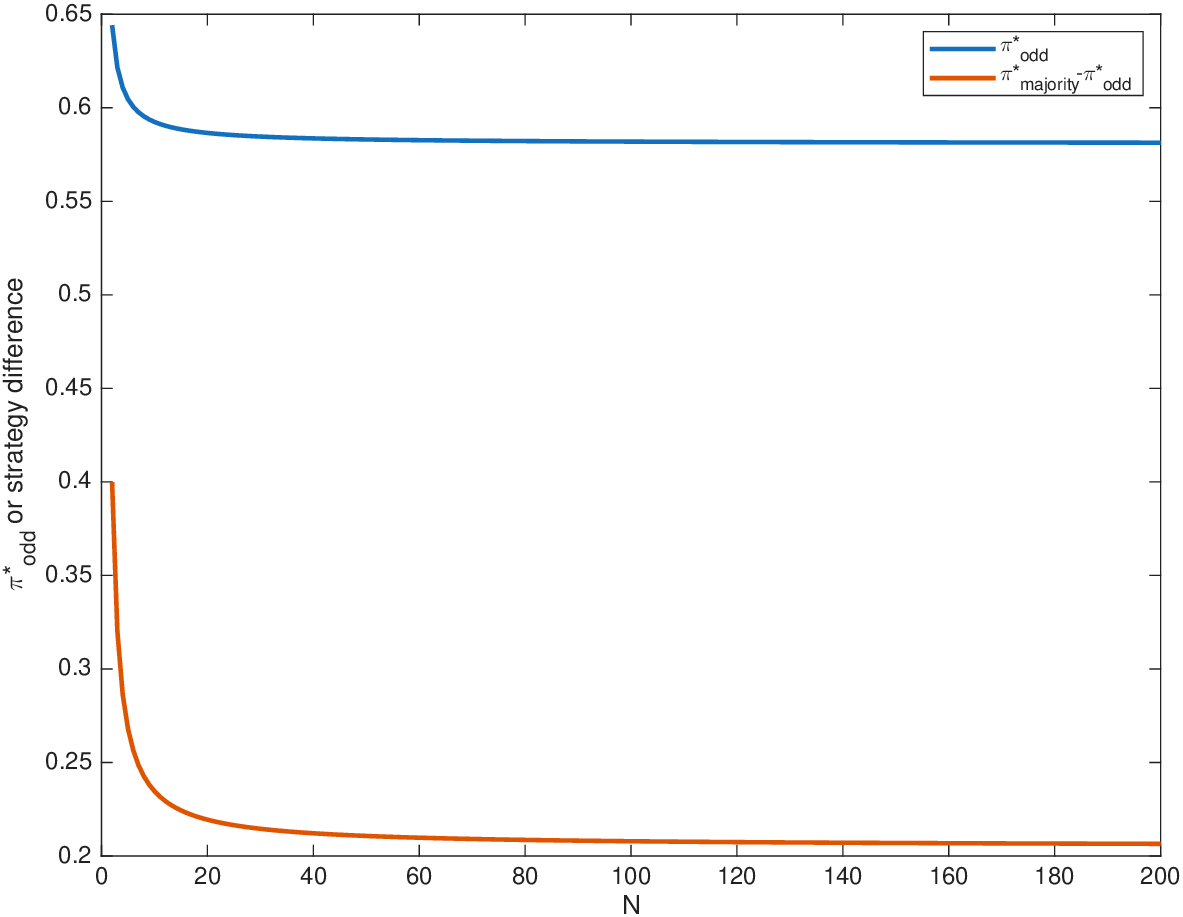}
  \caption{Heterogeneous $N$-Agent game: population effects\\
  {\footnotesize{\textit{Notes}. We fix the parameters of the odd agent as $(\gamma_O,\theta_O,u_O,d_O)=(2,0.2,1.4,0.8)$, who competes against $N-1$ homogeneous agents with parameters $(\gamma_M,\theta_M,u_M,d_M)=(3,0.5,1.2,0.9)$. We then vary the population size (N). The transition probabilities are set uniformly across agents as $(p^1,p^0,p^{cn})=(0.6,0.4,0.5)$.
}}}
  \label{figHeterN}
\end{figure}

We first find that the odd agent's equilibrium strategy decreases as the number of identical competitors increases. This can be understood from the structure of relative wealth, $\widetilde{X}_{(i),t}=\left(1-\frac{\theta_{(i)}}{N}\right)X_{(i),t}-\frac{\theta_{(i)}}{N}\sum_{k\neq i}X_{(k),t}.$ As the number of identical competitors grows, the average peer wealth against which performance is evaluated becomes less volatile. In particular, the odd agent's optimization problem is increasingly dominated by the first term, whose coefficient $1-\frac{\theta_{(i)}}{N}$ converges to one as $N$ increases, while the second term, which is responsible for inducing more aggressive risk taking in competitive environments, becomes less influential. Moreover, by the law of large numbers, the variance of the average peer wealth diminishes as $N$ grows, so that the benchmark can be regarded as asymptotically constant. Under CARA preferences, this implies that the odd agent's equilibrium strategy converges to the optimal strategy obtained in the absence of relative performance concerns. 
At the same time, across a wide range of parameter constellations, we consistently observe that the absolute difference between the equilibrium strategy of the odd agent and that of the identical competitors decreases as the population size increases. This pattern reflects an conformity effect: as the majority becomes larger, the odd agent is increasingly inclined to follow the prevailing investment behavior, gradually reducing deviations from the general market trend.

\section{The mean field game}\label{sec:MFG}
In this section, \textit{we assume that the market almost surely offers a positive expected excess return for the representative agent, i.e., $\E[R_t-1\vert \mathcal{F}^{MF}_{t-1}]>0$.} Recall that
the conditional risk-neutral probability of an upward move of the stock traded by the representative agent is given by
$q_t = \frac{1-d_t}{u_t-d_t}$, $t \in \mathbb{N}$.
The following auxiliary result will be used to prove the existence and uniqueness of the mean field equilibrium (MFE).
\begin{lemma}\label{Lemma:fixedpoint}
Let $F: \mathbb{N} \times \mathbb{R}\times \Omega \rightarrow \mathbb{R}$ denote a mapping given by
\begin{equation}\label{fixedpoint}
\begin{aligned}
    F(t,y,\omega):=\phi_1 ( \omega) + \E\left[\frac{({\Delta}_t^1-{\Delta}_t^0)}{\gamma(u_t-d_t)}\ln\left(1+\frac{(\frac{p_t^{0}}{p_t^{1}}-\frac{1-p_t^{0}}{1-p_t^{1}})}{\frac{p_t^{cn}}{1-p_t^{cn}}e^{\gamma \theta y}+\frac{1-p_t^{0}}{1-p_t^{0}}}\right)\Bigg\vert\mathcal{F}_{t-1}^{CN}\right] (\omega),
\end{aligned}
\end{equation}
where $\phi_1=\E\left[\frac{({\Delta}_t^1-{\Delta}_t^0)}{\gamma(u_t-d_t)}\ln\frac{(1-q_t)p_t^{1}}{q_t(1-p_t^{1})}\big\vert\mathcal{F}_{t-1}^{CN}\right]$, ${\Delta}_t^1=p_t^{1}(u_t-1)+(1-p_t^{1})(d_t-1)$, and ${\Delta}_t^0=p_t^{0}(u_t-1)+(1-p_t^{0})(d_t-1)$. 
Then for any fixed $t \in \mathbb{N}$ and almost all $\omega \in \Omega$, $y=F(t,y,\omega)$ admits a unique solution $y_t^{*}(\omega)\in \mathbb{R}^{+}$. 
Moreover, the random variable $y_t^{*}$ is $\mathcal{F}_{t-1}^{CN}$-measurable. 
\end{lemma}

By the proof of Lemma \ref{Lemma:fixedpoint}, $F(t,\cdot,\omega): \mathbb{R}^{+}\rightarrow \mathbb{R}$ is a contraction for fixed $t\in \mathbb{N}$ and fixed $\omega \in \Omega$. 
Thus, the unique fixed point can be found by Picard iteration: Start with an arbitrary element $y_t^0$, let $y_t^j(\omega)=F(t,y_t^{j-1}(\omega),\omega)$, $j=1,2,\dots$, and obtain $y_t^*(\omega)=\lim\limits_{j\rightarrow \infty}y_t^j(\omega)$ by the Banach fixed-point theorem.
 We next present another auxiliary result for each single-period backward problem.


\begin{lemma}\label{lemmaSupMFG}
Let $y_t^{*}$, $t\in \mathbb{N}$ be the unique positive solution of (\ref{fixedpoint}).
We have
\begin{align}\label{eq:SupMFG}
    \sup_{\pi_t}\E[-e^{-\gamma\widetilde X_{t}}\vert \mathcal{F}_{ t-1}^{MF}]
    &=-e^{-\gamma \widetilde X_{t-1}} G_t(\pi_t^{*}),
\end{align}
where
\begin{equation}\label{Gpit}
\begin{aligned}
G_t(\pi_t)=&p_t^{cn}e^{\gamma \theta\E\left[\pi_t{\Delta}_t^{1}\vert \mathcal{F}_{t-1}^{CN}\right]}p_t^{1}e^{-\gamma\pi_t(u_t-1)}+p_t^{cn}e^{\gamma \theta\E\left[\pi_t{\Delta}_t^{1}\vert \mathcal{F}_{t-1}^{CN}\right]}(1-p_t^{1})e^{-\gamma\pi_t(d_t-1)}\\&+(1-p_t^{cn})e^{\gamma \theta\E\left[\pi_t{\Delta}_t^{0}\vert \mathcal{F}_{t-1}^{CN}\right]}p_t^{0}e^{-\gamma\pi_t(u_t-1)}+(1-p_t^{cn})e^{\gamma \theta\E\left[\pi_t{\Delta}_t^{0}\vert \mathcal{F}_{t-1}^{CN}\right]}(1-p_t^{0})e^{-\gamma\pi_t(d_t-1)}.
\end{aligned}
\end{equation}
The optimal policy for \eqref{eq:SupMFG} is given by
\begin{align}\label{OptimalStrategyMFG}
    \pi_t^{*}=\frac{1}{\gamma(u_t-d_t)}\ln\left(\frac{(1-q_t)\left(p_t^{cn}e^{\gamma \theta y_t^*}p_t^{1}+(1-p_t^{cn})p_t^{0}\right)}{q_t\left(p_t^{cn}e^{\gamma \theta y_t^*}(1-p_t^{1})+(1-p_t^{cn})(1-p_t^{0})\right)}\right).
\end{align}
\end{lemma}

We are now able to present one of our main results on the construction of a forward MFE. 

\begin{theorem}\label{Thm:RelativeForwardMFG}
For $\widetilde{x}\in \mathbb{R}$ and $t \in \mathbb{N}$, the process
\begin{align}
    U_t(\widetilde{x})=-\frac{e^{-\gamma \widetilde{x}}}{\prod \limits_{n=1}^tG_n(\pi_n^{*})},
\end{align} 
with $U_0(\widetilde{x})=-e^{-\gamma\widetilde{x}}$, is a PRFPP with its associated MFE $\pi_t^{*}$ given by \eqref{OptimalStrategyMFG}.
\end{theorem}

The representative agent is characterized by stochastic type vectors that represent the distribution of market and preference characteristics across an infinite number of agents. 
Therefore, the PRFPPs and the corresponding MFE depend on the realization of the type vectors $(\psi_t)_{t \in \mathbb{N}}$. 

The equilibrium strategy of the representative agent given by (\ref{OptimalStrategyMFG}) is influenced by the preference parameters and the market parameters through the type distribution of the population as a whole, and is increasing in $y_t^*$. 
Then, given any specific agent in the game with fixed parameter values, an increase in her risky investment can be the consequence of an increase in the transition probability of the non-traded stochastic factor $p_t^{cn}$, or changes in some of the population, e.g., a decrease in risk aversion  or increase in competition weight of the population.
Indeed, the function $F(t,y,\omega)$ is increasing in the above mentioned parameters implying a larger $y_t^*$ as the fixed point solution to equation \eqref{fixedpoint}.
In a nutshell, an agent with relative performance concerns typically invests less in the stock if some of her competitors are more risk-averse or less competitive, or when the stocks traded in the population market become less attractive.


It turns out that the equilibrium strategy $(\ref{OptimalStrategyMFG})$ can be expressed as the sum of two components,
\begin{align*}
    \pi_t^{*}=\frac{1}{\gamma(u_t-d_t)}\left( \ln\frac{(1-q_t)p_t}{q_t(1-p_t)}+\ln\frac{(1-p_t)\left(p_t^{cn}e^{\gamma \theta y_t^*}p_t^{1}+(1-p_t^{cn})p_t^{0}\right)}{p_t\left(p_t^{cn}e^{\gamma \theta y_t^*}(1-p_t^{1})+(1-p_t^{cn})(1-p_t^{0})\right)}\right).
\end{align*}
One is the classical Merton portfolio, the second one is increasing in the competition weight $\theta$ and vanishes when $\theta=0$, all agents act uncompetitively.
In addition, the second component is always positive, we observe the same pattern as in \cite{lacker2019mean} that competition is an incentive for agents with CARA utilities to increase their risky investment.

\begin{remark}
The fixed point solution derived in Lemma \ref{Lemma:fixedpoint} satisfies $$y_t^{*}=\E\left[\pi_t^{*}({\Delta}_t^1-{\Delta}_t^0)\vert \mathcal{F}_{t-1}^{CN}\right]>0,$$ which holds if and only if
\begin{align*}
   \E\left[ \pi_t^{*} \left(p_t^{1}(u_t-1)+(1-p_t^{1})(d_t-1)\right) \big\vert \mathcal{F}_{t-1}^{CN} \right]
   > \E\left[\pi_t^{*} \left(p_t^{0}(u_t-1)+(1-p_t^{0})(d_t-1)\right) \big\vert \mathcal{F}_{t-1}^{CN} \right].
\end{align*}
Therefore, the financial interpretation of the positiveness of the fixed point is that the average of the expected gain across the entire population of agents in the state $\{\xi_t^{cn}=1\}$ is higher than in $\{\xi_t^{cn}=0\}$. 
This in turn justifies designation of the two states as bull respectively bear market.
\end{remark}

We close this section with a proposition showing that the forward
multi-agent game converges to the forward MFG given agents with homogeneous type vectors, as validated in \cite{lacker2019mean}, \cite{anthropelos2022competition} and \cite{dos2021forward, dos2022forward}.
\begin{proposition}\label{Prop:Convergence}
As the number of agents tends to infinity, the equilibrium strategy $(\ref{HomoNStrategy})$ of multi-agent game converges in distribution to the MFE $(\ref{OptimalStrategyMFG})$.
\end{proposition}
Whether an analogous result to Proposition 
\ref{Prop:Convergence} holds for heterogeneous agents is a challenging open problem left to future research.

\section{Numerical analysis of MFE}\label{sec:NumericalAnalysis}

In this section, we perform a numerical study to do a sensitivity analysis of model parameters for the MFE. 
We will investigate the sensitivity of the optimal strategies with regards to preference parameters in Subsection \ref{subs:SensitivityPreference} and the sensitivity with regards to market parameters in Subsection \ref{subs:SensitivityMarket}. 
Throughout the section, our analysis is based on a simulation of $10,000$ agents.  
We work with uniformly distributed benchmark parameters given by
$p_t^{1} \sim \mathcal{U}[0.5,0.7]$, $p_t^{0} \sim \mathcal{U}[0.3,0.5]$, $u_t \sim \mathcal{U}[1.16,1.24]$, $d_t \sim \mathcal{U}[0.86,0.94]$, $\gamma \sim \mathcal{U}[2,4]$, and $\theta \sim \mathcal{U}[0.2,0.6]$. 
We suppose that $p_t^{cn}=0.6$ for all agents. 
When considering a fixed agent, we set the benchmark parameters to $p_t^{1}=0.6$, $p_t^{0}=0.4$, $u_t=1.2$, $d_t=0.9$, $\theta=0.4$, and $\gamma=3$.

\subsection{Sensitivity analysis for preference parameters}\label{subs:SensitivityPreference}

We first investigate the impact of the preference parameters $\theta$ and $\gamma$ on optimal investment strategies.
To this end, we compute the optimal investment strategy of the agent numerically in two scenarios. 
First, we vary the competition weight $\theta$ and risk aversion $\gamma$ of that particular agent while keeping the preference parameters of the network fixed. 
Second, we fix the preference parameters of the agent and vary the distributions of the competition weight $\theta$ and risk aversion $\gamma$ across the network by changing their respective range.

\begin{figure}[ht]
   \includegraphics[width=7.96cm]{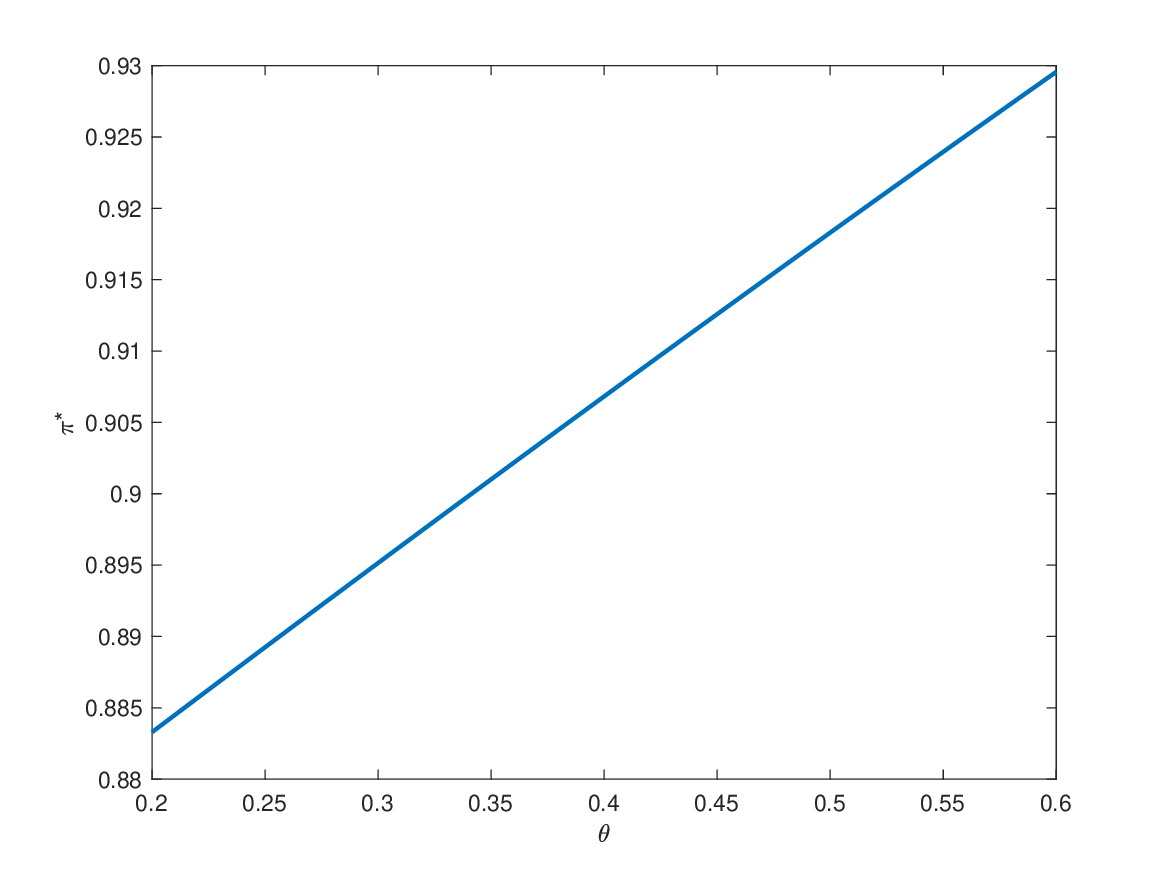}
  \includegraphics[width=7.96cm]{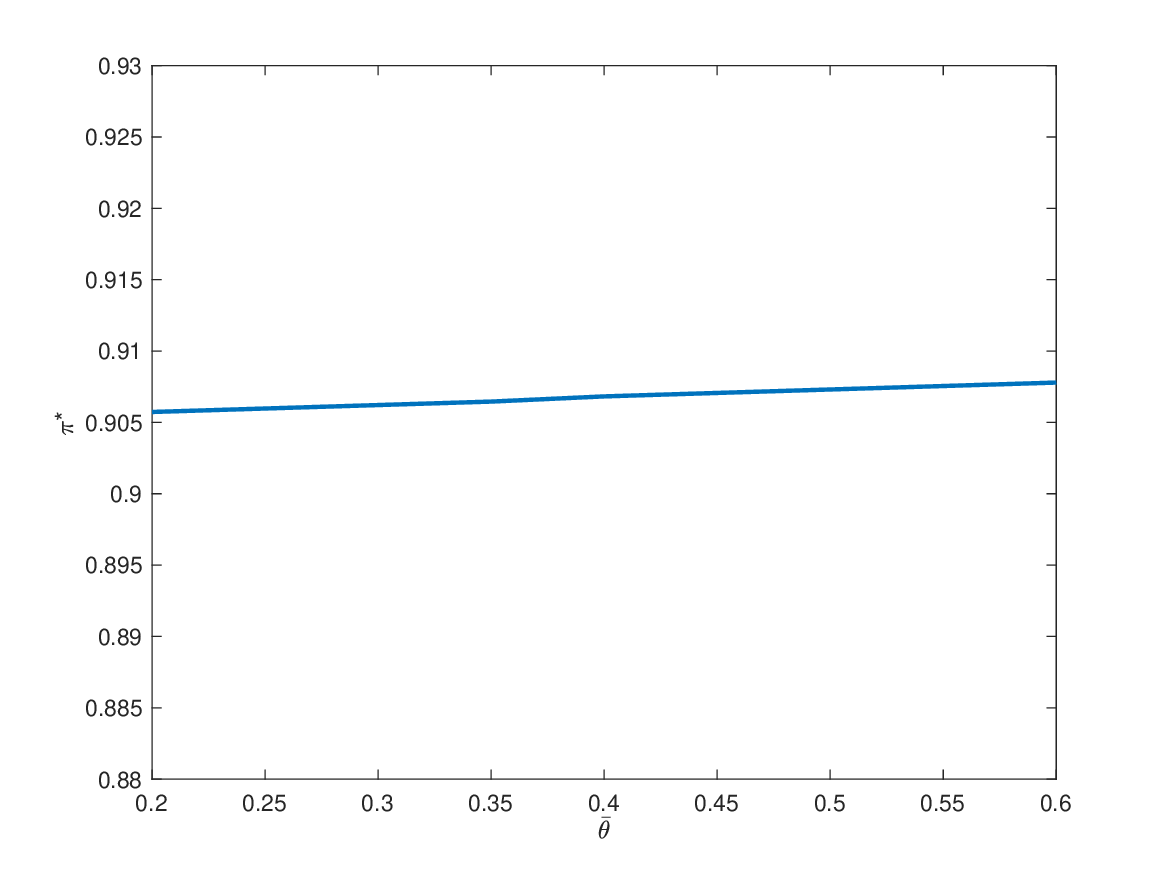}
  \caption{Equilibrium strategy of a fixed agent when varying competition weight of that agent or the network.\\
{\footnotesize{\textit{Notes}. On the left-hand side of Figure \ref{fig1-1}, we increase the competition weight of the fixed agent while keeping the parameters of the network fixed. On the right-hand side of Figure \ref{fig1-1}, we fix the parameters of agent and vary the distribution of the competition weight of the network as $\mathcal{U}[0.05z,0.4+0.05z]$, $z\in (0,8)$. }}}
  \label{fig1-1}
\end{figure}

\begin{figure}[ht]
   \includegraphics[width=7.96cm]{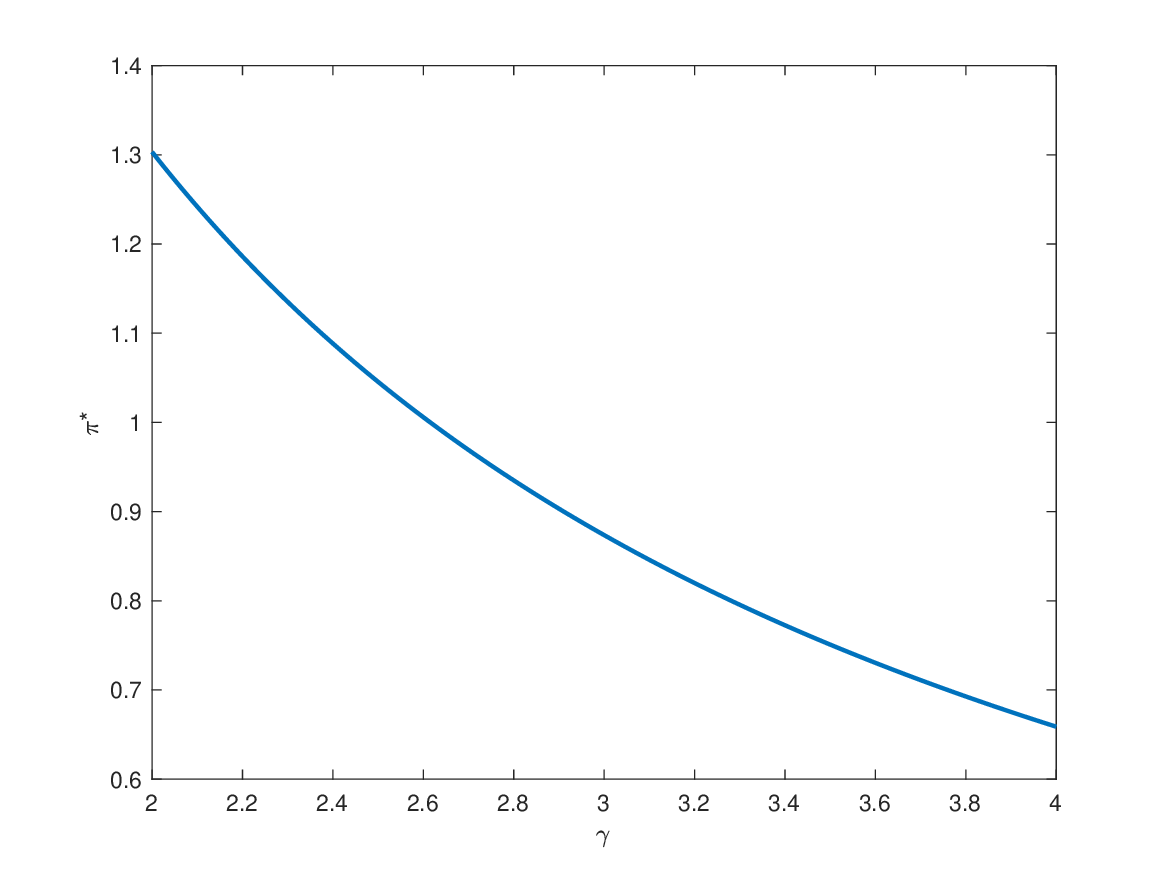}
  \includegraphics[width=7.96cm]{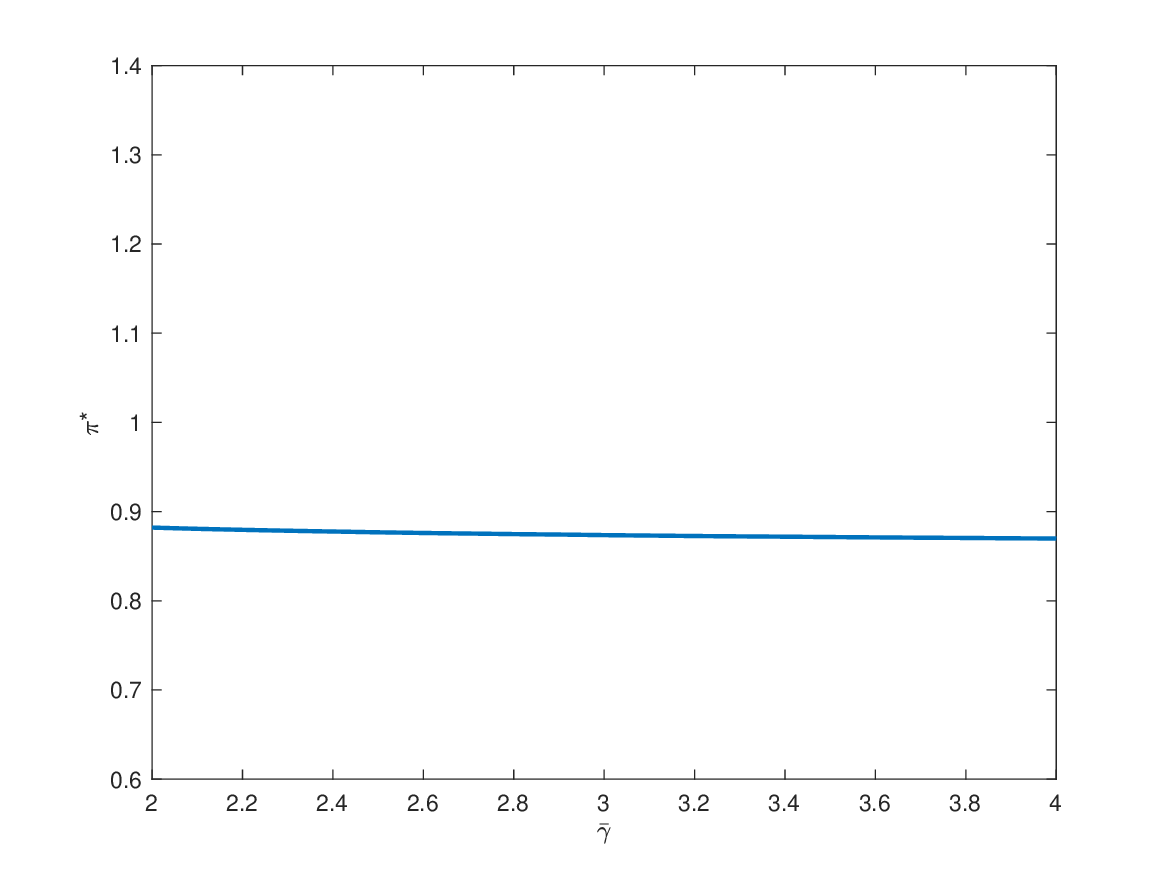}
  \caption{Equilibrium strategy of a fixed agent when varying risk aversion of that agent or the network.}
  \label{fig1-2}
{\footnotesize{\textit{Notes}. On the left-hand side of figure \ref{fig1-2}, we increase the risk aversion of the fixed agent while keeping the parameters of the network fixed. On the right-hand side of figure \ref{fig1-2}, we fix the parameters of agent and vary the distribution of the risk aversion of the network as $\mathcal{U}[1+0.1z,3+0.1z]$, $z\in (0,8)$.}}
\end{figure}


The left-hand side of Figures \ref{fig1-1} and \ref{fig1-2} shows the equilibrium strategy of a given agent when varying the preference parameters of that agent while the right-hand side shows the equilibrium strategy of the agent when varying the preference parameters of the network. 
We make the following observations. 
    

First, the allocation to the risky asset increases both in the competition weight of the agent herself as well as the average competition weight of her network.
The first observation shows that paying attention to the wealth of others leads to more risk taking. 
The second phenomenon is an indirect effect. When the average competition weight within the network increases, agents optimally raise their risky allocations, thereby increasing the magnitude and volatility of aggregate wealth. Since asset returns are positively correlated across investors, a given and fixed agent responds by increasing her risky allocation as well, in order to keep up with the increased volatility of aggregate wealth within the network.


Second, the allocation to the risky asset decreases both in the risk aversion of the agent herself as well as the average risk aversion of her network. 
An agent with higher risk aversion is more inclined to choose a more stable portfolio by reducing their exposure to the risky asset. 
The second phenomenon is again an indirect effect. 
If a given agent is part of a population with higher risk aversion, the peer wealth which the agent intends to outperform is less volatile and less aggressive.
In response, the agent reduces her exposure to the risky asset.

We also observe that a more risk-tolerant agent increases her risky positions at a slower pace, as she becomes increasingly competitive with higher relative performance concerns. 
In contrast, agents with higher risk aversion exhibit a more rapid adjustment in their risky positions, when faced with the same change in the competition weight. Indeed, if we do the numerical approximation for different agents with increasing risk aversion $\gamma=\{2, 2.5, 3, 2.5, 4\}$, and compute their corresponding percentage change in the equilibrium strategy, which represents the relative change between the case when $\theta=0.2$ and $\theta=0.6$, the result is given by $\{1.12\%, 1.39\%, 1.66\%, 1.93\%, 2.20\%\}$. 

An interesting observation is that varying the parameters of a given agent has a much stronger effect than varying the average of the same parameter across the network for both competition weight and risk aversion.  
When an investor's own competition weight increases or risk aversion decreases, holding all else fixed, her optimal strategy adjusts directly and becomes more aggressive. 
On the other hand, when preference parameters of the population change in the same direction, it is the competitors that increase their risky allocations, and the average wealth in the network thus increases on average and becomes more volatile. Since asset returns are positively correlated, the investor correspondingly increases her own risky exposure to keep up with the now higher average wealth in the network. However, this response is attenuated because the competition weight given to the average wealth in the network is smaller than the weight given to one's own wealth. Consequently, changes in network-wide average parameters exert a weaker influence on the optimal strategy than changes in individual parameters.
As we will see in the next subsection, this finding still holds for other model parameters as well.

\subsection{Sensitivity analysis of market parameters}\label{subs:SensitivityMarket}
We next turn to study the monotonicity of equilibrium strategy with respect to features of the financial market, including volatility, expected return, skewness, and the correlation structure of all stocks traded by the population.

\subsubsection{Volatility}

In the binomial setting, we fix the transition probability and the expected return of each stock while varying the price levels for them to represent different dispersions of returns.  The upper and lower price levels are adjusted simultaneously such that the expected return remains the same.

The left-hand side of Figure \ref{fig2} shows the equilibrium strategy of a given agent when varying the stock volatility of that agent while the right-hand side shows the equilibrium strategy of the agent when varying the stock volatility of the network.

As is expected, investors dislike volatility.
Interestingly, a given agent also reduces her exposure to the risky asset when the assets traded by others get more volatile, even though the asset traded by herself remains unchanged. 
The intuition is that when the stocks traded by others are more volatile, they decrease their exposure to stocks. Under the assumption that stocks are positively correlated, the performance of the benchmark becomes smaller on average and less volatile. 
This leads the agent to reduce her own exposure to her risky asset. 

\begin{figure}[H]
\includegraphics[width=7.96cm]{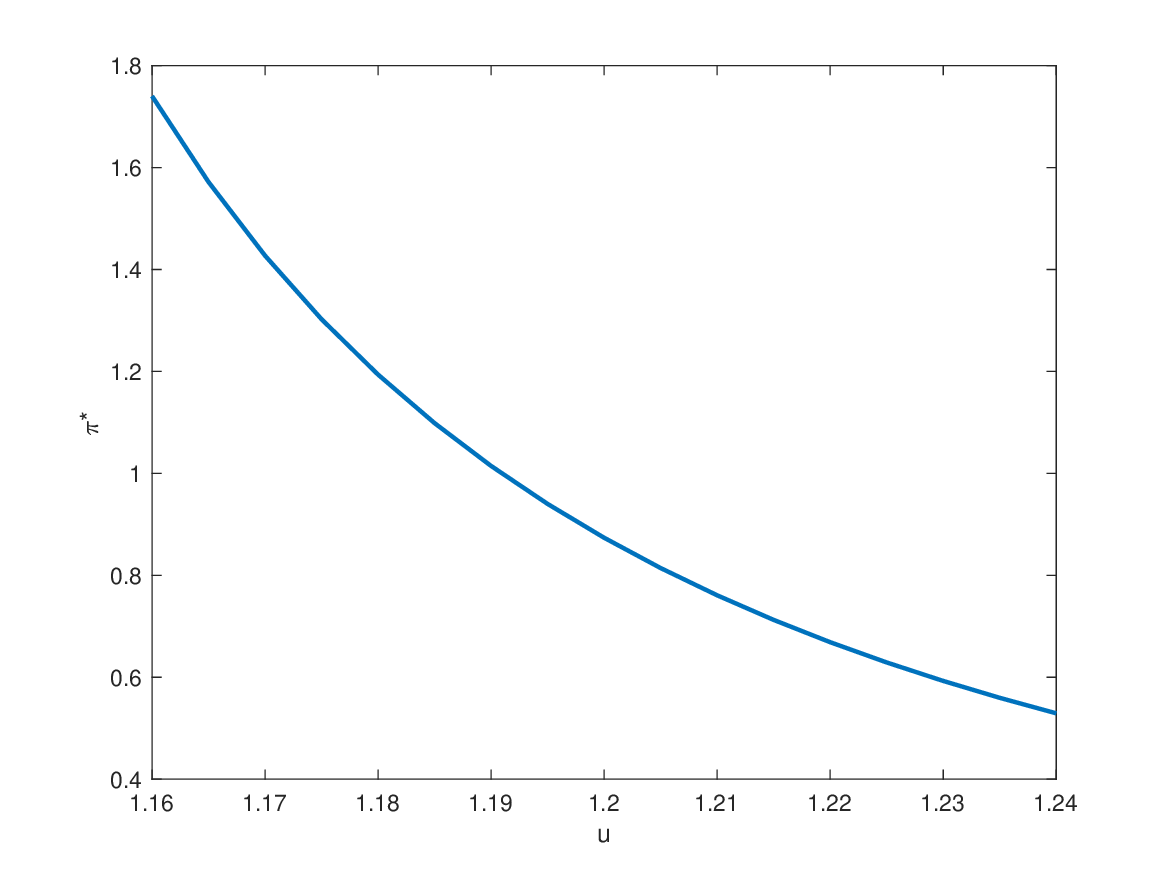}
\includegraphics[width=7.96cm]{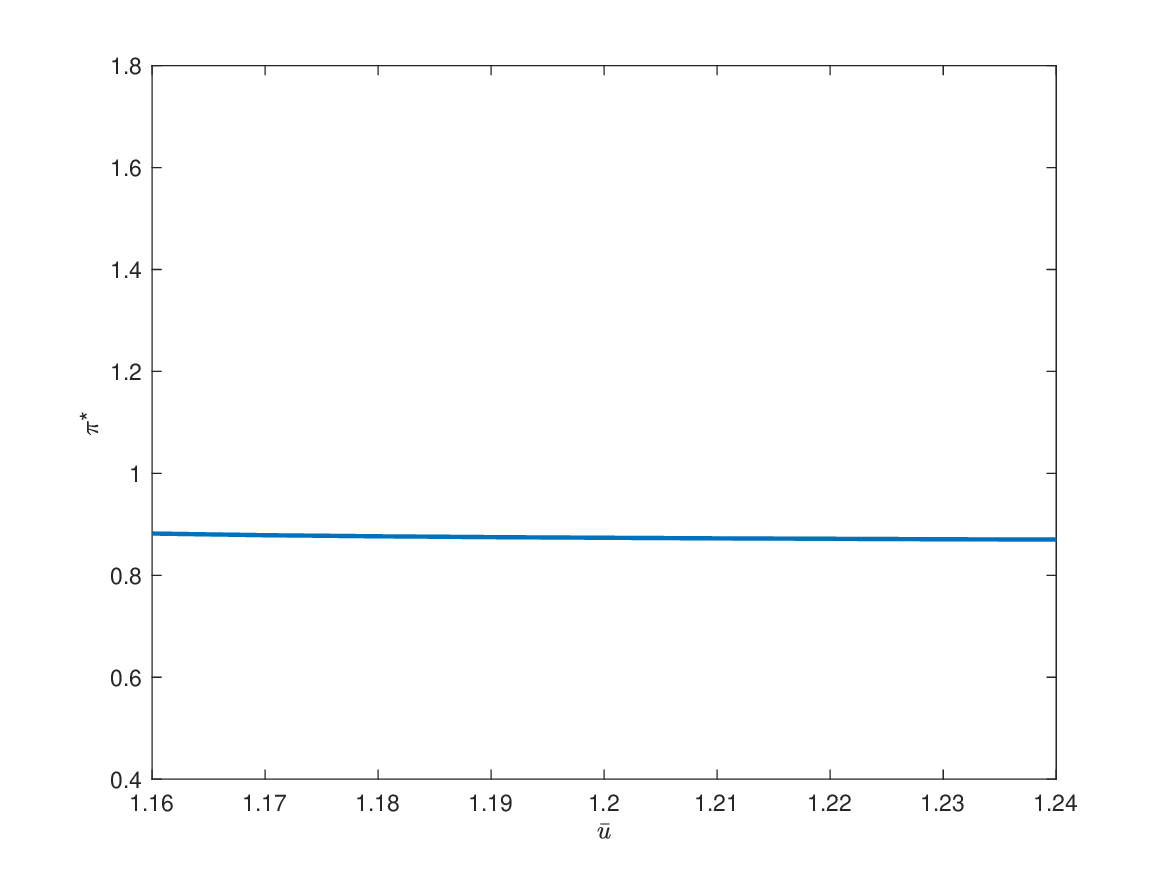}
\caption{Equilibrium strategy of a fixed agent when varying stock volatility of that agent or the population.\\
{\footnotesize{\textit{Notes}. On the left-hand side of Figure \ref{fig2}, we increase the stock volatility of a fixed agent while keeping the parameters of the population fixed. 
On the right-hand side of Figure \ref{fig2}, we fix the parameters of the agent and vary the distribution of the stock volatility of the population. 
This is done by varying the upper price level $u$ as $\mathcal{U}[1.12+0.01z,1.2+0.01z]$, $z\in (0,8)$ with lower price level $d$ adjusted accordingly such that the expected stock return remains unchanged.
}}
}
\label{fig2}
\end{figure}
\subsubsection{Expected Return}

Next, we investigate how the investment behaviour of the agent is influenced by the stock's expected return traded in the market.

The left-hand side of Figure \ref{fig3} shows the equilibrium strategy of a given agent when varying the stock's expected return of that agent while the right-hand side shows the equilibrium strategy of the agent when varying the stock's expected return of the network.

If a given agent trades a more attractive stock, she naturally increases the allocation to that stock. 
Relative competition concerns cause her to also increase the investment to the risky asset when the stocks traded by other investors become more attractive. 
This is an indirect effect of competitors increasing their allocation to now more attractive stocks resulting in higher average levels of wealth in the economy. Therefore, an agent who cares about her wealth in relation to that of others, chooses to allocate more in the risky stocks even though the characteristics of this own stock remain unchanged.

\begin{figure}[H]
\includegraphics[width=7.96cm]{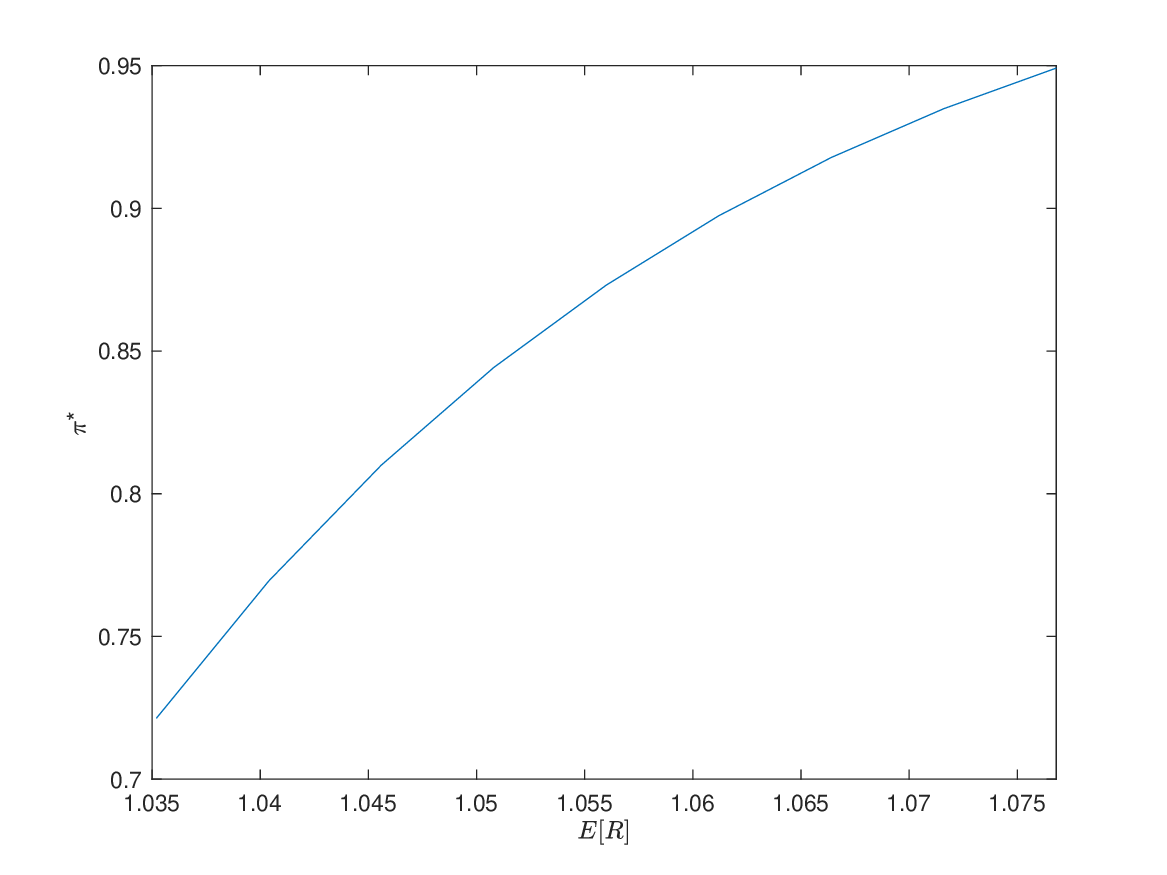}
\includegraphics[width=7.96cm]{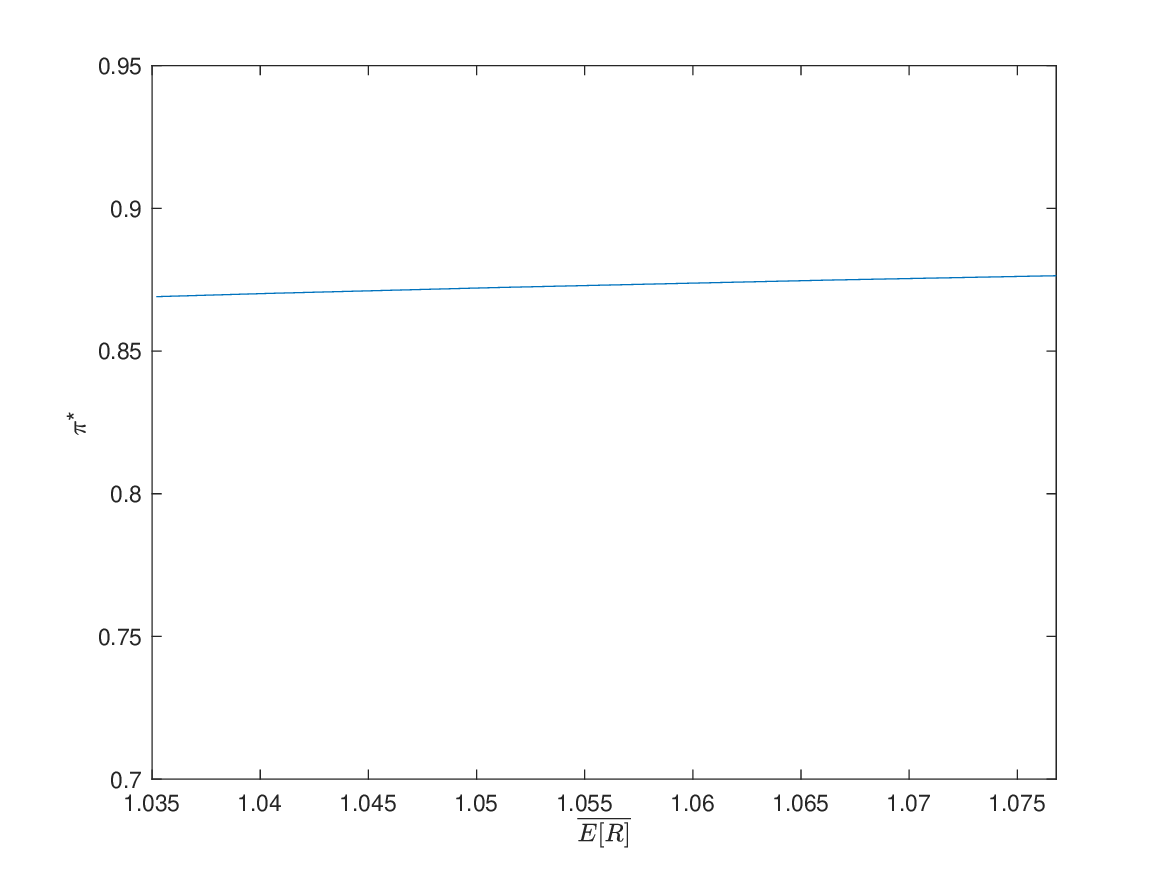}
\caption{Equilibrium strategy of a fixed agent when varying stock return of that agent or the network.\\
{\footnotesize{\textit{Notes}. On the left-hand side of Figure \ref{fig3}, we increase the expected stock return of the fixed agent while keeping the parameters of the network fixed. On the right-hand side of Figure \ref{fig2}, we fix the parameters of agent and vary the distribution of the expected stock return of the network. 
This is obtained by varying the upper price level $u$ as $\mathcal{U}[1.12+0.01z,1.2+0.01z]$, $z\in (0,8)$, while keeping the lower price level $d$ fixed.}}
}
  \label{fig3}
\end{figure}


\subsubsection{Skewness}

A key advantage of working with a discrete market over an Itô-diffusion model is the ability to vary skewness independently and examine its effects on investment behaviour.
While skewness may change with parameters in continuous-time models, it is generally not an independent modeling dimension. In diffusion settings such as the Black–Scholes model, the return distribution is fully determined by the drift and volatility parameters. As a result, changes in skewness are inseparable from changes in other distributional features, making it difficult to isolate the economic role of skewness. Moreover, asset prices are lognormally distributed and thus exhibit strictly positive skewness; negative skewness cannot be studied within this framework.
By contrast, discrete-time models retain additional degrees of freedom that allow the support of returns (including extreme outcomes) and their transition probabilities to be specified separately. This flexibility makes it possible to hold key quantities of interest, such as the mean return or downside exposure, fixed. For this reason, discrete-time models provide a more suitable framework to isolate and study the impact of skewness.
A binomial stock is positively skewed if the mean is greater than the median, meaning the stock price has a higher probability of going down. Conversely, the price of a binomial stock with a negatively skewed distribution has a higher probability of going up. This allows us to study the effects of lottery-type stocks with large chances of a small loss and small chances of a large gain that are considered in the market.




\begin{figure}[H]
   \includegraphics[width=7.96cm]{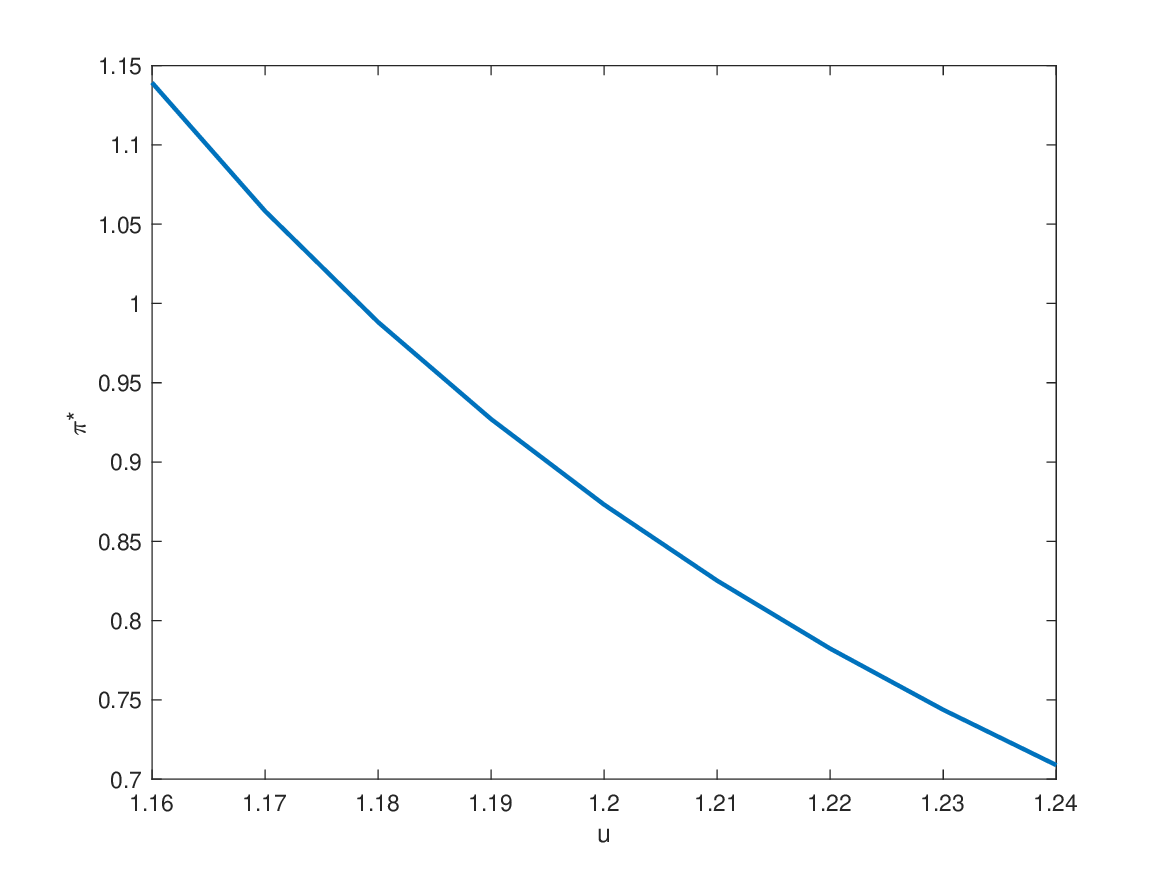}
  \includegraphics[width=7.96cm]{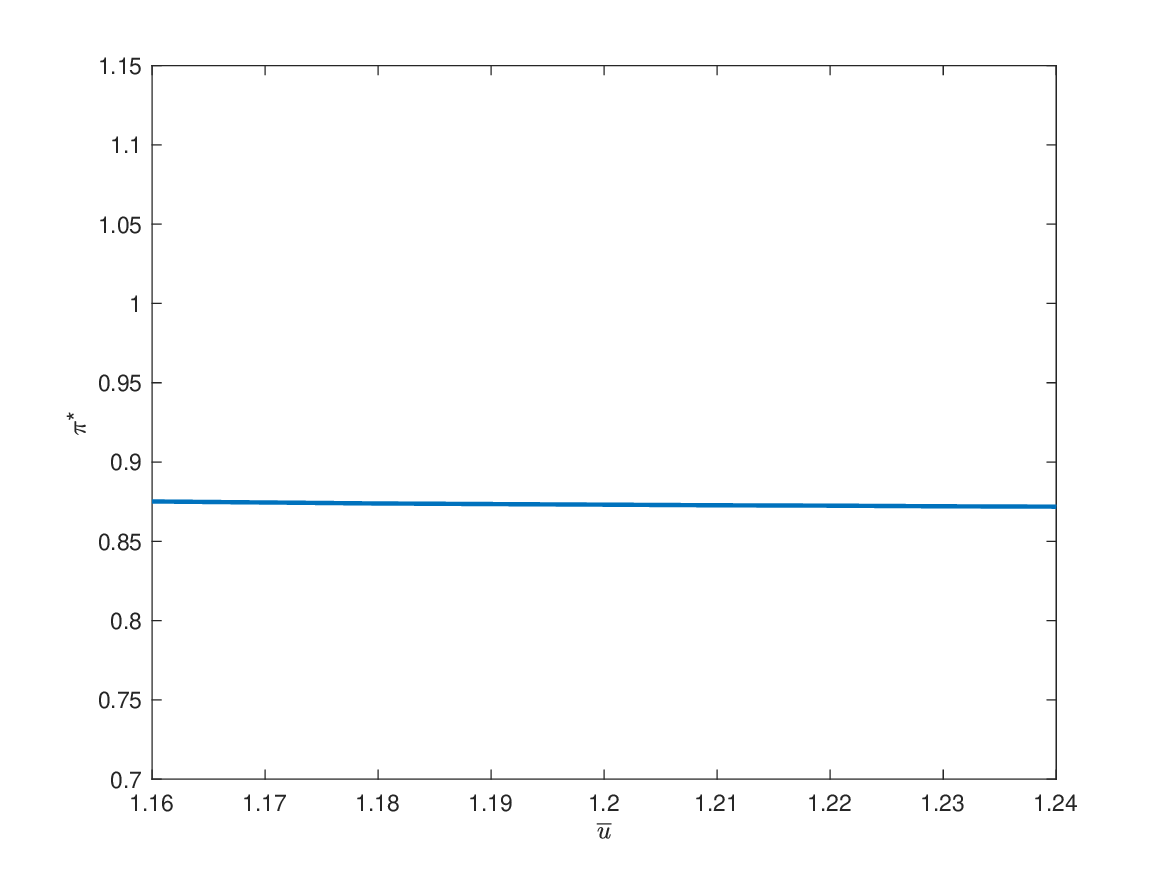}\\
  \includegraphics[width=7.96cm]{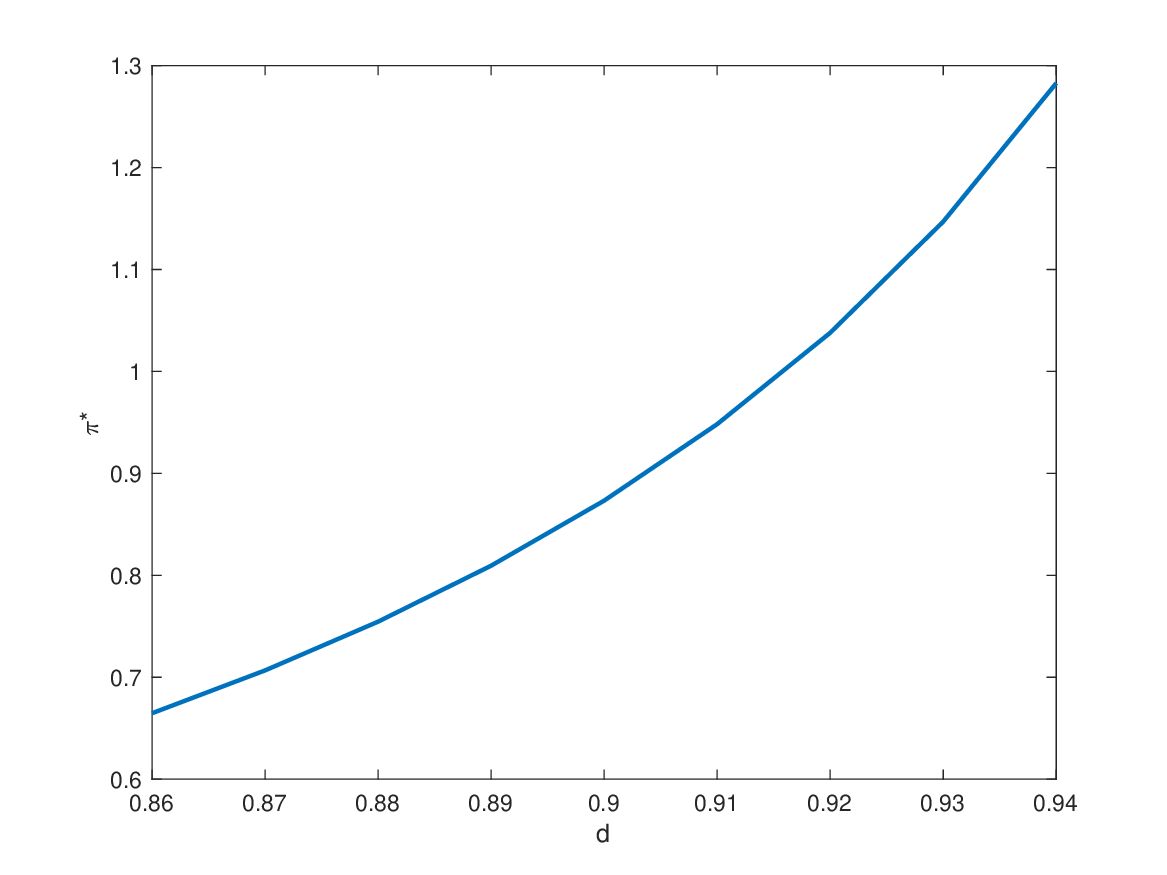}
  \includegraphics[width=7.96cm]{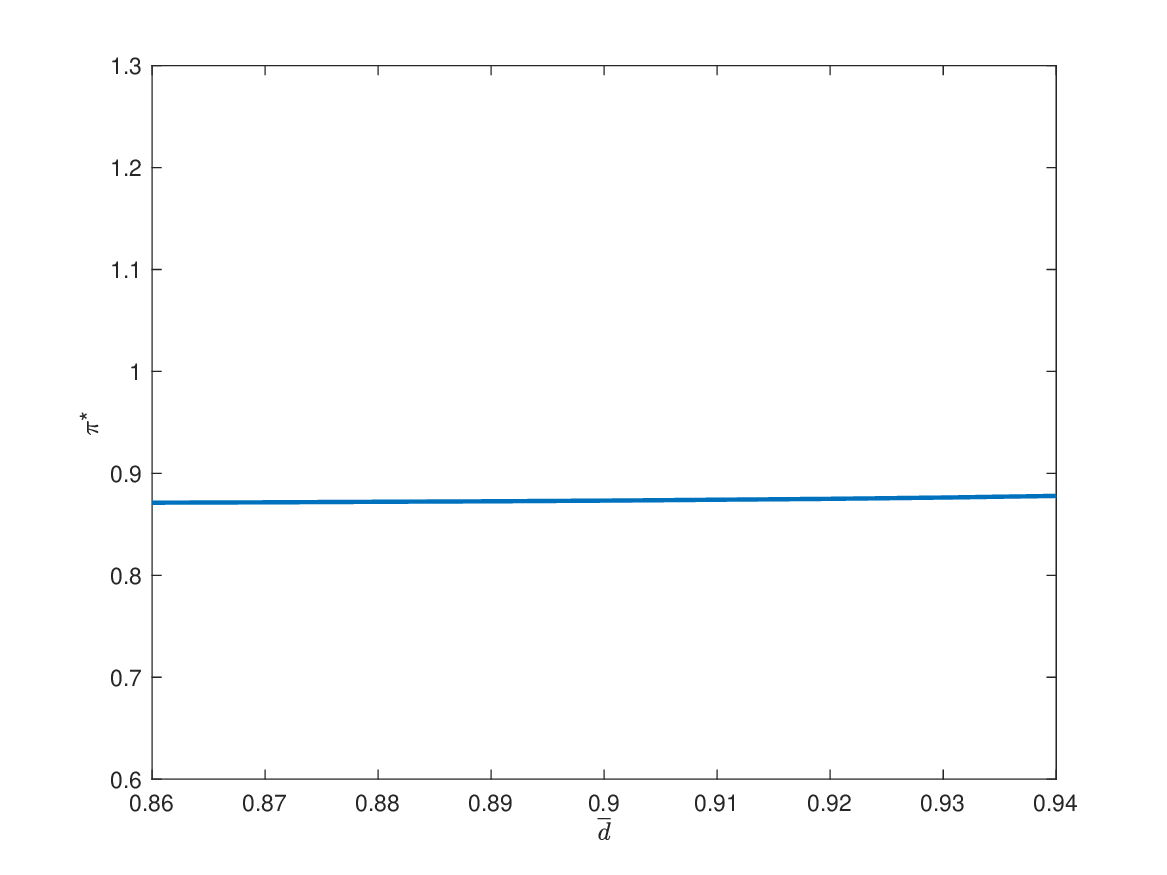}
  \caption{Equilibrium strategy of a fixed agent when varying stock skewness of that agent or the network.\\
  {\footnotesize{\textit{Notes}. On the left-hand side of Figure \ref{fig4},  we increase the positive or negative skewness of the fixed agent by varying price level $u$ or $d$ and the transition probability
  while keeping other parameters and the expected stock return of the network fixed.
  On the right-hand side of Figure \ref{fig4}, we fix the parameters of the agent and vary the distribution of the positive or negative skewness of the network. 
  This is obtained by varying the upper price level $u$ as $\mathcal{U}[1.12+0.01z,1.2+0.01z]$, $z\in (0,8)$, respectively the lower price level $d$ as $\mathcal{U}[0.9-0.01z,0.98-0.01z]$, $z\in (0,8)$, with the transition probability adjusted such that the expected stock return remains constant. 
}}}
  \label{fig4}
\end{figure}

The left-hand side of Figure \ref{fig4} shows the equilibrium strategy of a given agent when varying the skewness of the stock traded by that agent while the right-hand side shows the equilibrium strategy of the agent when varying the distribution of skewness of the stocks in the network. 
When the stock is positively skewed, we vary skewness by varying the upper return level and the probability of an increase in stock price simultaneously such that the expected return and lower return level remain constant. 
Thus, greater $u$ and $\bar{u}$ correspond to higher positive skewness of the stocks. 
In the second row, we consider stocks with negative skewness.
We vary the lower return level and the probability of going down simultaneously such that the expected return and upper return level remain constant.
Thus, lower $d$ and $\bar{d}$ represent more negatively skewed stocks.

Empirical research, for example, \cite{arditti1967risk}, \cite{brunnermeier2007optimal}, \cite{scott1980direction} and \cite{ebert2021skewness}, finds that investors prefer assets with positively skewed return distribution and dislike assets with negatively skewed return distribution. 
In contrast, we find that an investor with PRFPPs does neither like positively nor negatively skewed assets: Strategies are decreasing in the absolute value of skewness.
Furthermore, for a fixed agent, as the stocks traded by competitors become more skewed, the agent chooses to decrease her investment in the risky asset as an indirect effect.

\subsubsection{Correlation}
This subsection discusses the impact of the correlation between the stocks held by different agents on the sensitivity of the MFE with respect to the degree of relative performance concerns of the whole population for any given trading period $[t-1,t)$.


Recall that we interpret the mean field game as an economy comprising an infinite number of agents, each having type vectors drawn as i.i.d. copies of the type vector of the representative agent, given the non-traded stochastic factor.
To examine the impact of correlations between different stocks, we fix the unconditional probability and vary only the dependence structure given non-traded stochastic factor by changing the conditional probabilities. 
For simplicity, we assume that there is no randomness in the type vector of the representative agent arising from $\P[B_t=1\vert \mathcal{F}_{t-1}^{MF}]=p_t$ and price levels $u_t, d_t$. 
The unconditional (in terms of non-traded stochastic factor) probability of stock price going up, price levels, along with the variance of any agent indexed by $i$ in the MFG, are assumed to 
be $\mathcal{F}_{t-1}^{RC}$-measurable, where $\mathbb{F}^{RC} = (\mathcal{F}_t^{RC})_{t \in \mathbb{N}_0}$ represents the natural filtration generated by $(p_{m}^{cn})_{m=1}^{t+1}$, $(\xi_{m}^{cn})_{m=1}^{t}$ and $(B_{(i),m})_{m=1}^{t}, i \in \mathbb{N}$.

Covariance is the only characteristic of the stock return we vary, and this is done by varying the conditional probabilities $p_{(i),t}^1$ and $p_{(i),t}^0$ of any agent $i$.
By the law of total probability, we have
\begin{align*}
    p_t=\P[B_{(i),t} = 1,\xi_t^{cn}=1\vert \mathcal{F}_{t-1}^{MF}]+\P[B_{(i),t} = 1,\xi_t^{cn}=0\vert \mathcal{F}_{t-1}^{MF}]=(1-p_t^{cn})p_{(i),t}^0+p_t^{cn}p_{(i),t}^1.
\end{align*}
We point out that $p_{(i),t}^1$ and $p_{(i),t}^0$ should be adjusted in a certain manner such that $p_t$ remains $\mathcal{F}_{t-1}^{RC}$-measurable. 

We consider two agents of the population, $i \neq j$, with type vectors $\psi_{(i)}$ and $\psi_{(j)}$ respectively. Considering that the covariance refers to the degree to which any two agents are dependent, we herein need to consider an enlarged filtration $\mathbb{G}^{MF} = (\mathcal{G}_t^{MF})_{t \in \mathbb{N}_0}$, with $\mathcal{G}_t^{MF}$ being generated by $(p^{cn}_{m})_{m=1}^{t+1}$, $(\xi_{m}^{cn})_{m=1}^{t}$, $\eta_{(i),0}$, $(\zeta_{(i),m})_{m=1}^{t+1}$ and $(B_{(i),m})_{m=1}^{t}, i\in \mathbb{N}$. \textit{Assume that $B_{(i),t}$ is conditionally independent of $B_{(j),t}$ given $\mathcal{G}_{t-1}^{MF}\vee \sigma(\xi_t^{cn})$.}
Define the covariance of the stock returns over trading period $[t-1,t)$ in this mean field setup by
\begin{align*}
{\rm{cov}}_t^{\psi}:={\rm{cov}}\left[R_{(i),t}, R_{(j),t}\big\vert\mathcal{G}_{t-1}^{MF}\right]=\E[R_{(i),t}R_{(j),t}\vert\mathcal{G}_{t-1}^{MF}]-\E[R_{(i),t}\vert\mathcal{G}_{t-1}^{MF}]\E[R_{(j),t}\vert\mathcal{G}_{t-1}^{MF}].
\end{align*}
The proposition below formulates the covariance structure of the stocks traded by different agents in the group.
\begin{proposition}\label{Prop:Covariance}
For any two agents of the population group in MFG, without generality we assume that they are the $i$th agent and the $j$th agent respectively, the covariance between the returns of their specialized stocks is given by
\begin{align}\label{Formula:Cov}
     {\rm{cov}}_t^{\psi}=(u_t-d_t)^2\left((1-p_t^{cn})p_{(i),t}^0p_{(j),t}^0+p_t^{cn}p_{(i),t}^1p_{(j),t}^1-(p_t)^2\right).
\end{align}
\end{proposition}

Note that when there is no such non-traded stochastic factor affecting the transition probability of each risky asset, i.e., their correlation equals zero, we have $p_{(i),t}^0p_{(j),t}^0=p_{(i),t}^1p_{(j),t}^1=(p_t)^2$, and the expression of covariance in (\ref{Formula:Cov}) is indeed zero as expected.

It can be inferred from (\ref{Formula:Cov}) that the distribution of ${\rm{cov}}_t^{\psi}$ depends on the distribution of $p_{(i),t}^0$ and $p_{(j),t}^0$, which are i.i.d. random variables according to $p_t^{0}$.
With the aim of more explicitly comparing and investigating the investment strategies among different population groups, we consider in our problem setup homogeneous agents for each single network, i.e., $p_{(i),t}^0=p_t^{0}$ and $p_{(i),t}^1=p_t^{1}$ are now  $\mathcal{F}_{t-1}^{RC}$-measurable. Then the stock returns of any two agents from a fixed network have a fixed correlation, and the degree to which any two stocks considered in a MFG move in the same direction can be measured by a single number so that we are able to vary the covariance structure through the whole network.

By Proposition \ref{Prop:Covariance}, the covariance can be expressed as a function of all the model parameters. In the following analysis, we vary  one of them, namely, $p_t^{1}$, while keeping the unconditional probability unchanged by adjusting $p_t^{0}$ accordingly.
This variation allows us to trace out the dependence structure, first between two stocks and then in the MFG in general. To elaborate further, $p_t^{1}$ can be set by starting from the value of $p_t$, which corresponds to a scenario where the non-traded stochastic factor contains no information and every stock is uncorrelated. Subsequently, we increase the value of $p_t^{1}$ until it approaches the limiting value $1$, representing a setup where all stocks are perfectly correlated.

\begin{figure}[H]
\centering
  \includegraphics[width=11cm]{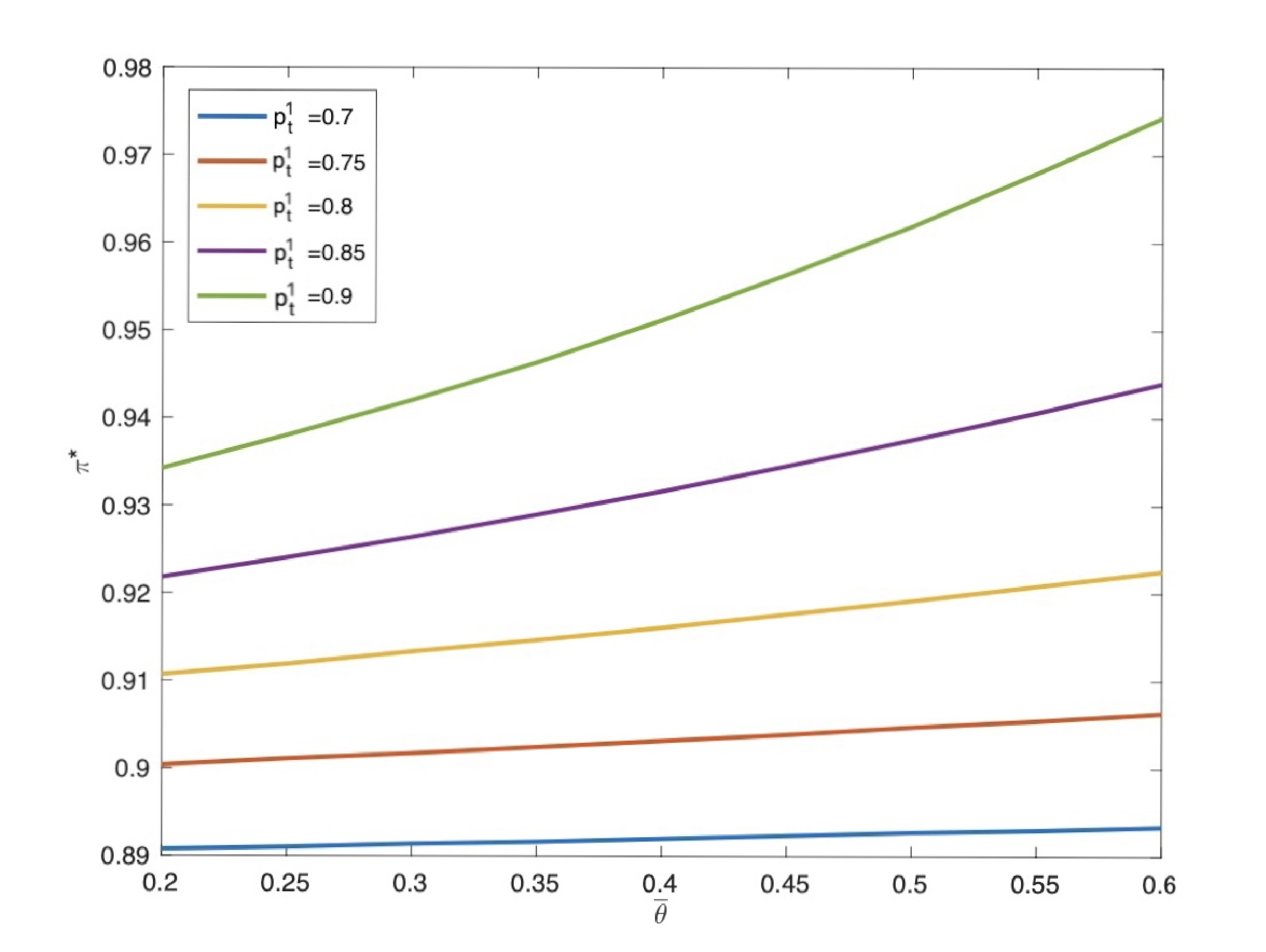}
  \caption{Equilibrium strategy of a fixed agent when varying competition weight of the network with different correlation structure.\\
  {\footnotesize{\textit{Notes}. We fix the parameters of agent and vary the distribution of the competition weight as $\mathcal{U}[0.05z,0.4+0.01z]$, $z\in (0,8)$. Different color refers to different $p_t^{1}$ but constant probability of an upward move.
}}}
  \label{fig5}
\end{figure}

We limit our analysis on a numerical example comparing five distinct scenarios characterized by varying values of $p_t^{1}$.
Specifically, we set $p_t^{1}$ at $0.9, 0.85, 0.8, 0.75, 0.7$.
We then proceed to plot and investigate the sensitivity of the MFE concerning the dependence structure among all agents.
As illustrated in Figure \ref{fig5}, we observe a noteworthy pattern: When the competing agents exhibit higher levels of correlation, the rate at which a given agent's MFE increases with respect to the population competition weight accelerates.

As demonstrated in Proposition \ref{Prop:NoCommonNoise}, when the risky assets are completely uncorrelated, the MFE $\pi^{*}$ remains independent of the average competition weight $\overline{\theta}$ characterizing the population. At the opposite extreme, the effect of relative performance concerns becomes most pronounced when all stocks are fully correlated.

\section{Conclusion}\label{sec:Conclusion}
We studied discrete-time 
predictable forward performance processes under relative performance concerns in an $N$-agent game and a mean field game setting, where both the performance measurement and the model parameters are evaluated and dynamically updated at the beginning of each trading period. 
Our main contribution is the explicit construction of PRFPPs, along with the associated Nash and mean field equilibria under exponential performance criteria. 
We also showed that the equilibrium for the $N$-agent game converges to the mean field equilibrium in distribution.
We then conducted a numerical analysis on the properties of the MFE and showed that
changing an individual agent's parameters has a more pronounced effect than adjusting the network-wide average of the same parameter. 
An agent driven by relative performance concerns increases her stock holdings when competitors demonstrate higher competitiveness, lower risk aversion, or when the competitors' stocks exhibit reduced volatility, increased expected returns, or skewness.

 \section*{Acknowledgments}
The authors are grateful to Larbi Alili, Guanxing Fu, Martin Herdergen, Yinlian Zeng, Chao Zhou, and two anonymous referees for their helpful comments and suggestions.
This work was presented at NUS and CityU HK Quantitative Finance \& Mean Field Game Seminar Series 2023, Recent Advances on Quantitative Finance 2023, The 5th (2023) Academic Conference on Quantitative Finance and Insurance of the Chinese Society of Optimization, Overall Planning and Economical Mathematics.
The authors thank the participants for their comments.
This research was supported by the Gillmore Centre for Financial Technology. 
Author order is alphabetical. All authors are co-first authors of this paper. 

\section*{Appendix. Discussion of the two extreme cases}
\renewcommand{\thesubsection}{\Alph{subsection}}
In this appendix, we study the Nash and mean field equilibria in two extreme cases: (1) when each risky asset is identical almost surely, and (2) when there is no non-traded stochastic factor.
These two cases will yield more tractable and interpretable solutions.

\subsection{Identical stock case}

Some of the effects of competition and risk aversion are more transparent if all agents trade exactly the \textit{same stock}, characterized by transition probability $p_t$ and price levels $u_t$ and $d_t$ over the trading period $[t-1, t)$. 
This setting can be studied within our model by considering stocks that are identical almost surely.
Let $q_t=(1-d_t)/(u_t-d_t)$, $t\in\mathbb{N}$.
In the following proposition, we describe the resulting forward Nash equilibrium of the $N$-agent game and MFE explicitly.

\begin{proposition}\label{Prop:NAgentGameSingleStock}
If each risky asset is identical almost surely,
the forward Nash equilibrium of the $N$-agent game and the MFE in the game of infinite agents over trading period $[t-1, t)$, $t\in\mathbb{N}$, are respectively given by
\begin{align}\label{NagentGameSingleStock}
    \pi_{(i),t}^{*}=\frac{1}{(u_t-d_t)}\left(\ln\frac{p_t(1-q_t)}{(1-p_t)q_t}\right)\left(\frac{1}{\gamma_{(i)}}+\frac{\theta_{(i)}\sum\limits_{k=1}^N\frac{1}{\gamma_{(k)}}}{N\left(1-\frac{\sum\nolimits_{k=1}^N\theta_{(k)}}{N}\right)}\right),
\end{align}
and
\begin{equation}\label{MFESingleStock}
\begin{aligned}
    \pi_t^*=\frac{1}{u_t-d_t}\left(\ln\frac{p_t(1-q_t)}{(1-p_t)q_t}\right)\left(\frac{1}{\gamma}+\frac{\theta\overline{1/\gamma}}{(1-\overline{\theta})}\right).
\end{aligned}
\end{equation}
where $\overline{1/\gamma}=\lim\limits_{N\rightarrow +\infty}\frac{1}{N}\sum\limits_{k=1}^{N}\frac{1}{\gamma_{(k)}}$ and $\overline{\theta}=\lim\limits_{N\rightarrow +\infty}\frac{1}{N}\sum\limits_{i=1}^{N}\theta_{(i)}$. Furthermore, fix $\psi_t=\widetilde{\psi}_{(i),t}$, then $\pi_{(i),t}^{*}$ in (\ref{NagentGameSingleStock}) converges to (\ref{MFESingleStock}) in distribution as $N\rightarrow \infty$. 
\end{proposition}

The Nash equilibrium is of Merton type but with a modified risk tolerance, which depends on the individual and average risk aversion and competition parameters.
Furthermore, we observe that in the identical stock case, the limiting optimal strategy converges to the MFE in distribution.
With (\ref{MFESingleStock}), we are able to study the monotonicity of the MFE $\pi_t^{*}$ for the identical stock case analytically. 
First, we know that when the market offers positive expected excess return, then $\ln\frac{p_t(1-q_t)}{(1-p_t)q_t}>0$. 
Indeed, when $\E[R_t-1\vert \mathcal{F}_{t-1}^{MF}]>0$, we have $p_t(u_t-1)+(1-p_t)(d_t-1)>0$, and hence $\frac{p_t(1-q_t)}{(1-p_t)q_t}=\frac{p_t(u_t-1)}{(1-p_t)(1-d_t)}>1$ naturally follows. 
Second, taking first-order derivatives with respect to $\theta$, $\gamma$, $\overline{\theta}$, and $\overline{1/\gamma}$ yields
\begin{align*}
    &(\pi_t^*)^{'}(\theta)=\frac{1}{u_t-d_t}\frac{\overline{1/\gamma}}{(1-\overline{\theta})}\ln\frac{p_t(1-q_t)}{(1-p_t)q_t}>0,
    \\&(\pi_t^*)^{'}(\gamma)=\frac{-1}{u_t-d_t}\frac{1}{\gamma^2}\ln\frac{p_t(1-q_t)}{(1-p_t)q_t}<0,
    \\&(\pi_t^*)^{'}(\overline{\theta})=\frac{1}{u_t-d_t}\frac{\theta \overline{1/\gamma}}{(1-\overline{\theta})^2}\ln\frac{p_t(1-q_t)}{(1-p_t)q_t}>0,
    \\&(\pi_t^*)^{'}(\overline{1/\gamma})=\frac{1}{u_t-d_t}\frac{\theta}{(1-\overline{\theta})}\ln\frac{p_t(1-q_t)}{(1-p_t)q_t}>0.
\end{align*}
From the above, we conclude that $\pi_t^{*}$ is increasing in $\theta$, $\overline{\theta}$, and $\overline{1/\gamma}$, while decreasing in $\gamma$.

\subsection{No non-traded stochastic factor}
Next, we study the other extreme case where there is no non-traded stochastic factor, or equivalently, all the risky assets are mutually independent.
\begin{proposition}\label{Prop:NoCommonNoise}
If the returns of all risky assets are mutually independent, the forward Nash equilibrium of the $N$-agent game and the MFE in the game of infinite agents over trading period $[t-1, t)$, $t\in\mathbb{N}$, are respectively given by
\begin{align}\label{NagentGameNoCN}
    \pi^{*}_{(i),t}=\frac{1}{\gamma_{(i)}(1-\frac{\theta_{(i)}}{N})(u_{(i),t}-d_{(i),t})}\ln\frac{p_{(i),t}(1-q_{(i),t})}{(1-p_{(i),t})q_{(i),t}},
\end{align}
and 
\begin{align}\label{MFENoCN}
    \pi_t^{*}=\frac{1}{\gamma(u_t-d_t)}\ln\frac{p_t(1-q_t)}{(1-p_t)q_t}.
\end{align}
Furthermore, fix $\psi_t=\widetilde{\psi}_{(i),t}$, then $\pi_{(i),t}^{*}$ in (\ref{NagentGameNoCN}) converges to (\ref{MFENoCN}) in distribution as $N\rightarrow \infty$. 
\end{proposition}

In the $N$-agent game, if all individual stocks are mutually independent, the equilibrium can be explicitly solved, and we observe that the limiting optimal strategy also converges in distribution to the MFE. When market offers positive excess expected return, the forward Nash equilibrium is larger than the MFE, and decreases monotonically to the MFE as the population size grows.
In this scenario, all agents act independently, as if there were no competition, but with a modified parameter of risk aversion $\gamma_{(i)}(1-\frac{\theta_{(i)}}{N})$. 
Notably, they do not take into account the performance of their competitors, and
it is interesting to observe that even though the optimal strategy $\pi^{*}_{(i),t}$ does not depend on the strategies of the other agents, agent $i$ is still influenced by the competition weights of the whole network through this modified parameter. 
We also find that the MFE is the same as the classical Merton portfolio in the absence of relative performance concerns as obtained in 
cf. \cite[Theorem 3] {musiela2016forward} 
for the single-agent forward exponential preferences.

\section*{Appendix. Proofs}
\renewcommand{\thesubsection}{\Alph{subsection}}

Proofs of Proposition \ref{Prop:NAgentGameSingleStock} and Proposition \ref{Prop:NoCommonNoise} follow by direct computation and are akin to the proofs of previous theorems, therefore, they are omitted herein.

\subsection{Proof of Proposition \ref{Prop:MFGAverageWealth}}

By the conditional independence among $\psi_{(k),t}$ and $(B_{(k),j})_{j=1}^{t}$
across different $k$, and considering that $\psi_{(k),t}$ follows the identical distribution as $\psi_t$, $\pi_{(k),t}=\pi_t\left(\psi_{(k),t}, (p^{cn}_{j})_{j=1}^{t}, (\xi_{j}^{cn})_{j=1}^{t-1}, (B_{(k),j})_{j=1}^{t-1}\right)$ then follows the identical distribution as $\pi_t$ and is conditionally independent across different $k$ given $\mathcal{F}_{t}^{CN}$.
Furthermore, it follows that 
$\pi_{(k),t}\left(u_{(k),t} B_{(k),t} + d_{(k),t} (1-B_{(k),t})-1\right)$, as a random function of $\psi_{(k),t}, (p^{cn}_{j})_{j=1}^{t},$ $(\xi_{j}^{cn})_{j=1}^{t-1}$ and $(B_{(k),j})_{j=1}^{t}$, is conditionally independent across different $k$ and identically distributed with $\pi_t\left(u_t B_t + d_t (1-B_t)-1\right)$ given $\mathcal{F}_{t}^{CN}$.
By the law of large numbers, the sample mean converges to the distribution mean in probability.

We aim to show by mathematical induction that $\frac{1}{N}\sum\limits_{i=1}^{N}X_{(i),t}\stackrel{\text{p}}{\rightarrow}\E\left[X_t\vert\mathcal{F}_{t}^{CN}\right]$ for any $t\in\mathbb{N}$.
First, when $t=1$, the statement naturally holds.
In the general inductive step, we show that if
$\frac{1}{N}\sum\limits_{i=1}^{N}X_{(i),t-1}\stackrel{\text{p}}{\rightarrow}\E\left[X_{t-1}\vert\mathcal{F}_{t-1}^{CN}\right]$, then the statement is also true for $\frac{1}{N}\sum\limits_{i=1}^{N}X_{(i),t}$.

Indeed, the average wealth can be computed as
\begin{equation}\label{AverageWealth}
\begin{aligned}
    &\frac{1}{N}\sum\limits_{i=1}^N{X}_{(i),t}\\  =& \frac{1}{N}\left( \sum\limits_{i=1}^NX_{(i),t-1}+\sum\limits_{i=1}^N\pi_{(i),t}(R_{(i),t}-1)\right)
    \\ =  &\frac{1}{N}\left( \sum\limits_{i=1}^NX_{(i),t-1}+\sum\limits_{i=1}^N\pi_{(i),t}\left(u_{(i),t} B_{(i),t} + d_{(i),t} (1-B_{(i),t})-1\right)\right)
    \\
     \stackrel{\text{p}}{\rightarrow} &\E\left[X_{t-1}\vert\mathcal{F}_{t-1}^{CN}\right]+\E\left[\pi_t\left(u_t B_t + d_t (1-B_t)-1\right)\vert \mathcal{F}_{t}^{CN}\right]\\=&\E\left[X_{t-1}+\pi_t(R_t-1)\vert\mathcal{F}_{t}^{CN}\right]\\=&\E\left[X_t\vert\mathcal{F}_{t}^{CN}\right],
\end{aligned}
\end{equation}
where the convergence holds by the inductive hypothesis and the law of large numbers as argued above.
This proves the claim and thus shows the convergence in probability between $\frac{1}{N}\sum\limits_{i=1}^N{X}_{(i),t}$ and $\E\left[X_t\vert\mathcal{F}_{t}^{CN}\right]$ as $N$ tends to infinity. 

Since it can be further computed by the tower property of conditional expectation and the $\mathcal{F}_{t-1}^{MF}$-measurability of $\pi_t$ that 
\begin{equation}
\begin{aligned}
    &\E\left[\pi_t\left(u_t B_t + d_t (1-B_t)-1\right)\vert \mathcal{F}_{t}^{CN}\right]\mathbbm{1}_{\{\xi_t^{cn}=1\}}
     \\=&\E\left[\E\left[\pi_t\left(u_t B_t + d_t (1-B_t)-1\right)\vert \mathcal{F}_{t-1}^{MF}\vee\sigma(\xi_t^{cn})\right]\big\vert \mathcal{F}_{t}^{CN}\right]\mathbbm{1}_{\{\xi_t^{cn}=1\}}\\=&\E\left[\pi_t\E\left[\left(u_t B_t + d_t (1-B_t)-1\right)\vert \mathcal{F}_{t-1}^{MF}\vee\sigma(\xi_t^{cn})\right]\big\vert \mathcal{F}_{t}^{CN}\right]\mathbbm{1}_{\{\xi_t^{cn}=1\}} 
    \\= & \E\left[\pi_t\Delta_t^{1}\vert \mathcal{F}_{t-1}^{CN}\right]\mathbbm{1}_{\{\xi_t^{cn}=1\}},
\end{aligned}
\end{equation}
where the final line holds because $\pi_t\Delta_t^{1}$ is independent of $\sigma(\xi_t^{cn})$.
Analogously,
\begin{equation}
\begin{aligned}
    &\E\left[\pi_t\left(u_t B_t + d_t (1-B_t)-1\right)\vert \mathcal{F}_{t}^{CN}\right]\mathbbm{1}_{\{\xi_t^{cn}=0\}}
     =\E\left[\pi_t\Delta_t^{0}\vert \mathcal{F}_{t-1}^{CN}\right]\mathbbm{1}_{\{\xi_t^{cn}=0\}}.
\end{aligned}
\end{equation}
It then follows that 
\begin{align*}
\E\left[X_t\vert\mathcal{F}_{t}^{CN}\right]=\E\left[X_{t-1}\vert\mathcal{F}_{t-1}^{CN}\right]+\E\left[\pi_t\Delta_t^{1}\vert \mathcal{F}_{t-1}^{CN}\right]\mathbbm{1}_{\{\xi_t^{cn}=1\}}+\E\left[\pi_t\Delta_t^{0}\vert \mathcal{F}_{t-1}^{CN}\right]\mathbbm{1}_{\{\xi_t^{cn}=0\}}.
\end{align*}
\qed

\subsection{Proof of Theorem \ref{Thm:RelativeForwardNAgents}}
We first present the following auxiliary result for the single-period backward problem.
\begin{lemma}\label{LemmaSup}
Let $i \in \{1, \dots, N\}$ and competitor policies $\pi_{(j),t}, j\neq i$ be given. Then,
\begin{align}\label{SupremumEqn}
    \sup_{\pi_{(i),t} }\E\left[-e^{-\gamma_{(i)}\left(\left(1-\frac{\theta_{(i)}}{N}\right)\pi_{(i),t}(R_{(i),t}-1)-\frac{\theta_{(i)}}{N}\sum\limits_{j\neq i}\pi_{(j),t}(R_{(j),t}-1)\right)}\Bigg\vert \mathcal{F}_{t-1}\right]=-(A_{(i),t}^{1}B_{(i),t}^{1}+A_{(i),t}^{2}B_{(i),t}^{2}).
\end{align}
\end{lemma}
\begin{proof}
Define the function $f$ by
\begin{align*}
    f(\pi_{(i),t})&=\E\left[-e^{-\gamma_{(i)}\left((1-\frac{\theta_{(i)}}{N})\pi_{(i),t}(R_{(i),t}-1)-\frac{\theta_{(i)}}{N}\sum_{j\neq i}\pi_{(j),t}(R_{(j),t}-1)\right)}\big\vert \mathcal{F}_{t-1}\right]\\&=\E\left[\E\left[-e^{-\gamma_{(i)}\left((1-\frac{\theta_{(i)}}{N})\pi_{(i),t}(R_{(i),t}-1)-\frac{\theta_{(i)}}{N}\sum_{j\neq i}\pi_{(j),t}(R_{(j),t}-1)\right)}\big\vert \mathcal{F}_{t-1}\vee \sigma(\xi_t^{cn})\right]\bigg\vert \mathcal{F}_{t-1}\right]\\&=\E\bigg[\E\left[-e^{-\gamma_{(i)}(1-\frac{\theta_{(i)}}{N})\pi_{(i),t}(R_{(i),t}-1)}\big\vert \mathcal{F}_{t-1}\vee \sigma(\xi_t^{cn})\right]\\&\quad \qquad \times \prod \limits_{j\neq i}\E\left[e^{\gamma_{(i)}\frac{\theta_{(i)}}{N}\pi_{(j),t}(R_{(j),t}-1)}\big\vert \mathcal{F}_{t-1}\vee \sigma(\xi_t^{cn})\right]\bigg\vert \mathcal{F}_{t-1}\bigg],
\end{align*}
where
\begin{align*}
&\E\left[-e^{-\gamma_{(i)}(1-\frac{\theta_{(i)}}{N})\pi_{(i),t}(R_{(i),t}-1)}\big\vert \mathcal{F}_{t-1}\vee \sigma(\xi_t^{cn})\right]\mathbbm{1}_{\{\xi_t^{cn}=k\}}\\=&-\left(p_{(i),t}^{k}e^{-\gamma_{(i)}(1-\frac{\theta_{(i)}}{N})\pi_{(i),t}(u_{(i),t}-1)}+(1-p_{(i),t}^{k})e^{-\gamma_{(i)}(1-\frac{\theta_{(i)}}{N})\pi_{(i),t}(d_{(i),t}-1)}\right)\mathbbm{1}_{\{\xi_t^{cn}=k\}}, \quad k=0,1,
\end{align*}
and
\begin{align*}
    &\E\left[e^{\gamma_{(i)}\frac{\theta_{(i)}}{N}\pi_{(j),t}(R_{(j),t}-1)}\big\vert \mathcal{F}_{t-1}\vee \sigma(\xi_t^{cn})\right]\mathbbm{1}_{\{\xi_t^{cn}=k\}}\\=&\left(p_{(j),t}^{k}e^{\gamma_{(i)}\frac{\theta_{(i)}}{N}\pi_{(j),t}(u_{(j),t}-1)}+(1-p_{(j),t}^{k})e^{\gamma_{(i)}\frac{\theta_{(i)}}{N}\pi_{(j),t}(d_{(j),t}-1)}\right)\mathbbm{1}_{\{\xi_t^{cn}=k\}}, \quad k=0,1.
\end{align*}
After introducing the simpler notation
\begin{align*}
    C_{jt}^{(i),1}&=p_{(j),t}^1e^{\gamma_{(i)}\frac{\theta_{(i)}}{N}\pi_{(j),t}(u_{(j),t}-1)}+(1-p_{(j),t}^1)e^{\gamma_{(i)}\frac{\theta_{(i)}}{N}\pi_{(j),t}(d_{(j),t}-1)},\\
    C_{jt}^{(i),0}&=p_{(j),t}^0e^{\gamma_{(i)}\frac{\theta_{(i)}}{N}\pi_{(j),t}(u_{(j),t}-1)}+(1-p_{(j),t}^0)e^{\gamma_{(i)}\frac{\theta_{(i)}}{N}\pi_{(j),t}(d_{(j),t}-1)},
\end{align*}
we then have
{{
\begin{align*}
    f(\pi_{(i),t})=&-p_t^{cn}\prod \limits_{j\neq i}C_{jt}^{(i),1}\left(p_{(i),t}^1e^{-\gamma_{(i)}(1-\frac{\theta_{(i)}}{N})\pi_{(i),t}(u_{(i),t}-1)}+(1-p_{(i),t}^1)e^{-\gamma_{(i)}(1-\frac{\theta_{(i)}}{N})\pi_{(i),t}(d_{(i),t}-1)}\right)
    \\& -(1-p_t^{cn})\prod \limits_{j\neq i}C_{jt}^{(i),0}\left(p_{(i),t}^0e^{-\gamma_{(i)}(1-\frac{\theta_{(i)}}{N})\pi_{(i),t}(u_{(i),t}-1)}+(1-p_{(i),t}^0)e^{-\gamma_{(i)}(1-\frac{\theta_{(i)}}{N})\pi_{(i),t}(d_{(i),t}-1)}\right).
\end{align*}}}
Differentiating over $\pi_{(i),t}$ yields
{{
\begin{align*}
    f^{'}(\pi_{(i),t})&=\gamma_{(i)}(1-\frac{\theta_{(i)}}{N})(u_{(i),t}-1)e^{-\gamma_{(i)}(1-\frac{\theta_{(i)}}{N})\pi_{(i),t}(u_{(i),t}-1)}\\&\qquad \times\left(p_t^{cn}\prod \limits_{j\neq i}C_{jt}^{(i),1}p_{(i),t}^1+(1-p_t^{cn})\prod \limits_{j\neq i}C_{jt}^{(i),0}p_{(i),t}^0\right)\\&\quad+\gamma_{(i)}(1-\frac{\theta_{(i)}}{N})(d_{(i),t}-1)e^{-\gamma_{(i)}(1-\frac{\theta_{(i)}}{N})\pi_{(i),t}(d_{(i),t}-1)}\\&\qquad \times\left(p_t^{cn}\prod \limits_{j\neq i}C_{jt}^{(i),1}(1-p_{(i),t}^1)+(1-p_t^{cn})\prod \limits_{j\neq i}C_{jt}^{(i),0}(1-p_{(i),t}^0)\right).
\end{align*}}}
By solving $f^{'}(\pi_{(i),t}^{*})=0$, we obtain the optimal strategy
\begin{align*}
    \pi_{(i),t}^{*}=\frac{1}{\gamma_{(i)} (1-\frac{\theta_{(i)}}{N})(u_{(i),t}-d_{(i),t})}\ln \frac{(1-q_{(i),t})\left(p_t^{cn}\prod \limits_{j\neq i}C_{jt}^{(i),1}p_{(i),t}^1+(1-p_t^{cn})\prod \limits_{j\neq i}C_{jt}^{(i),0}p_{(i),t}^0\right)}{q_{(i),t}\left(p_t^{cn}\prod \limits_{j\neq i}C_{jt}^{(i),1}(1-p_{(i),t}^1)+(1-p_t^{cn})\prod \limits_{j\neq i}C_{jt}^{(i),0}(1-p_{(i),t}^0)\right)}.
\end{align*}
Since $f^{''}(\pi_{(i),t})<0$, $\pi_{(i),t}^{*}$ is indeed the strategy that maximizes $f(\pi_{(i),t})$, and we have, by plugging $\pi_{(i),t}^{*}$ back in function $f(\pi_{(i),t})$,
{{
\begin{align*}
   f(\pi_{(i),t}^{*})=&-\left(p_{(i),t}^1p_t^{cn}\prod \limits_{j\neq i}C_{jt}^{(i),1}+p_{(i),t}^0(1-p_t^{cn})\prod \limits_{j\neq i}C_{jt}^{(i),0}\right)\\&\qquad \times \left(\frac{q_{(i),t}\left(p_t^{cn}\prod \limits_{j\neq i}C_{jt}^{(i),1}(1-p_{(i),t}^1)+(1-p_t^{cn})\prod \limits_{j\neq i}C_{jt}^{(i),0}(1-p_{(i),t}^0)\right)}{(1-q_{(i),t})\left(p_t^{cn}\prod \limits_{j\neq i}C_{jt}^{(i),1}p_{(i),t}^1+(1-p_t^{cn})\prod \limits_{j\neq i}C_{jt}^{(i),0}p_{(i),t}^0\right)}\right)^{1-q_{(i),t}}
    \\
    &-\left((1-p_{(i),t}^1)p_t^{cn}\prod \limits_{j\neq i}C_{jt}^{(i),1}+(1-p_{(i),t}^0)(1-p_t^{cn})\prod \limits_{j\neq i}C_{jt}^{(i),0}\right)\\&\qquad \times \left(\frac{q_{(i),t}\left(p_t^{cn}\prod \limits_{j\neq i}C_{jt}^{(i),1}(1-p_{(i),t}^1)+(1-p_t^{cn})\prod \limits_{j\neq i}C_{jt}^{(i),0}(1-p_{(i),t}^0)\right)}{(1-q_{(i),t})\left(p_t^{cn}\prod \limits_{j\neq i}C_{jt}^{(i),1}p_{(i),t}^1+(1-p_t^{cn})\prod \limits_{j\neq i}C_{jt}^{(i),0}p_{(i),t}^0\right)}\right)^{-q_{(i),t}}.
\end{align*}}}
It naturally follows that $f(\pi_{(i),t}^{*})=-\left(A_{(i),t}^{1}B_{(i),t}^{1}+A_{(i),t}^{2}B_{(i),t}^{2}\right)$ by letting $A_{(i),t}^{1}, A_{(i),t}^{2}$, $B_{(i),t}^{1}$ and $B_{(i),t}^{2}$ be given by (\ref{expressionA1}), (\ref{expressionA2}), (\ref{expressionB1}) and (\ref{expressionB2}).
\end{proof}
We continue with the proof of Theorem \ref{Thm:RelativeForwardNAgents}.

Conditions $(i)$ and $(ii)$ of Definition \ref{def:RelativeFrowardPreferences-NQ} follow directly.
Indeed, let $\widetilde{x} \in \mathbb{R}$, we first note that for any $t\in \mathbb{N}$, the sequences 
$(A_{(i),n}^{1})_{n=1}^t,(A_{(i),n}^{1})_{n=1}^t,(B_{(i),n}^{1})_{n=1}^t,(B_{(i),n}^{1})_{n=1}^t,(C_{jn}^{(i),0})_{n=1}^t$ and $(C_{jn}^{(i),1})_{n=1}^t$ are all Borel-measurable functions of the competitor policies $(\pi_{(j),n}^{*})_{n=1}^t$ and market parameters $(p_{(k),n}^{1})_{n=1}^t, (p_{(k),n}^{0})_{n=1}^t,(u_{(k),n})_{n=1}^t,(d_{(k),n})_{n=1}^t, k=1,\dots,N$ which are $\mathcal{F}_{t-1}$-measurable. 
It then follows that $U_{(i),t}(\widetilde{x})$ given by \eqref{U_tNAgent}) is $\mathcal{F}_{t-1}$-measurable. 

Next, we show Condition $(iii)$ of Definition \ref{def:RelativeFrowardPreferences-NQ}. 
Let $j \in \{ 1, \dots, N\}$, $t\in \mathbb{N}$, and 
 $X_{(j)}^{*}\in {\mathcal{X}}_{(j)}(x)$ with  $\pi_{(j)}^{*}$ for all $j\neq i$ be given as in the theorem. 
 Consider an arbitrary $X_{(i)}\in {\mathcal{X}}^{(i)}(\widetilde{x}_{(i)})$ and admissible strategy $\pi_{(i)}$.
 We must show that
\begin{align*}
    &U_{(i),t}(\widetilde{X}_{(i),t})\geq \E\left[-\frac{e^{-\gamma_{(i)}\widetilde{X}_{(i),t+1}}}{\prod \limits_{n=1}^{t+1}(A_{(i),n}^{1}B_{(i),n}^{1}+A_{(i),n}^{2}B_{(i),n}^{2})}\Bigg\vert \mathcal{F}_{t}\right],
\end{align*}
with now the $\pi_{(j),t}$ in computing $A_{(i),n}^{1}, A_{(i),n}^{2}, B_{(i),n}^{1}$ and $B_{(i),n}^{2}$ being replaced by $\pi_{(j),t}^{*}$.

Straightforward computation shows
\begin{align*}
    \E & \left[-\frac{e^{-\gamma_{(i)}\widetilde{X}_{(i),t+1}}}{\prod \limits_{n=1}^{t+1}(A_{(i),n}^{1}B_{(i),n}^{1}+A_{(i),n}^{2}B_{(i),n}^{2})}\Bigg\vert \mathcal{F}_{t}\right]\\
    & = \E\left[-\frac{e^{-\gamma_{(i)}(1-\frac{\theta_{(i)}}{N})(X_{(i),t}+\pi_{(i),t+1}(R_{(i),t+1}-1))-\frac{\theta_{(i)}}{N}\sum_{j\neq i}X_{(j),t}^{*}+\pi_{(j),t+1}^{*}(R_{(j),t+1}-1)}}{\prod \limits_{n=1}^{t+1}(A_{(i),n}^{1}B_{(i),n}^{1}+A_{(i),n}^{2}B_{(i),n}^{2})}\Bigg\vert \mathcal{F}_{t}\right]
    \\&=-e^{-\gamma_{(i)}((1-\frac{\theta_{(i)}}{N})X_{(i),t}-\frac{\theta_{(i)}}{N}\sum_{j\neq i}X_{(j),t})}\E\left[\frac{e^{-\gamma_{(i)}(1-\frac{\theta_{(i)}}{N})\pi_{(i),t+1}(R_{(i),t+1}-1)-\frac{\theta_{(i)}}{N}\sum_{j\neq i}\pi_{(j),t+1}^{*}(R_{(j),t+1}-1)}}{(A_{(i),t+1}^{1}B_{(i),t+1}^{1}+A_{(i),t+1}^{2}B_{(i),t+1}^{2})}\Bigg\vert \mathcal{F}_{t}\right]\\&=U_{(i),t}(\widetilde{X}_{(i),t})\E\left[\frac{e^{-\gamma_{(i)}(1-\frac{\theta_{(i)}}{N})\pi_{(i),t+1}(R_{(i),t+1}-1)-\frac{\theta_{(i)}}{N}\sum_{j\neq i}\pi_{(j),t+1}^{*}(R_{(j),t+1}-1)}}{(A_{(i),t+1}^{1}B_{(i),t+1}^{1}+A_{(i),t+1}^{2}B_{(i),t+1}^{2})}\Bigg\vert \mathcal{F}_{t}\right].
\end{align*}
Hence, it remains to show that
\begin{align*}
    1\leq \frac{1}{(A_{(i),t+1}^{1}B_{(i),t+1}^{1}+A_{(i),t+1}^{2}B_{(i),t+1}^{2})}\E\left[e^{-\gamma_{(i)}(1-\frac{\theta_{(i)}}{N})\pi_{(i),t+1}(R_{(i),t+1}-1)-\frac{\theta_{(i)}}{N}\sum_{j\neq i}\pi_{(j),t+1}^{*}(R_{(j),t+1}-1)}\bigg\vert \mathcal{F}_{t}\right],
\end{align*}
which is equivalent to 
\begin{align*}
    \E\left[-e^{-\gamma_{(i)}(1-\frac{\theta_{(i)}}{N})\pi_{(i),t+1}(R_{(i),t+1}-1)-\frac{\theta_{(i)}}{N}\sum_{j\neq i}\pi_{(j),t+1}^{*}(R_{(j),t+1}-1)}\bigg\vert \mathcal{F}_{t}\right]\leq-(A_{(i),t+1}^{1}B_{(i),t+1}^{1}+A_{(i),t+1}^{2}B_{(i),t+1}^{2}),
\end{align*}
The above easily follows using Lemma \ref{LemmaSup}.

To show Condition $(iv)$, we let $X_{(i),t}^{*}$, $t \in \mathbb{N}_0$, be given by (\ref{OptimalWealthN}). We need to establish
\begin{align*}
    -\frac{e^{-\gamma_{(i)}\widetilde{X}_{(i),t}^{*}}}{\prod \limits_{n=1}^t(A_{(i),n}^{1}B_{(i),n}^{1}+A_{(i),n}^{2}B_{(i),n}^{2})}= \E\left[-\frac{e^{-\gamma_{(i)}\widetilde{X}_{(i),t+1}^{*}}}{\prod \limits_{n=1}^{t+1}(A_{(i),n}^{1}B_{(i),n}^{1}+A_{(i),n}^{2}B_{(i),n}^{2})}\Bigg\vert \mathcal{F}_{t}\right].
\end{align*}
Since ${X}_{(i),t+1}^{*}={X}_{(i),t}^{*}+\pi_{(i),t+1}^{*}(R_{(i),t+1}-1)$ and thus
$\widetilde{X}_{(i),t+1}^{*}=\widetilde{X}_{(i),t}^{*}+(1-\frac{\theta_{(i)}}{N})\pi^{*}_{(i),t+1}(R_{(i),t+1}-1)-\frac{\theta_{(i)}}{N}\sum_{k\neq m}\pi^{*}_{(k),t+1}(R_{(k),t+1}-1)$, the above equality reduces to 
\begin{align*}
    1= \E\left[-\frac{e^{-\gamma_{(i)}\left((1-\frac{\theta_{(i)}}{N})\pi^{*}_{(i),t+1}(R_{(i),t+1}-1)-\frac{\theta_{(i)}}{N}\sum_{k\neq m}\pi^{*}_{(k),t+1}(R_{(k),t+1}-1)\right)}}{(A_{(i),t+1}^{1}B_{(i),t+1}^{1}+A_{(i),t+1}^{2}B_{(i),t+1}^{2})}\Bigg\vert \mathcal{F}_{t}\right].
\end{align*}
This then directly follows with Lemma \ref{LemmaSup}.
\qed

\subsection{Proof of Theorem \ref{Thm:HomoNAgentsGame}}
When $y_t=0$, we have
\begin{align*}
f(0)&=\frac{1-q_t}{q_t}\frac{p_t^{1}p_t^{cn}\left(\frac{1-p_t^{1}}{1-p_t^{0}}\right)^{N-1}+p_t^{0}(1-p_t^{cn})}{(1-p_t^{1})p_t^{cn}\left(\frac{1-p_t^{1}}{1-p_t^{0}}\right)^{N-1}+(1-p_t^{0})(1-p_t^{cn})}\\&=\frac{1-q_t}{q_t}\left(-1+\frac{p_t^{cn}\left(\frac{1-p_t^{1}}{1-p_t^{0}}\right)^{N-1}+(1-p_t^{cn})}{(1-p_t^{1})p_t^{cn}\left(\frac{1-p_t^{1}}{1-p_t^{0}}\right)^{N-1}+(1-p_t^{0})(1-p_t^{cn})}\right)>0,
\end{align*}
when $y_t\rightarrow +\infty$, we have
\begin{align*}
f(+\infty)=\frac{1-q_t}{q_t}\frac{p_t^{1}p_t^{cn}\left(\frac{p_t^{1}}{p_t^{0}}\right)^{N-1}+p_t^{0}(1-p_t^{cn})}{(1-p_t^{1})p_t^{cn}\left(\frac{p_t^{1}}{p_t^{0}}\right)^{N-1}+(1-p_t^{0})(1-p_t^{cn})}<+\infty.
\end{align*}

Next we compute the first and the second order derivatives with respect to $f(y_t)$, since $f(y_t)$ can be expressed by $f(y_t)=\frac{1-q_t}{q_t}g_1(g_2(y_t))$, where $g_1(x)=\frac{p_t^{1}p_t^{cn}x+p_t^{0}(1-p_t^{cn})}{(1-p_t^{1})p_t^{cn}x+(1-p_t^{0})(1-p_t^{cn})}$, and $g_2(x)=\left(\frac{p_t^{1}x^{\frac{\theta}{N-\theta}}+(1-p_t^{1})}{p_t^{0}x^{\frac{\theta}{N-\theta}}+(1-p_t^{0})}\right)^{N-1}$, then we have
\begin{align*}
    f^{'}(y_t)=\frac{1-q_t}{q_t}g_1^{'}(g_2(y_t))g_2^{'}(y_t),
\end{align*}
and
\begin{align*}
    f^{''}(y_t)=\frac{1-q_t}{q_t}\left(g_1^{''}(g_2(y_t))(g_2^{'}(y_t))^2+g_1^{'}(g_2(y_t)g_2^{''}(y_t))\right).
\end{align*}

It is easy to see that $g_1^{'}(x)=(1-p_t^{cn})p_t^{cn}\frac{p_t^{1}(1-p_t^{0})-p_t^{0}(1-p_t^{1})}{\left((1-p_t^{1})x+(1-p_t^{0})(1-p_t^{cn})\right)^2}>0$, $g_1^{''}(x)=-2(1-p_t^{1})(1-p_t^{cn})p_t^{cn}\frac{p_t^{1}(1-p_t^{0})-p_t^{0}(1-p_t^{1})}{\left((1-p_t^{1})x+(1-p_t^{0})(1-p_t^{cn})\right)^3}<0$, 
\begin{align*}
g_2^{'}(x)=(N-1)\frac{\theta}{N-\theta}(p_t^{1}-p_t^{0})x^{\frac{\theta}{N-\theta}-1}\frac{\left(p_t^{1}x^{\frac{\theta}{N-\theta}}+(1-p_t^{1})\right)^{N-2}}{\left(p_t^{0}x^{\frac{\theta}{N-\theta}}+(1-p_t^{0})\right)^N}>0,
\end{align*}
and 
\begin{align*}
g_2^{''}(x)=\frac{\theta(N-1)(p_t^{1}-p_t^{0})\left(p_t^{1}x^{\frac{\theta}{N-\theta}}+(1-p_t^{1})\right)^{N-3}x^{\frac{\theta}{N-\theta}-2}\left(A(x)+B(x)+C(x)\right)}{(N-\theta)\left(p_t^{0}x^{\frac{\theta}{N-\theta}}+(1-p_t^{0})\right)^{N+1}},
\end{align*}
where 
\begin{align*}
    A(x)=(N-2)\frac{\theta}{N-\theta}x^{\frac{\theta}{N-\theta}}p_t^{1}\left(p_t^{0}x^{\frac{\theta}{N-\theta}}+(1-p_t^{0})\right),
\end{align*}
\begin{align*}
    B(x)=\left(\frac{\theta}{N-\theta}-1\right)\left(p_t^{0}x^{\frac{\theta}{N-\theta}}+(1-p_t^{0})\right)\left(p_t^{1}x^{\frac{\theta}{N-\theta}}+(1-p_t^{1})\right),
\end{align*}
and
\begin{align*}
    C(x)=-\frac{N\theta}{N-\theta}x^{\frac{\theta}{N-\theta}}p_t^{0}\left(p_t^{1}x^{\frac{\theta}{N-\theta}}+(1-p_t^{1})\right).
\end{align*}

After some tedious computation, we finally have
\begin{align*}
    &A(x)+B(x)+C(x)\\=&p_t^{1}p_t^{0}x^{\frac{2\theta}{N-\theta}}\left(-\frac{\theta}{N-\theta}-1\right)+p_t^{1}x^{\frac{\theta}{N-\theta}}\left(\frac{\theta(N-1)}{N-\theta}-1\right)+p_t^{0}x^{\frac{\theta}{N-\theta}}(-\frac{\theta(N-1)}{N-\theta}-1)\\&+2p_t^{1}p_t^{0}x^{\frac{\theta}{N-\theta}}+(1-p_t^{1})(1-p_t^{0})(\frac{\theta}{N-\theta}-1).
\end{align*}
Since $(\frac{\theta(N-1)}{N-\theta}-1)\leq0$, it implies 
\begin{align*}
    p_t^{1}x^{\frac{\theta}{N-\theta}}(\frac{\theta(N-1)}{N-\theta}-1)\leq p_t^{0}x^{\frac{\theta}{N-\theta}}(\frac{\theta(N-1)}{N-\theta}-1),
\end{align*}
and furthermore
\begin{align*}
    &p_t^{1}x^{\frac{\theta}{N-\theta}}(\frac{\theta(N-1)}{N-\theta}-1)+p_t^{0}x^{\frac{\theta}{N-\theta}}(-\frac{\theta(N-1)}{N-\theta}-1)+2p_t^{u,}p_t^{0}x^{\frac{\theta}{N-\theta}}\\\leq& 2(p_t^{1}-1)p_t^{0}x^{\frac{\theta}{N-\theta}}<0,
\end{align*}
$A(x)+B(x)+C(x)$ must be negative as the first and the fifth term are negative. Overall, $g_2^{''}(x)<0$, which means $f^{'}(y_t)>0$, and $f^{''}(y_t)<0$. We also know from above computation that $f(0)>0$ and $f(+\infty)+\infty$, hence there must exist a positive point $y_t^*$ such that $f(y_t^*)=y_t^*$, and the solution to $f(y_t)=y_t$ is unique because $f^{''}(y_t)<0$, $f(y_t)$ cannot exceed $y_t$ anymore after the first intersection. 

Therefore, we have proved the existence and uniqueness of the solution to equation (\ref{HomoN}). Reverting to the original variable, we obtain the equilibrium strategy given by (\ref{HomoNStrategy}).
Furthermore, the expression of the PRFPP for each agent follows by applying Theorem \ref{Thm:RelativeForwardNAgents}.
\qed

\subsection{Proof of Lemma \ref{ExistenceoFor2agents}}
We see that (\ref{EqnOfUnknowny}) is equivalent to equation
\begin{align*}
    &y_t^{\frac{(2-\theta_{(1)})\gamma_{(1)}}{\theta_{(2)}\gamma_{(2)}}}\left(p_t^{0,1}\left(\frac{(1-q_{(2),t})(p_t^{1,1}y_t+p_t^{0,1})}{q_{(2),t}(p_t^{1,0}y_t+p_t^{0,0})}\right)^{\frac{\gamma_{(1)}\theta_{(1)}}{\gamma_{(2)}(2-\theta_{(2)})}}+p_t^{0,0}\right)\\=&\frac{(1-q_{(1),t})}{q_{(1),t}}\times\left(p_t^{1,1}\left(\frac{(1-q_{(2),t})y_t+p_t^{0,1})}{q_{(2),t}(p_t^{1,0}y_t+p_t^{0,0})}\right)^{\frac{\gamma_{(1)}\theta_{(1)}}{\gamma_{(2)}(2-\theta_{(2)})}}+p_t^{1,0}\right),
\end{align*}
which can be transformed to
\begin{align*}
    & \left(\frac{(1-q_{(2),t})(p_t^{1,1}y_t+p_t^{0,1})}{q_{(2),t}(p_t^{1,0}y_t+p_t^{0,0})}\right)^{\frac{\gamma_{(1)}\theta_{(1)}}{\gamma_{(2)}(2-\theta_{(2)})}}\left(p_t^{0,1}y_t^{\frac{(2-\theta_{(1)})\gamma_{(1)}}{\theta_{(2)}\gamma_{(2)}}}-\frac{p_t^{1,1}(1-q_{(1),t})}{q_{(1),t}}\right)\\
    & \qquad =-p_t^{0,0}y_t^{\frac{(2-\theta_{(1)})\gamma_{(1)}}{\theta_{(2)}\gamma_{(2)}}}+\frac{p_t^{1,0}(1-q_{(1),t})}{q_{(1),t}}.
\end{align*}
Thus, \eqref{EqnOfUnknowny} is equivalent to
\begin{align*}
    \left(\frac{(1-q_{(2),t})(p_t^{1,1}y_t+p_t^{0,1})}{q_{(2),t}(p_t^{1,0}y_t+p_t^{0,0})}\right)^{\frac{\gamma_{(1)}\theta_{(1)}}{\gamma_{(2)}(2-\theta_{(2)})}}=\frac{\frac{p_t^{1,0}(1-q_{(1),t})}{q_{(1),t}}-p_t^{0,0}y_t^{\frac{(2-\theta_{(1)})\gamma_{(1)}}{\theta_{(2)}\gamma_{(2)}}}}{-\frac{p_t^{1,1}(1-q_{(1),t})}{q_{(1),t}}+p_t^{0,1}y_t^{\frac{(2-\theta_{(1)})\gamma_{(1)}}{\theta_{(2)}\gamma_{(2)}}}}.
\end{align*}

Define the left-hand side and right-hand side by 
\begin{align*}
    L(y_t)=\left(\frac{(1-q_{(2),t})(p_t^{1,1}y_t+p_t^{0,1})}{q_{(2),t}(p_t^{1,0}y_t+p_t^{0,0})}\right)^{\frac{\gamma_{(1)}\theta_{(1)}}{\gamma_{(2)}(2-\theta_{(2)})}} \quad \mathrm{and} \quad R(y_t)=\frac{\frac{p_t^{1,0}(1-q_{(1),t})}{q_{(1),t}}-p_t^{0,0}y_t^{\frac{(2-\theta_{(1)})\gamma_{(1)}}{\theta_{(2)}\gamma_{(2)}}}}{-\frac{p_t^{1,1}(1-q_{(1),t})}{q_{(1),t}}+p_t^{0,1}y_t^{\frac{(2-\theta_{(1)})\gamma_{(1)}}{\theta_{(2)}\gamma_{(2)}}}} .
\end{align*}
Since $0<\theta_{(1)},\theta_{(2)}<1$, the exponents $\frac{\gamma_{(1)}\theta_{(1)}}{\gamma_{(2)}(2-\theta_{(2)})}$ and $\frac{(2-\theta_{(1)})\gamma_{(1)}}{\theta_{(2)}\gamma_{(2)}}$ are strictly positive. 
Thus, $L(y_t) > 0$ for any $y_t>0$ almost surely, $R(y_t)$ converges to $-\frac{p_t^{1,0}}{p_t^{1,1}}$ when $y_t\rightarrow0$, and $R(y_t)$ converges to $-\frac{p_t^{0,0}}{p_t^{0,1}}$ when $y_t\rightarrow +\infty$.
Their derivatives are
\begin{align*}
    L^{'}(y_t)=\frac{\gamma_{(1)}\theta_{(1)}}{\gamma_{(2)}(2-\theta_{(2)})}\left(\frac{(1-q_{(2),t})(p_t^{1,1}y_t+p_t^{0,1})}{q_{(2),t}(p_t^{1,0}y_t+p_t^{0,0})}\right)^{\frac{\gamma_{(1)}\theta_{(1)}}{\gamma_{(2)}(2-\theta_{(2)})}-1}\frac{(1-q_{(2),t})(p_t^{1,1}p_t^{0,0}-p_t^{0,1}p_t^{1,0})}{q_{(2),t}(p_t^{1,0}y_t+p_t^{0,0})^2},
\end{align*}
and 
\begin{align*}
    R^{'}(y_t)=\frac{\left(p_t^{1,1}p_t^{0,0}-p_t^{0,1}p_t^{1,0}\right)\frac{(2-\theta_{(1)})\gamma_{(1)}}{\theta_{(2)}\gamma_{(2)}}y_t^{\frac{(2-\theta_{(1)})\gamma_{(1)}}{\theta_{(2)}\gamma_{(2)}}-1}\frac{1-q_{(1),t}}{q_{(1),t}}}{\left(-\frac{p_t^{1,1}(1-q_{(1),t})}{q_{(1),t}}+p_t^{0,1}y_t^{\frac{(2-\theta_{(1)})\gamma_{(1)}}{\theta_{(2)}\gamma_{(2)}}}\right)^2}.
\end{align*}
Furthermore, 
\begin{align*}
p_t^{1,1}p_t^{0,0}-p_t^{0,1}p_t^{1,0}=p_t^{cn}(1-p_t^{cn}&)\left(p_{(2),t}^1(1-p_{(2),t}^0)-p_{(2),t}^0(1-p_{(2),t}^1)\right)\\&\times \left(p_{(1),t}^1(1-p_{(1),t}^0)-p_{(1),t}^0(1-p_{(1),t}^1)\right).
\end{align*}
Notice that $p_{(1),t}^1(1-p_{(1),t}^0)-p_{(1),t}^0(1-p_{(1),t}^1)>0, p_{(2),t}^1(1-p_{(2),t}^0)-p_{(2),t}^0(1-p_{(2),t}^1)>0$. Hence, $p_t^{1,1}p_t^{0,0}-p_t^{0,1}p_t^{1,0}>0$, both $L(y_t)$ and $R(y_t)$ are increasing.

However, there exists a singular point for $R(y_t)$ when the denominator equals zero almost surely, i.e., $-\frac{p_t^{1,1}(1-q_{(1),t})}{q_{(1),t}}+p_t^{0,1}y_t^{\frac{(2-\theta_{(1)})\gamma_{(1)}}{\theta_{(2)}\gamma_{(2)}}}=0$, and $y_t=\left(\frac{p_t^{1,1}(1-q_{(1),t})}{p_t^{0,1}q_{(1),t}}\right)^{\frac{\theta_{(2)}\gamma_{(2)}}{(2-\theta_{(1)})\gamma_{(1)}}}$. Thus the value of $R(y_t)$ tends to $+\infty$ from the left of this singular point and tends to $-\infty$ from the right of this singular point.
Hence, there must exist a point $y_t^{*}>0$ where $L(y_t)=R(y_t)$, and $y_t^*<\left(\frac{p_t^{1,1}(1-q_{(1),t})}{p_t^{0,1}q_{(1),t}}\right)^{\frac{\theta_{(2)}\gamma_{(2)}}{(2-\theta_{(1)})\gamma_{(1)}}}$ almost surely.
\qed



As a complement to the proof, we present the graph below as an example which helps to visualize the relationship between $L(y)$ and $R(y), y>0$.
\begin{figure}[H]
\centering
   \includegraphics[width=9.6cm]{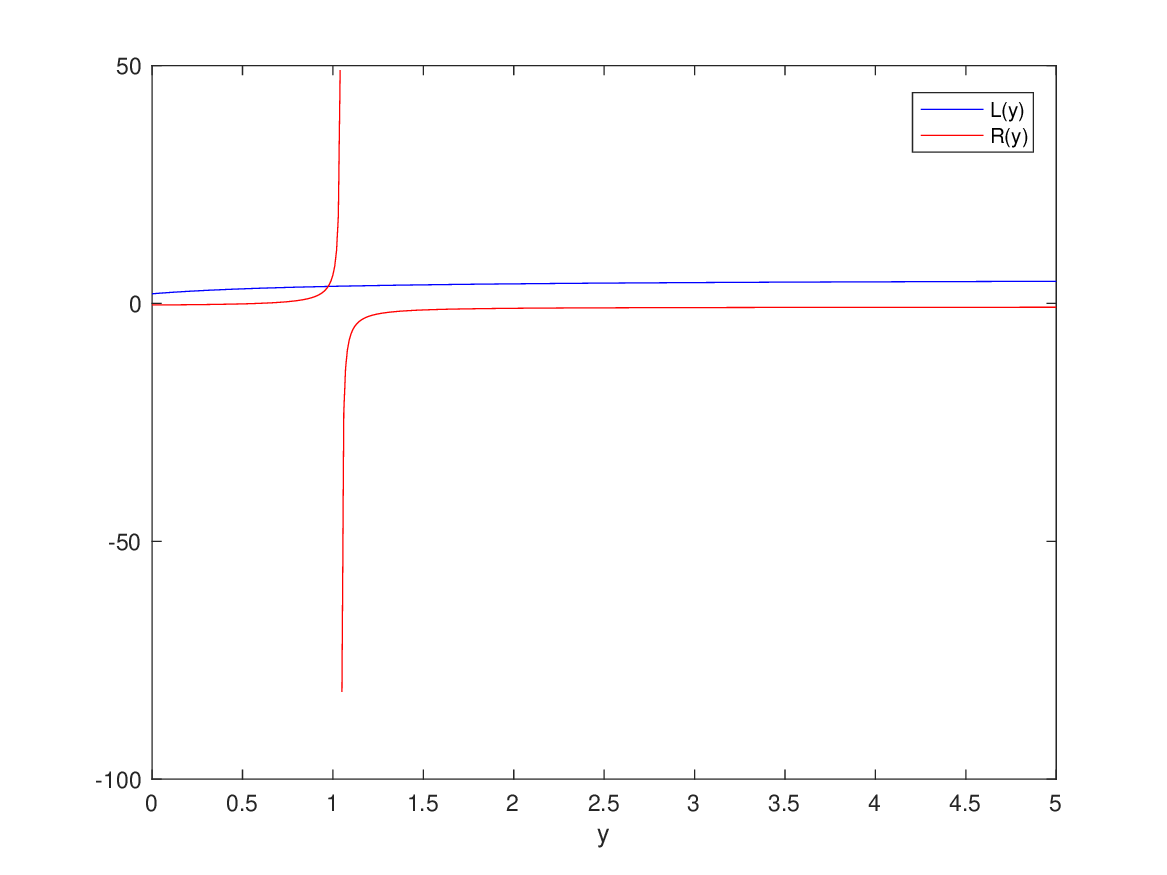}
  \caption{The comparison between $L(y)$ and $R(y)$}
\end{figure}

\subsection{Proof of Theorem \ref{Thm:2AgentsGame}}

From Theorem \ref{Thm:RelativeForwardNAgents}, the system of two equations (\ref{OptimalStrategyN1}) which needs to be solved for the $2$-agent game over trading period $[t-1, t)$ is 
\begin{align*}
    \pi_{(1),t}^{*}&=\frac{1}{\gamma_{(1)} (1-\frac{\theta_{(1)}}{2})(u_{(1),t}-d_{(1),t})}\ln \frac{(1-q_{(1),t})\left(p_t^{1,1}e^{\gamma_{(1)} \frac{\theta_{(1)}}{2}\pi_{(2),t}^{*}(u_{(2),t}-1)}+p_t^{1,0}e^{\gamma_{(1)} \frac{\theta_{(1)}}{2}\pi_{(2),t}^{*}(d_{(2),t}-1)}\right)}{q_{(1),t}\left(p_t^{0,1}e^{\gamma_{(1)} \frac{\theta_{(1)}}{2}\pi_{(2),t}^{*}(u_{(2),t}-1)}+p_t^{0,0}e^{\gamma_{(1)} \frac{\theta_{(1)}}{2}\pi_{(2),t}^{*}(d_{(2),t}-1)}\right)},
    \\\pi_{(2),t}^{*}&=\frac{1}{\gamma_{(2)} (1-\frac{\theta_{(2)}}{2})(u_{(2),t}-d_{(2),t})}\ln \frac{(1-q_{(2),t})\left(p_t^{1,1}e^{\gamma_{(2)} \frac{\theta_{(2)}}{2}\pi_{(1),t}^{*}(u_{(1),t}-1)}+p_t^{0,1}e^{\gamma_{(2)} \frac{\theta_{(2)}}{2}\pi_{(1),t}^{*}(d_{(1),t}-1)}\right)}{q_{(2),t}\left(p_t^{1,0}e^{\gamma_{(2)} \frac{\theta_{(2)}}{2}\pi_{(1),t}^{*}(u_{(1),t}-1)}+p_t^{0,0}e^{\gamma_{(2)} \frac{\theta_{(2)}}{2}\pi_{(1),t}^{*}(d_{(1),t}-1)}\right)}.
\end{align*}

By solving the system of the above two equations of unknowns $\pi_{(1),t}^{*}$ and $\pi_{(2),t}^{*}$, and writing $A_1=\gamma_{(2)} \frac{\theta_{(2)}}{2}(u_{(1),t}-1)$, $A_2=\gamma_{(2)} \frac{\theta_{(2)}}{2}(d_{(1),t}-1)$, $B_1=\gamma_{(1)} \frac{\theta_{(1)}}{2}(u_{(2),t}-1), B_2=\gamma_{(1)} \frac{\theta_{(1)}}{2}(d_{(2),t}-1)$, $C_1=\gamma_{(2)} (1-\frac{\theta_{(2)}}{2})(u_{(2),t}-d_{(2),t})$, and $C_2=\gamma_{(1)} (1-\frac{\theta_{(1)}}{2})(u_{(1),t}-d_{(1),t})$ to simplify notation, we obtain the following equation with unknown $x=\pi_{(1),t}^{*}$,
\begin{align*}
    x=\frac{1}{C_2}\ln \frac{(1-q_{(1),t})\left(p_t^{1,1}e^{\frac{B_1}{C_1}\ln\frac{(1-q_{(2),t})(p_t^{1,1}e^{A_1x}+p_t^{0,1}e^{A_2x})}{q_{(2),t}(p_t^{1,0}e^{A_1x}+p_t^{0,0}e^{A_2x})}}+p_t^{1,0}e^{\frac{B_2}{C_1}\ln\frac{(1-q_{(2),t})(p_t^{1,1}e^{A_1x}+p_t^{0,1}e^{A_2x})}{q_{(2),t}(p_t^{1,0}e^{A_1x}+p_t^{0,0}e^{A_2x})}}\right)}{q_{(1),t}\left(p_t^{0,1}e^{\frac{B_1}{C_1}\ln\frac{(1-q_{(2),t})(p_t^{1,1}e^{A_1x}+p_t^{0,1}e^{A_2x})}{q_{(2),t}(p_t^{1,0}e^{A_1x}+p_t^{0,0}e^{A_2x})}}+p_t^{0,0}e^{\frac{B_2}{C_1}\ln\frac{(1-q_{(2),t})(p_t^{1,1}e^{A_1x}+p_t^{0,1}e^{A_2x})}{q_{(2),t}(p_t^{1,0}e^{A_1x}+p_t^{0,0}e^{A_2x})}}\right)}.
\end{align*}
This is equivalent to 
\begin{align*}
    e^x=\left(\frac{(1-q_{(1),t})\left(p_t^{1,1}e^{\frac{B_1-B_2}{C_1}\ln\frac{(1-q_{(2),t})(p_t^{1,1}e^{(A_1-A_2)x}+p_t^{0,1})}{q_{(2),t}(p_t^{1,0}e^{(A_1-A_2)x}+p_t^{0,0})}}+p_t^{1,0}\right)}{q_{(1),t}\left(p_t^{0,1}e^{\frac{B_1-B_2}{C_1}\ln\frac{(1-q_{(2),t})(p_t^{1,1}e^{(A_1-A_2)x}+p_t^{0,1})}{q_{(2),t}(p_t^{1,0}e^{(A_1-A_2)x}+p_t^{0,0})}}+p_t^{0,0}\right)}\right)^{\frac{1}{C_2}}.
\end{align*}
After a further change of variable $y=e^{(A_1-A_2)x}$, we have
\begin{align*}
    y^{\frac{C_2}{A_1-A_2}}=\frac{(1-q_{(1),t})\left(p_t^{1,1}e^{\frac{B_1-B_2}{C_1}\ln\frac{(1-q_{(2),t})(p_t^{1,1}y+p_t^{0,1})}{q_{(2),t}(p_t^{1,0}y+p_t^{0,0})}}+p_t^{1,0}\right)}{q_{(1),t}\left(p_t^{0,1}e^{\frac{B_1-B_2}{C_1}\ln\frac{(1-q_{(2),t})(p_t^{1,1}y+p_t^{0,1})}{q_{(2),t}(p_t^{1,0}y+p_t^{0,0})}}+p_t^{0,0}\right)}.
\end{align*}
We note that the equation above is exactly the same as equation (\ref{EqnOfUnknowny}).
Furthermore, by Lemma \ref{ExistenceoFor2agents}, equation (\ref{EqnOfUnknowny}) has a positive solution $y_t^{*}$.
We can then express $\pi_{(1),t}^{*}$ in terms of $y_t^{*}$ with the relationship given in (\ref{OptimalStrategyj}).
By this representation, the equilibrium strategy of agent 2 can be accordingly simplified to (\ref{OptimalStrategyi}).
\qed

\subsection{Proof of Lemma \ref{Lemma:fixedpoint}}
For any fixed $\omega \in \Omega$, we have after taking derivative against $y\in \mathbb{R}$ that
\begin{align*}
    F^{'}_{y}(t,y,\omega)&=\E\left[\frac{({\Delta}_t^1-{\Delta}_t^0)}{\gamma(u_t-d_t)}\left(\frac{\frac{p_t^{cn}}{1-p_t^{cn}}e^{\gamma \theta y}+\frac{1-p_t^{0}}{1-p_t^{1}}}{\frac{p_t^{cn}}{1-p_t^{cn}}e^{\gamma \theta y}+\frac{p_t^{0}}{p_t^{1}}}\right)
    \frac{-\frac{p_t^{cn}}{1-p_t^{cn}}\gamma \theta e^{\gamma \theta y}\left(\frac{p_t^{0}}{p_t^{1}}-\frac{1-p_t^{0}}{1-p_t^{1}}\right)}{\left(\frac{p_t^{cn}}{1-p_t^{cn}}e^{\gamma \theta y}+\frac{1-p_t^{0}}{1-p_t^{1}}\right)^2}\Bigg\vert\mathcal{F}_{t-1}^{CN}\right]
    \\&=\E\left[\frac{({\Delta}_t^1-{\Delta}_t^0)}{\gamma(u_t-d_t)}\frac{\frac{p_t^{cn}}{1-p_t^{cn}}\gamma \theta e^{\gamma \theta y}\left(\frac{1-p_t^{0}}{1-p_t^{1}}-\frac{p_t^{0}}{p_t^{1}}\right)}{\left(\frac{p_t^{cn}}{1-p_t^{cn}}e^{\gamma \theta y}+\frac{1-p_t^{0}}{1-p_t^{1}}\right)\left(\frac{p_t^{cn}}{1-p_t^{cn}}e^{\gamma \theta y}+\frac{p_t^{0}}{p_t^{1}}\right)}\Bigg\vert\mathcal{F}_{t-1}^{CN}\right].
\end{align*}
Under the assumption that the market performs better on $\{\xi_t^{cn}=1\}$ than on $\{\xi_t^{cn}=0\}$, i.e.,
$p_t^{0}<p_t^{1}$,
it follows that ${\Delta}_t^1>{\Delta}_t^0$ and $\frac{1-p_t^{0}}{1-p_t^{1}}-\frac{p_t^{0}}{p_t^{1}}>0$. We conclude that $F^{'}_{y}(t,y,\omega)>0$ because the term inside the expectation is strictly positive. 
Furthermore, $F^{'}_{y}(t,y,\omega)$ can be written as
\begin{align*}
    F^{'}_{y}(t,y,\omega)&=\E\left[\frac{\theta({\Delta}_t^1-{\Delta}_t^0)}{(u_t-d_t)}\frac{\left(\frac{1-p_t^{0}}{1-p_t^{1}}-\frac{p_t^{0}}{p_t^{1}}\right)}{\left(1+\frac{1-p_t^{cn}}{p_t^{cn}}e^{-\gamma \theta y}\frac{1-p_t^{0}}{1-p_t^{1}}\right)\left(\frac{p_t^{cn}}{1-p_t^{cn}}e^{\gamma \theta y}+\frac{p_t^{0}}{p_t^{1}}\right)}\Bigg\vert\mathcal{F}_{t-1}^{CN}\right]\\&=\E\left[\frac{\theta({\Delta}_t^1-{\Delta}_t^0)}{(u_t-d_t)}\frac{\left(\frac{1-p_t^{0}}{1-p_t^{1}}-\frac{p_t^{0}}{p_t^{1}}\right)}{\frac{p_t^{cn}}{1-p_t^{cn}}e^{\gamma \theta y}+\frac{p_t^{0}}{p_t^{1}}+\frac{1-p_t^{0}}{1-p_t^{1}}+\frac{1-p_t^{cn}}{p_t^{cn}}e^{-\gamma \theta y}\frac{1-p_t^{0}}{1-p_t^{1}}\frac{p_t^{0}}{p_t^{1}}}\Bigg\vert\mathcal{F}_{t-1}^{CN}\right].
\end{align*}
Since $\frac{1-p_t^{0}}{1-p_t^{1}}-\frac{p_t^{0}}{p_t^{1}}<\frac{1-p_t^{0}}{1-p_t^{1}}+\frac{p_t^{0}}{p_t^{1}}$, and all other terms in the denominator are positive, 
\begin{align*}
    \frac{\frac{1-p_t^{0}}{1-p_t^{1}}-\frac{p_t^{0}}{p_t^{1}}}{\frac{p_t^{cn}}{1-p_t^{cn}}e^{\gamma \theta y}+\frac{p_t^{0}}{p_t^{1}}+\frac{1-p_t^{0}}{1-p_t^{1}}+\frac{1-p_t^{cn}}{p_t^{cn}}e^{-\gamma \theta y}\frac{1-p_t^{0}}{1-p_t^{1}}\frac{p_t^{0}}{p_t^{1}}}<1 .
\end{align*}
In addition, $\theta \in[0,1],$ and  $\frac{{\Delta}_t^1-{\Delta}_t^0}{u_t-d_t}\in(0,1)$ because $d_t-1<{\Delta}_t^0<{\Delta}_t^1<u_t-1$, then $ F^{'}_{y}(t,y,\omega) < 1$ follows.

Therefore, with $t\in \mathbb{N}$ fixed, we are able to prove that $0<F^{'}_{y}(t,y,\omega)<1$ for any $y\in\mathbb{R}$, and since $F(t,+\infty)=\phi_1$, which is finite, the fixed point to equation (\ref{fixedpoint}) exists and is unique.

To ensure that the fixed point is positive we need to check the value of $F(t,0)$ given by
\begin{align*}
    F(t,0)&=\phi_1+\E\left[\frac{{\Delta}_t^1-{\Delta}_t^0}{\gamma(u_t-d_t)}\ln\left(\frac{\frac{p_t^{cn}}{1-p_t^{cn}}e^{\gamma \theta y}+\frac{p_t^{0}}{p_t^{1}}}{\frac{p_t^{cn}}{1-p_t^{cn}}+\frac{1-p_t^{0}}{1-p_t^{1}}}\right)\Bigg\vert\mathcal{F}_{t-1}^{CN}\right]\\&=\E\left[\frac{({\Delta}_t^1-{\Delta}_t^0)}{\gamma(u_t-d_t)}\ln\frac{1-q_t}{q_t}\frac{p_t^{cn}p_t^{1}+(1-p_t^{cn})p_t^{0}}{p_t^{cn}(1-p_t^{1})+(1-p_t^{cn})(1-p_t^{0})}\Bigg\vert\mathcal{F}_{t-1}^{CN}\right]\\&=\E\left[\frac{({\Delta}_t^1-{\Delta}_t^0)}{\gamma(u_t-d_t)}\ln\frac{1-q_t}{q_t}\frac{\P [B_t = 1, \xi_t^{cn}=1 \vert \mathcal{F}_{t-1}^{MF}]+\P [B_t = 1, \xi_t^{cn}=0 \vert \mathcal{F}_{t-1}^{MF}]}{\P [B_t = 0, \xi_t^{cn}=1 \vert \mathcal{F}_{t-1}^{MF}]+\P [B_t = 0, \xi_t^{cn}=0 \vert \mathcal{F}_{t-1}^{MF}]}\Bigg\vert\mathcal{F}_{t-1}^{CN}\right]\\&=\E\left[\frac{({\Delta}_t^1-{\Delta}_t^0)}{\gamma(u_t-d_t)}\ln\frac{1-q_t}{q_t}\frac{\P [B_t = 1\vert \mathcal{F}_{t-1}^{MF}]}{\P [B_t = 0\vert \mathcal{F}_{t-1}^{MF}]}\Bigg\vert\mathcal{F}_{t-1}^{CN}\right]>0.
\end{align*}
Indeed, because the expected excess return is positive and the bond offers zero interest rate,
we have $\E[R_t-1\vert \mathcal{F}_{t-1}^{MF}]>0$ and $$\P [B_t = 1\vert \mathcal{F}_{t-1}^{MF}](u_t-1)+\P [B_t = 0\vert \mathcal{F}_{t-1}^{MF}](d_t-1)>0.$$ 
Recalling that $\frac{1-q_t}{q_t}=\frac{u_t-1}{1-d_t}$, it then follows that
$$\ln\frac{1-q_t}{q_t}\frac{\P [B_t = 1\vert \mathcal{F}_{t-1}^{MF}]}{\P [B_t = 0\vert \mathcal{F}_{t-1}^{MF}]}>0.$$
The positiveness of the fixed point, $y^{*}_t(\omega)>0$ for almost all $\omega \in \Omega$ is then a direct consequence of $F(t,0)>0$.

By the Banach fixed-point theorem, the fixed point can be found by iteration: Start with an arbitrary element $y_t^0$, let $y_t^j(\omega)=F(t,y_t^{j-1}(\omega),\omega)$, $j=1,2,\dots$, and obtain $\lim\limits_{j\rightarrow \infty}y_t^j(\omega)=y_t^*(\omega)$.
Hence, $y_t^*$ is $\mathcal{F}_{t-1}^{CN}$-measurable as the limit of a sequence of $\mathcal{F}_{t-1}^{CN}$-measurable random variables.
\qed

\subsection{Proof of Lemma \ref{lemmaSupMFG}}

For the representative agent with type vector $\psi_t$ the optimization problem over period $[t-1,t)$, $t\in \mathbb{N}$, is given by
\begin{equation}\label{OptimiRepresen}
\begin{aligned}
\sup_{\pi_t}\E[U_{t}(\widetilde X_{t})\vert \mathcal{F}_{t-1}^{MF}] & =\sup_{\pi_t}\E[-e^{-\gamma (X_t-\theta \overline{X}_t)}\big\vert \mathcal{F}_{t-1}^{MF}]\\
& = \sup_{\pi_t}\E\left[\E[-e^{-\gamma (X_t-\theta \overline{X}_t)}\big\vert \mathcal{F}_{t-1}^{MF}\vee \sigma(\xi_t^{cn})]\big\vert \mathcal{F}_{t-1}^{MF}\right]
\\
& =\sup_{\pi_t}\E\left[\E[-e^{-\gamma X_t}\big\vert \mathcal{F}_{t-1}^{MF}\vee \sigma(\xi_t^{cn})]\E[e^{\gamma\theta \overline{X}_t}\big\vert \mathcal{F}_{t-1}^{MF}\vee \sigma(\xi_t^{cn})]\big\vert \mathcal{F}_{t-1}^{MF}\right]
\\ & = \sup_{\pi_t}\E\left[\E[-e^{-\gamma X_t}\big\vert \mathcal{F}_{t-1}^{MF}\vee \sigma(\xi_t^{cn})]\E[e^{\gamma\theta \overline{X}_t}\big\vert \mathcal{F}_{t-1}^{MF}\vee \sigma(\xi_t^{cn})]\mathbbm{1}_{\{\xi_t^{cn}=1\}}\big\vert \mathcal{F}_{t-1}^{MF}\right]\\&+\E\left[\E[-e^{-\gamma X_t}\big\vert \mathcal{F}_{t-1}^{MF}\vee \sigma(\xi_t^{cn})]\E[e^{\gamma\theta \overline{X}_t}\big\vert \mathcal{F}_{t-1}^{MF}\vee \sigma(\xi_t^{cn})]\mathbbm{1}_{\{\xi_t^{cn}=0\}}\big\vert \mathcal{F}_{t-1}^{MF}\right],
\end{aligned}
\end{equation}
where 
the 3rd line follows from the independence between the representative agent and the entire network conditionally on the non-traded stochastic factor.

Furthermore,
\begin{align*}
\E[-e^{-\gamma X_t}\big\vert \mathcal{F}_{t-1}^{MF}\vee \sigma(\xi_t^{cn})]\mathbbm{1}_{\{\xi_t^{cn}=1\}}=\left(-p_t^{1}e^{-\gamma[X_{t-1}+\pi_t(u_t-1)]}-(1-p_t^{1})e^{-\gamma[X_{t-1}+\pi_t(d_t-1)]}\right)\mathbbm{1}_{\{\xi_t^{cn}=1\}},
\end{align*}
and
\begin{align*}
\E[-e^{-\gamma X_t}\big\vert \mathcal{F}_{t-1}^{MF}\vee \sigma(\xi_t^{cn})]\mathbbm{1}_{\{\xi_t^{cn}=0\}}=\left(-p_t^{0}e^{-\gamma[X_{t-1}+\pi_t(u_t-1)]}-(1-p_t^{0})e^{-\gamma[X_{t-1}+\pi_t(d_t-1)]}\right)\mathbbm{1}_{\{\xi_t^{cn}=0\}}.
\end{align*}
By Proposition \ref{Prop:MFGAverageWealth}, we know that $\overline{X}_t\mathbbm{1}_{\{\xi_t^{cn}=1\}}$ converges to $\overline{X}_{t-1}\mathbbm{1}_{\{\xi_t^{cn}=1\}}+\E\left[\pi_t\Delta_t^{1}\vert \mathcal{F}_{t-1}^{CN}\right]\mathbbm{1}_{\{\xi_t^{cn}=1\}}$ and $\overline{X}_t\mathbbm{1}_{\{\xi_t^{cn}=0\}}$ converges to $\overline{X}_{t-1}\mathbbm{1}_{\{\xi_t^{cn}=0\}}+\E\left[\pi_t\Delta_t^{0}\vert \mathcal{F}_{t-1}^{CN}\right]\mathbbm{1}_{\{\xi_t^{cn}=0\}}$ respectively in probability conditional on $\mathcal{F}_{t-1}^{CN}$, and thus
\begin{align*}
\E[e^{\gamma \theta \overline{X}_t}\big\vert \mathcal{F}_{t-1}^{MF}\vee \sigma(\xi_t^{cn})]\mathbbm{1}_{\{\xi_t^{cn}=1\}}=e^{\gamma \theta \left(\overline{X}_{t-1}+\E\left[\pi_t\Delta_t^{1}\vert \mathcal{F}_{t-1}^{CN}\right]\right)}\mathbbm{1}_{\{\xi_t^{cn}=1\}},
\end{align*}
\begin{align*}
\E[e^{\gamma \theta \overline{X}_t}\big\vert \mathcal{F}_{t-1}^{MF}\vee \sigma(\xi_t^{cn})]\mathbbm{1}_{\{\xi_t^{cn}=0\}}=e^{\gamma \theta \left(\overline{X}_{t-1}+\E\left[\pi_t\Delta_t^{0}\vert \mathcal{F}_{t-1}^{CN}\right]\right)}\mathbbm{1}_{\{\xi_t^{cn}=0\}}.
\end{align*}
Therefore, (\ref{OptimiRepresen}) becomes
\begin{align*}
    e^{-\gamma \widetilde X_{t-1}}\sup_{\pi_t}&-p_t^{cn}e^{\gamma \theta\E\left[\pi_t\Delta_t^{1}\vert \mathcal{F}_{t-1}^{CN}\right]}\left(p_t^{1}e^{-\gamma\pi_t(u_t-1)}+(1-p_t^{1})e^{-\gamma\pi_t(d_t-1)}\right)\\&-(1-p_t^{cn})e^{\gamma \theta\E\left[\pi_t\Delta_t^{0}\vert \mathcal{F}_{t-1}^{CN}\right]}\left(p_t^{0}e^{-\gamma\pi_t(u_t-1)}+(1-p_t^{0})e^{-\gamma\pi_t(d_t-1)}\right).
\end{align*}

Let
\begin{align*}\label{fpit}
    f(\pi_t)=&-p_t^{cn}e^{\gamma \theta\E\left[\pi_t\Delta_t^{1}\vert \mathcal{F}_{t-1}^{CN}\right]}\left(p_t^{1}e^{-\gamma\pi_t(u_t-1)}+(1-p_t^{1})e^{-\gamma\pi_t(d_t-1)}\right)\\&-(1-p_t^{cn})e^{\gamma \theta\E\left[\pi_t\Delta_t^{0}\vert \mathcal{F}_{t-1}^{CN}\right]}\left(p_t^{0}e^{-\gamma\pi_t(u_t-1)}+(1-p_t^{0})e^{-\gamma\pi_t(d_t-1)}\right).
\end{align*} 
Differentiate $f(\pi_t)$ over $\pi_t$ and let $f^{'}(\pi_t)=0$, we have
\begin{align*}
    e^{-\gamma \pi_t(u_t-d_t)}=\frac{(1-d_t)}{(u_t-1)}\frac{p_t^{cn}e^{\gamma \theta\E\left[\pi_t\Delta_t^{1}\vert \mathcal{F}_{t-1}^{CN}\right]}(1-p_t^{1})+(1-p_t^{cn})e^{\gamma \theta\E\left[\pi_t\Delta_t^{0}\vert \mathcal{F}_{t-1}^{CN}\right]}(1-p_t^{0})}{p_t^{cn}e^{\gamma \theta\E\left[\pi_t\Delta_t^{1}\vert \mathcal{F}_{t-1}^{CN}\right]}p_t^{1}+(1-p_t^{cn})e^{\gamma \theta\E\left[\pi_t\Delta_t^{0}\vert \mathcal{F}_{t-1}^{CN}\right]}p_t^{0}},
\end{align*}
and the maximizer of $f(\pi_t)$ given by $$\pi_t^{*}=\frac{1}{\gamma(u_t-d_t)}\ln\left(\frac{1-q_t}{q_t}\frac{p_t^{cn}e^{\gamma \theta\E\left[\pi_t\Delta_t^{1}\vert \mathcal{F}_{t-1}^{CN}\right]}p_t^{1}+(1-p_t^{cn})e^{\gamma \theta\E\left[\pi_t\Delta_t^{0}\vert \mathcal{F}_{t-1}^{CN}\right]}p_t^{0}}{p_t^{cn}e^{\gamma \theta\E\left[\pi_t\Delta_t^{1}\vert \mathcal{F}_{t-1}^{CN}\right]}(1-p_t^{1})+(1-p_t^{cn})e^{\gamma \theta\E\left[\pi_t\Delta_t^{0}\vert \mathcal{F}_{t-1}^{CN}\right]}(1-p_t^{0})}\right).$$ Since $f^{''}(\pi_t)<0$, $\pi_t^{*}$ is indeed the strategy that maximises $f(\pi_t)$, 
and furthermore it can be simplified to
\begin{align*}
    \pi_t^{*}=\frac{1}{\gamma(u_t-d_t)}\left(\ln\frac{1-q_t}{q_t}+\ln\frac{p_t^{cn}e^{\gamma \theta\E\left[\pi_t({\Delta}_t^1-{\Delta}_t^0)\vert \mathcal{F}_{t-1}^{CN}\right]}p_t^{1}+(1-p_t^{cn})p_t^{0}}{p_t^{cn}e^{\gamma \theta\E\left[\pi_t({\Delta}_t^1-{\Delta}_t^0)\vert \mathcal{F}_{t-1}^{CN}\right]}(1-p_t^{1})+(1-p_t^{cn})(1-p_t^{0})}\right).
\end{align*}

Since for the candidate control $\pi$ to be a MFE, we need $\pi=\pi^{*}$,
then multiply both sides by $({\Delta}_t^1-{\Delta}_t^0)$ and take conditional expectation to find that $\E\left[\pi_t^{*}({\Delta}_t^1-{\Delta}_t^0)\vert \mathcal{F}_{t-1}^{CN}\right]$ must satisfy
\begin{align*}
    \E\left[\pi_t^{*}({\Delta}_t^1-{\Delta}_t^0)\vert \mathcal{F}_{t-1}^{CN}\right]=\phi_0+\E\Bigg[&\frac{({\Delta}_t^1-{\Delta}_t^0)}{\gamma(u_t-d_t)}\\&\times \ln\frac{p_t^{cn}e^{\gamma \theta\E\left[\pi_t^{*}({\Delta}_t^1-{\Delta}_t^0)\vert \mathcal{F}_{t-1}^{CN}\right]}p_t^{1}+(1-p_t^{cn})p_t^{0}}{p_t^{cn}e^{\gamma \theta\E\left[\pi_t^{*}({\Delta}_t^1-{\Delta}_t^0)\vert \mathcal{F}_{t-1}^{CN}\right]}(1-p_t^{1})+(1-p_t^{cn})(1-p_t^{0})}\Bigg\vert \mathcal{F}_{t-1}^{CN}\Bigg],
\end{align*}
where $\phi_0=\E\left[\frac{({\Delta}_t^1-{\Delta}_t^0)}{\gamma(u_t-d_t)}\ln\frac{(1-q_t)}{q_t}\big\vert\mathcal{F}_{t-1}^{CN}\right]$.

Since
\begin{align*}
    &\ln\frac{p_t^{cn}e^{\gamma \theta\E\left[\pi_t^{*}({\Delta}_t^1-{\Delta}_t^0)\vert \mathcal{F}_{t-1}^{CN}\right]}p_t^{1}+(1-p_t^{cn})p_t^{0}}{p_t^{cn}e^{\gamma \theta\E\left[\pi_t^{*}({\Delta}_t^1-{\Delta}_t^0)\vert \mathcal{F}_{t-1}^{CN}\right]}(1-p_t^{1})+(1-p_t^{cn})(1-p_t^{0})}\\=&\ln\left(\frac{p_t^{1}}{(1-p_t^{1})}+\frac{\frac{(1-p_t^{cn})p_t^{0}p_t^{cn}(1-p_t^{1})-p_t^{cn}p_t^{1}(1-p_t^{cn})(1-p_t^{0})}
    {(p_t^{cn})^2(1-p_t^{1})^2}}{e^{\gamma \theta\E\left[\pi_t^{*}({\Delta}_t^1-{\Delta}_t^0)\vert \mathcal{F}_{t-1}^{CN}\right]}+\frac{(1-p_t^{cn})(1-p_t^{0})}{p_t^{cn}(1-p_t^{1})}}\right)\\=&\ln\left(\frac{p_t^{1}}{1-p_t^{1}}\left(1+\frac{\frac{1-p_t^{cn}}{p_t^{cn}}\frac{p_t^{0}(1-p_t^{1})(p_t^{1})^{-1}-(1-p_t^{0})}{(1-p_t^{1})}}{e^{\gamma \theta\E\left[\pi_t^{*}({\Delta}_t^1-{\Delta}_t^0)\vert \mathcal{F}_{t-1}^{CN}\right]}+\frac{(1-p_t^{cn})(1-p_t^{0})}{p_t^{cn}(1-p_t^{1})}}\right)\right)\\=&\ln\left(\frac{p_t^{1}}{1-p_t^{1}}\left(1+\frac{(\frac{p_t^{0}}{p_t^{1}}-\frac{1-p_t^{0}}{1-p_t^{1}})}{\frac{p_t^{cn}}{1-p_t^{cn}}e^{\gamma \theta\E\left[\pi_t^{*}({\Delta}_t^1-{\Delta}_t^0)\vert \mathcal{F}_{t-1}^{CN}\right]}+\frac{1-p_t^{0}}{1-p_t^{1}}}\right)\right)\\=&\ln\frac{p_t^{1}}{1-p_t^{1}}+\ln\left(1+\frac{(\frac{p_t^{0}}{p_t^{1}}-\frac{1-p_t^{0}}{1-p_t^{1}})}{\frac{p_t^{cn}}{1-p_t^{cn}}e^{\gamma   \theta\E\left[\pi_t^{*}({\Delta}_t^1-{\Delta}_t^0)\vert \mathcal{F}_{t-1}^{CN}\right]}+\frac{1-p_t^{0}}{1-p_t^{1}}}\right),
\end{align*}
Let $y_t=\E\left[\pi_t^{*}({\Delta}_t^1-{\Delta}_t^0)\vert \mathcal{F}_{t-1}^{CN}\right]$, our problem then is to find the fixed point of the equation in terms of $y_t$ of the same form as equation (\ref{fixedpoint}).

By Lemma \ref{Lemma:fixedpoint}, there exists a unique positive solution $y_t^{*}(\omega)\in \mathbb{R}^{+}$ to the equation given any fixed $\omega \in \Omega$, the Lemma naturally follows with the maximum occurring at $\pi_t^{*}$ given by (\ref{OptimalStrategyMFG}).
\qed

\subsection{Proof of Theorem \ref{Thm:RelativeForwardMFG}}
Condition $(i)$ and $(ii)$ in Definition \ref{def:RelativeFrowardPreferences-NQ} follows directly.

Next, we show Condition $(iii)$ in Definition \ref{def:RelativeFrowardPreferences-NQ}, i.e. for $t\in \mathbb{N}$ and $\widetilde{x}\in \mathbb{R}$, given any
admissible strategy $\pi_{i}, i=1,\dots,t+1$ and the corresponding wealth process $X_t^{\pi}$,
\begin{align*}
    &U_t (\widetilde{X}_t^{\pi})=-\frac{e^{-\gamma\widetilde{X}_t^{\pi}}}{\prod \limits_{n=1}^tG_n(\pi_n)}\geq \E\left[-\frac{e^{-\gamma \widetilde{X}_{t+1} }}{\prod \limits_{n=1}^{t+1}G_n(\pi_n)}\Bigg\vert \mathcal{F}_{t}^{MF}\right].
\end{align*}
i.e., 
\begin{align*}
-e^{-\gamma\widetilde{X}_t^{\pi}}G_{t+1}(\pi_{t+1}) \geq \E\left[-e^{-\gamma \widetilde{X}_{t+1} }\Bigg\vert \mathcal{F}_{t}^{MF}\right],
\end{align*}
and we easily conclude using Lemma \ref{lemmaSupMFG}.

To show Condition $(iv)$, we let $X_t^{*},t=0,1,\dots$ be the corresponding optimal wealth process evolving according to $X_t^{\pi^{*}} = x + \sum\nolimits_{i=1}^t\pi^{*}_i(R_i-1)
$ by following optimal strategy given in (\ref{OptimalStrategyMFG}). We need to establish
\begin{align*}
    &U_t (\widetilde{X}_t^{*})=-\frac{e^{-\gamma\widetilde{X}_t^{*}}}{\prod \limits_{n=1}^tG_n(\pi_n^{*})} =
    \E\left[-\frac{e^{-\gamma \widetilde{X}_{t+1} }}{\prod \limits_{n=1}^{t+1}G_n(\pi_n^{*})}\Bigg\vert \mathcal{F}_{t}^{MF}\right].
\end{align*}
i.e., 
\begin{align*}
-e^{-\gamma\widetilde{X}_t^{\pi}}G_{t+1}(\pi_{t+1}^{*}) = \E\left[-e^{-\gamma \widetilde{X}_{t+1}^{*}}\Bigg\vert \mathcal{F}_{t}^{MF}\right],
\end{align*}
which directly follows with Lemma \ref{lemmaSupMFG}.
\qed

\subsection{Proof of Proposition \ref{Prop:Convergence}}
Since in the mean field setup of homogeneous population, we have
\begin{align*}
    \Delta_t^1-\Delta_t^0=p_t^{1}u_t+(1-p_t^{1})d_t-p_t^{0}u_t-(1-p_t^{0})d_t=(p_t^{1}-p_t^{0})(u_t-d_t),
\end{align*}
then the corresponding fixed point equation $y_t^{mf}=F(t,y_t^{mf})$, where $F(t,y_t^{mf})$ is given by $(\ref{fixedpoint})$, can be simplified to
\begin{equation}\label{SimplifiedMFGEqn}
    y_t^{mf}=\frac{p_t^{1}-p_t^{0}}{\gamma}\ln\left(\frac{(1-q_t)\left(p_t^{cn}e^{\gamma \theta y_t^{mf}}p_t^{1}+(1-p_t^{cn})p_t^{0}\right)}{q_t\left(p_t^{cn}e^{\gamma \theta y_t^{mf}}p_t^{1}+(1-p_t^{cn})p_t^{0}\right)}\right),
\end{equation}
and the optimal strategy $(\ref{OptimalStrategyMFG})$ can also be written by $\pi_t^*=\frac{y_t^{mf,*}}{(u_t-d_t)(p_t^{1}-p_t^{0})}$ in terms of the solution to equation $(\ref{SimplifiedMFGEqn})$.

Next we argue that $\left(\frac{p_t^{1}y_t^{\frac{\theta}{N-\theta}}+1-p_t^{1}}{p_t^{0}y_t^{\frac{\theta}{N-\theta}}+1-p_t^{0}}\right)^{N-1}\rightarrow y_t^{\theta(p_t^{1}-p_t^{0})}$ as $N\rightarrow \infty$ for any $y_t>0$. Indeed, taking logarithm and then dividing both sides by $(N-1)$ yields the limits
\begin{align*}
\lim \limits_{N\rightarrow +\infty} \log \left(\frac{p_t^{1}y_t^{\frac{\theta}{N-\theta}}+1-p_t^{1}}{p_t^{0}y_t^{\frac{\theta}{N-\theta}}+1-p_t^{0}}\right)= 0,
\end{align*}
and
\begin{align*}
\lim \limits_{N\rightarrow +\infty}\frac{\theta(p_t^{1}-p_t^{0})}{N-1}\log y_t=0.
\end{align*}

By knowing this relationship and making the substitution $y_t=e^{\frac{\gamma y_t^{mf}}{(p_t^{1}-p_t^{0})}}$, it naturally follows that equation (\ref{HomoN}) can be transformed to equation (\ref{SimplifiedMFGEqn}) as $N\rightarrow +\infty$, which implies the convergence in distribution between the solution to (\ref{HomoN}) and the solution to (\ref{SimplifiedMFGEqn}) as $N\rightarrow +\infty$.
The corresponding optimal strategies of the $N$-agent game then converges in distribution to the MFE. 
\qed

\subsection{Proof of Proposition \ref{Prop:Covariance}}
Consider stock returns $R_{(i),t}$ and $R_{(j),t}$ for fixed $i \neq j$.
We first compute 
\begin{align*}
    \E[R_{(i),t}R_{(j),t}\vert\mathcal{G}_{t-1}^{MF}]&=\E\left[\E[R_{(i),t}R_{(j),t}\vert\mathcal{G}_{t-1}^{MF}\vee \sigma(\xi_t^{cn})]\vert\mathcal{G}_{t-1}^{MF}\right]\\&=\E\left[\E[R_{(i),t}\vert\mathcal{G}_{t-1}^{MF}\vee \sigma(\xi_t^{cn})]\E[R_{(j),t}\vert\mathcal{G}_{t-1}^{MF}\vee \sigma(\xi_t^{cn})]\vert\mathcal{G}_{t-1}^{MF}\right]\\&=\E\left[\E[R_{(i),t}\vert\mathcal{G}_{t-1}^{MF}\vee \sigma(\xi_t^{cn})]\E[R_{(j),t}\vert\mathcal{G}_{t-1}^{MF}\vee \sigma(\xi_t^{cn})]\mathbbm{1}_{\{\xi_t^{cn}=1\}}\vert\mathcal{G}_{t-1}^{MF}\right]\\&\qquad+\E\left[\E[R_{(i),t}\vert\mathcal{G}_{t-1}^{MF}\vee \sigma(\xi_t^{cn})]\E[R_{(j),t}\vert\mathcal{G}_{t-1}^{MF}\vee \sigma(\xi_t^{cn})]\mathbbm{1}_{\{\xi_t^{cn}=0\}}\vert\mathcal{G}_{t-1}^{MF}\right],
\end{align*}
where the second equality holds because the $B_{(i),t}$ and $B_{(j),t}$ of different stocks are independent of each other conditional on $\mathcal{G}_{t-1}^{MF}\vee \sigma(\xi_t^{cn})$. 
Furthermore, for $k=i,j$, we have
\begin{align*}
\E[R_{(k),t}\vert\mathcal{G}_{t-1}^{MF}\vee \sigma(\xi_t^{cn})]\mathbbm{1}_{\{\xi_t^{cn}=1\}}=\left(p_{(k),t}^1u_t+(1-p_{(k),t}^1)d_t\right)\mathbbm{1}_{\{\xi_t^{cn}=1\}},\\
\E[R_{(k),t}\vert\mathcal{G}_{t-1}^{MF}\vee \sigma(\xi_t^{cn})]\mathbbm{1}_{\{\xi_t^{cn}=0\}}=\left(p_{(k),t}^0u_t+(1-p_{(k),t}^0)d_t\right)\mathbbm{1}_{\{\xi_t^{cn}=0\}}.
\end{align*}
It then follows after some manipulations that 
\begin{align*}
    \E[R_{(i),t}R_{(j),t}\vert\mathcal{G}_{t-1}^{MF}]&=p_t^{cn}\left(p_{(i),t}^1u_t+(1-p_{(i),t}^1)d_t\right)\left(p_{(j),t}^1u_t+(1-p_{(j),t}^1)d_t\right)\\&\qquad+(1-p_t^{cn})\left(p_{(i),t}^0u_t+(1-p_{(i),t}^0)d_t\right)\left(p_{(j),t}^0u_t+(1-p_{(j),t}^0)d_t\right)\\=&\left((1-p_t^{cn})p_{(i),t}^0p_{(j),t}^0+p_t^{cn}p_{(i),t}^1p_{(j),t}^1\right)(u_t-d_t)^2+2d_t(u_t-d_t)p_t+d_t^2.
\end{align*}
In addition, $\E[R_{(j),t}\vert\mathcal{G}_{t-1}^{MF}]=\E[R_{(i),t}\vert\mathcal{G}_{t-1}^{MF}]=u_tp_t+d_t(1-p_t)$, and thus
\begin{align*}
    \E[R_{(j),t}\vert\mathcal{G}_{t-1}^{MF}]\E[R_{(i),t}\vert\mathcal{G}_{t-1}^{MF}]=(u_t-d_t)^2(p_t)^2+2d_t(u_t-d_t)p_t+d_t^2.
\end{align*} 
Finally, we can derive the covariance between two stock returns $R_{(i),t}$ and $R_{(j),t}$, 
\begin{align*}
{\rm{cov}}_t^{\psi}&=\E[R_{(i),t}R_{(j),t}\vert\mathcal{G}_{t-1}^{MF}]-\E[R_{(i),t}\vert\mathcal{G}_{t-1}^{MF}]\E[R_{(j),t}\vert\mathcal{G}_{t-1}^{MF}]\\&=(u_t-d_t)^2\left((1-p_t^{cn})p_{(i),t}^0p_{(j),t}^0+p_t^{cn}p_{(i),t}^1p_{(j),t}^1-(p_t)^2\right).
\end{align*}
\qed

\bibliography{BibFile}

\end{document}